\documentclass[preprint]{elsarticle}

\usepackage[english]{babel}
\usepackage{amssymb}
\usepackage{amsmath}
\usepackage{amsfonts}
\usepackage{amsthm}
\usepackage{shuffle}
\usepackage{lmodern}

\usepackage[top=3.5cm,bottom=3.5cm,left=3.8cm,right=3.8cm]{geometry}

\usepackage[dvipsnames]{xcolor}
\usepackage[hyperindex=true,
  frenchlinks=true,
  colorlinks=true,
  citecolor=Mahogany,
  linkcolor=DarkOrchid,
  urlcolor=Plum,
  linktocpage,
  pagebackref=true]{hyperref}

\usepackage{dsfont}
\usepackage[basic]{complexity}
\usepackage{stmaryrd}
\usepackage{subfigure}
\usepackage{multirow}
\usepackage{enumitem}
\usepackage{multicol}
\usepackage{xspace}
\usepackage[colorinlistoftodos]{todonotes}

\usepackage{tikz}
\usetikzlibrary{shadows}
\usetikzlibrary{arrows}
\usetikzlibrary{snakes}
\usetikzlibrary{backgrounds}
\usetikzlibrary{positioning}
\usetikzlibrary{calc}
\usetikzlibrary{shapes}
\usetikzlibrary{matrix}
\usetikzlibrary{decorations.pathreplacing}
\usetikzlibrary{decorations.pathmorphing} 
\usetikzlibrary{fit}          
\usetikzlibrary{backgrounds}  
\usetikzlibrary{fadings}

\numberwithin{equation}{section}
\renewcommand{\leq}{\leqslant}
\renewcommand{\geq}{\geqslant}

\newtheorem{Theorem}{Theorem}[section]
\newtheorem{Proposition}[Theorem]{Proposition}
\newtheorem{Lemma}[Theorem]{Lemma}
\newtheorem{Claim}[Theorem]{\normalfont {\em Claim}}
\newtheorem{Definition}[Theorem]{Definition}
\newtheorem{Corollary}[Theorem]{Corollary}

\DeclareMathOperator{\QQ}{\mathbb{Q}}
\DeclareMathOperator{\NN}{\mathbb{N}}
\DeclareMathOperator{\STD}{\mathrm{s}}
\DeclareMathOperator{\SHUFFLE}{\bullet}
\DeclareMathOperator{\BINTOPERM}{\mathrm{btp}}
\DeclareMathOperator{\PERMTOBIN}{\mathrm{ptb}}
\DeclareMathOperator{\DMATCHING}{\mathcal{M}}

\newcommand{\OEIS}[1]{\href{http://oeis.org/#1}{{\bf #1}}}

\newcommand{\CrossingLL}{%
  \begin{tikzpicture}[yscale=0.07,xscale=0.07,inner sep=.8pt,
      node distance=.25cm,>=latex',line width=1pt,baseline={(0,0.04)}]
    \draw node [draw,circle,fill=black] (U1)               {};
    \draw node [draw,circle,fill=black] [right of=U1] (U2) {};
    \draw node [draw,circle,fill=black] [right of=U2] (U3) {};
    \draw node [draw,circle,fill=black] [right of=U3] (U4) {};
    \draw [<-] (U1.north) .. controls ($ (U1.north) + (0,6) $) and ($ (U3.north) + (0,6) $) .. (U3.north);
    \draw [<-] (U2.north) .. controls ($ (U2.north) + (0,6) $) and ($ (U4.north) + (0,6) $) .. (U4.north);
\end{tikzpicture}}

\newcommand{\CrossingLR}{%
  \begin{tikzpicture}[yscale=0.07,xscale=0.07,inner sep=.8pt,
      node distance=.25cm,>=latex',line width=1pt,baseline={(0,0.04)}]
    \draw node [draw,circle,fill=black] (U1)               {};
    \draw node [draw,circle,fill=black] [right of=U1] (U2) {};
    \draw node [draw,circle,fill=black] [right of=U2] (U3) {};
    \draw node [draw,circle,fill=black] [right of=U3] (U4) {};
    \draw [<-] (U1.north) .. controls ($ (U1.north) + (0,6) $) and ($ (U3.north) + (0,6) $) .. (U3.north);
    \draw [->] (U2.north) .. controls ($ (U2.north) + (0,6) $) and ($ (U4.north) + (0,6) $) .. (U4.north);
\end{tikzpicture}}

\newcommand{\CrossingRL}{%
  \begin{tikzpicture}[yscale=0.07,xscale=0.07,inner sep=.8pt,
      node distance=.25cm,>=latex',line width=1pt,baseline={(0,0.04)}]
    \draw node [draw,circle,fill=black] (U1)               {};
    \draw node [draw,circle,fill=black] [right of=U1] (U2) {};
    \draw node [draw,circle,fill=black] [right of=U2] (U3) {};
    \draw node [draw,circle,fill=black] [right of=U3] (U4) {};
    \draw [->] (U1.north) .. controls ($ (U1.north) + (0,6) $) and ($ (U3.north) + (0,6) $) .. (U3.north);
    \draw [<-] (U2.north) .. controls ($ (U2.north) + (0,6) $) and ($ (U4.north) + (0,6) $) .. (U4.north);
\end{tikzpicture}}

\newcommand{\CrossingRR}{%
  \begin{tikzpicture}[yscale=0.07,xscale=0.07,inner sep=.8pt,
      node distance=.25cm,>=latex',line width=1pt,baseline={(0,0.04)}]
    \draw node [draw,circle,fill=black] (U1)               {};
    \draw node [draw,circle,fill=black] [right of=U1] (U2) {};
    \draw node [draw,circle,fill=black] [right of=U2] (U3) {};
    \draw node [draw,circle,fill=black] [right of=U3] (U4) {};
    \draw [->] (U1.north) .. controls ($ (U1.north) + (0,6) $) and ($ (U3.north) + (0,6) $) .. (U3.north);
    \draw [->] (U2.north) .. controls ($ (U2.north) + (0,6) $) and ($ (U4.north) + (0,6) $) .. (U4.north);
\end{tikzpicture}}

\newcommand{\InclusionLL}{%
  \begin{tikzpicture}[yscale=0.07,xscale=0.07,inner sep=.8pt,
      node distance=.25cm,>=latex',line width=1pt,baseline={(0,0.04)}]
    \draw node [draw,circle,fill=black] (U1)               {};
    \draw node [draw,circle,fill=black] [right of=U1] (U2) {};
    \draw node [draw,circle,fill=black] [right of=U2] (U3) {};
    \draw node [draw,circle,fill=black] [right of=U3] (U4) {};
    \draw [<-] (U1.north) .. controls ($ (U1.north) + (0,6) $) and ($ (U4.north) + (0,6) $) .. (U4.north);
    \draw [<-] (U2.north) .. controls ($ (U2.north) + (0,5) $) and ($ (U3.north) + (0,5) $) .. (U3.north);
\end{tikzpicture}}

\newcommand{\InclusionLR}{%
  \begin{tikzpicture}[yscale=0.07,xscale=0.07,inner sep=.8pt,
      node distance=.25cm,>=latex',line width=1pt,baseline={(0,0.04)}]
    \draw node [draw,circle,fill=black] (U1)               {};
    \draw node [draw,circle,fill=black] [right of=U1] (U2) {};
    \draw node [draw,circle,fill=black] [right of=U2] (U3) {};
    \draw node [draw,circle,fill=black] [right of=U3] (U4) {};
    \draw [<-] (U1.north) .. controls ($ (U1.north) + (0,6) $) and ($ (U4.north) + (0,6) $) .. (U4.north);
    \draw [->] (U2.north) .. controls ($ (U2.north) + (0,5) $) and ($ (U3.north) + (0,5) $) .. (U3.north);
\end{tikzpicture}}

\newcommand{\InclusionRL}{%
  \begin{tikzpicture}[yscale=0.07,xscale=0.07,inner sep=.8pt,
      node distance=.25cm,>=latex',line width=1pt,baseline={(0,0.04)}]
    \draw node [draw,circle,fill=black] (U1)               {};
    \draw node [draw,circle,fill=black] [right of=U1] (U2) {};
    \draw node [draw,circle,fill=black] [right of=U2] (U3) {};
    \draw node [draw,circle,fill=black] [right of=U3] (U4) {};
    \draw [->] (U1.north) .. controls ($ (U1.north) + (0,6) $) and ($ (U4.north) + (0,6) $) .. (U4.north);
    \draw [<-] (U2.north) .. controls ($ (U2.north) + (0,5) $) and ($ (U3.north) + (0,5) $) .. (U3.north);
\end{tikzpicture}}

\newcommand{\InclusionRR}{%
  \begin{tikzpicture}[yscale=0.07,xscale=0.07,inner sep=.8pt,
      node distance=.25cm,>=latex',line width=1pt,baseline={(0,0.04)}]
    \draw node [draw,circle,fill=black] (U1)               {};
    \draw node [draw,circle,fill=black] [right of=U1] (U2) {};
    \draw node [draw,circle,fill=black] [right of=U2] (U3) {};
    \draw node [draw,circle,fill=black] [right of=U3] (U4) {};
    \draw [->] (U1.north) .. controls ($ (U1.north) + (0,6) $) and ($ (U4.north) + (0,6) $) .. (U4.north);
    \draw [->] (U2.north) .. controls ($ (U2.north) + (0,5) $) and ($ (U3.north) + (0,5) $) .. (U3.north);
\end{tikzpicture}}

\newcommand{\LabeledInclusionRL}[5]{%
  \begin{tikzpicture}[scale=#1,inner sep=.1pt,node distance=.25cm,
      >=latex',text height=1.5ex,text depth=.25ex,line width=1pt,baseline={(0,0.05)}]
    \draw node [font=\tiny] (U1)               {#2};
    \draw node [font=\tiny] [right of=U1] (U2) {#3};
    \draw node [font=\tiny] [right of=U2] (U3) {#4};
    \draw node [font=\tiny] [right of=U3] (U4) {#5};
    \draw [->] (U1.north) .. controls ($ (U1.north) + (0,2) $) and ($ (U4.north) + (0,2) $) .. (U4.north);
    \draw [<-] (U2.north) .. controls ($ (U2.north) + (0,1.5) $) and ($ (U3.north) + (0,1.5) $) .. (U3.north);
\end{tikzpicture}}

\newcommand{\ConfigurationInclusionLL}[5]{%
  \begin{tikzpicture}[scale=#1,inner sep=.1pt,node distance=.7cm,
      >=latex',text height=1.5ex,text depth=.25ex,line width=1pt,baseline={(0,0.05)}]
    \draw node [font=\tiny,draw,shape=circle,circular drop shadow,fill=black!10,inner sep=2pt] (U1)               {#2};
    \draw node [font=\tiny,draw,shape=circle,circular drop shadow,fill=black!10,inner sep=2pt] [right of=U1] (U2) {#3};
    \draw node [font=\tiny,draw,shape=circle,circular drop shadow,fill=black!10,inner sep=2pt] [right of=U2] (U3) {#4};
    \draw node [font=\tiny,draw,shape=circle,circular drop shadow,fill=black!10,inner sep=2pt] [right of=U3] (U4) {#5};
    \draw [<-] (U1.north) .. controls ($ (U1.north) + (0,2) $) and ($ (U4.north) + (0,2) $) .. (U4.north);
    \draw [<-] (U2.north) .. controls ($ (U2.north) + (0,1.5) $) and ($ (U3.north) + (0,1.5) $) .. (U3.north);
\end{tikzpicture}}

\newcommand{\ConfigurationInclusionLR}[5]{%
  \begin{tikzpicture}[scale=#1,inner sep=.1pt,node distance=.7cm,
      >=latex',text height=1.5ex,text depth=.25ex,line width=1pt,baseline={(0,0.05)}]
    \draw node [font=\tiny,draw,shape=circle,circular drop shadow,fill=black!10,inner sep=2pt] (U1)               {#2};
    \draw node [font=\tiny,draw,shape=circle,circular drop shadow,fill=black!10,inner sep=2pt] [right of=U1] (U2) {#3};
    \draw node [font=\tiny,draw,shape=circle,circular drop shadow,fill=black!10,inner sep=2pt] [right of=U2] (U3) {#4};
    \draw node [font=\tiny,draw,shape=circle,circular drop shadow,fill=black!10,inner sep=2pt] [right of=U3] (U4) {#5};
    \draw [<-] (U1.north) .. controls ($ (U1.north) + (0,2) $) and ($ (U4.north) + (0,2) $) .. (U4.north);
    \draw [->] (U2.north) .. controls ($ (U2.north) + (0,1.5) $) and ($ (U3.north) + (0,1.5) $) .. (U3.north);
\end{tikzpicture}}

\newcommand{\ConfigurationInclusionRL}[5]{%
  \begin{tikzpicture}[scale=#1,inner sep=.1pt,node distance=.7cm,
      >=latex',text height=1.5ex,text depth=.25ex,line width=1pt,baseline={(0,0.05)}]
    \draw node [font=\tiny,draw,shape=circle,circular drop shadow,fill=black!10,inner sep=2pt] (U1)               {#2};
    \draw node [font=\tiny,draw,shape=circle,circular drop shadow,fill=black!10,inner sep=2pt] [right of=U1] (U2) {#3};
    \draw node [font=\tiny,draw,shape=circle,circular drop shadow,fill=black!10,inner sep=2pt] [right of=U2] (U3) {#4};
    \draw node [font=\tiny,draw,shape=circle,circular drop shadow,fill=black!10,inner sep=2pt] [right of=U3] (U4) {#5};
    \draw [->] (U1.north) .. controls ($ (U1.north) + (0,2) $) and ($ (U4.north) + (0,2) $) .. (U4.north);
    \draw [<-] (U2.north) .. controls ($ (U2.north) + (0,1.5) $) and ($ (U3.north) + (0,1.5) $) .. (U3.north);
\end{tikzpicture}}

\newcommand{\ConfigurationInclusionRR}[5]{%
  \begin{tikzpicture}[scale=#1,inner sep=.1pt,node distance=.7cm,
      >=latex',text height=1.5ex,text depth=.25ex,line width=1pt,baseline={(0,0.05)}]
    \draw node [font=\tiny,draw,shape=circle,circular drop shadow,fill=black!10,inner sep=2pt] (U1)               {#2};
    \draw node [font=\tiny,draw,shape=circle,circular drop shadow,fill=black!10,inner sep=2pt] [right of=U1] (U2) {#3};
    \draw node [font=\tiny,draw,shape=circle,circular drop shadow,fill=black!10,inner sep=2pt] [right of=U2] (U3) {#4};
    \draw node [font=\tiny,draw,shape=circle,circular drop shadow,fill=black!10,inner sep=2pt] [right of=U3] (U4) {#5};
    \draw [->] (U1.north) .. controls ($ (U1.north) + (0,2) $) and ($ (U4.north) + (0,2) $) .. (U4.north);
    \draw [->] (U2.north) .. controls ($ (U2.north) + (0,1.5) $) and ($ (U3.north) + (0,1.5) $) .. (U3.north);
\end{tikzpicture}}

\newcommand{\LabeledCrossingLL}[5]{%
  \begin{tikzpicture}[scale=#1,inner sep=.1pt,node distance=.25cm,
      >=latex',text height=1.5ex,text depth=.25ex,line width=1pt,baseline={(0,0.05)}]
    \draw node [font=\tiny] (U1)               {#2};
    \draw node [font=\tiny] [right of=U1] (U2) {#3};
    \draw node [font=\tiny] [right of=U2] (U3) {#4};
    \draw node [font=\tiny] [right of=U3] (U4) {#5};
    \draw [<-] (U1.north) .. controls ($ (U1.north) + (0,6) $) and ($ (U3.north) + (0,6) $) .. (U3.north);
    \draw [<-] (U2.north) .. controls ($ (U2.north) + (0,6) $) and ($ (U4.north) + (0,6) $) .. (U4.north);
\end{tikzpicture}}

\newcommand{\LabeledCrossingRL}[5]{%
  \begin{tikzpicture}[scale=#1,inner sep=.1pt,node distance=.25cm,
      >=latex',text height=1.5ex,text depth=.25ex,line width=1pt,baseline={(0,0.05)}]
    \draw node [font=\tiny] (U1)               {#2};
    \draw node [font=\tiny] [right of=U1] (U2) {#3};
    \draw node [font=\tiny] [right of=U2] (U3) {#4};
    \draw node [font=\tiny] [right of=U3] (U4) {#5};
    \draw [->] (U1.north) .. controls ($ (U1.north) + (0,6) $) and ($ (U3.north) + (0,6) $) .. (U3.north);
    \draw [<-] (U2.north) .. controls ($ (U2.north) + (0,6) $) and ($ (U4.north) + (0,6) $) .. (U4.north);
\end{tikzpicture}}

\newcommand{\LabeledCrossingRR}[5]{%
  \begin{tikzpicture}[scale=#1,inner sep=.1pt,node distance=.25cm,
      >=latex',text height=1.5ex,text depth=.25ex,line width=1pt,baseline={(0,0.05)}]
    \draw node [font=\tiny] (U1)               {#2};
    \draw node [font=\tiny] [right of=U1] (U2) {#3};
    \draw node [font=\tiny] [right of=U2] (U3) {#4};
    \draw node [font=\tiny] [right of=U3] (U4) {#5};
    \draw [->] (U1.north) .. controls ($ (U1.north) + (0,6) $) and ($ (U3.north) + (0,6) $) .. (U3.north);
    \draw [->] (U2.north) .. controls ($ (U2.north) + (0,6) $) and ($ (U4.north) + (0,6) $) .. (U4.north);
\end{tikzpicture}}

\newcommand{\ConfigurationCrossingRL}[5]{%
  \begin{tikzpicture}[scale=#1,inner sep=.1pt,node distance=.7cm,
      >=latex',text height=1.5ex,text depth=.25ex,line width=1pt,baseline={(0,0.05)}]
    \draw node [font=\tiny,draw,shape=circle,circular drop shadow,fill=black!10,inner sep=2pt] (U1)               {#2};
    \draw node [font=\tiny,draw,shape=circle,circular drop shadow,fill=black!10,inner sep=2pt] [right of=U1] (U2) {#3};
    \draw node [font=\tiny,draw,shape=circle,circular drop shadow,fill=black!10,inner sep=2pt] [right of=U2] (U3) {#4};
    \draw node [font=\tiny,draw,shape=circle,circular drop shadow,fill=black!10,inner sep=2pt] [right of=U3] (U4) {#5};
    \draw [->] (U1.north) .. controls ($ (U1.north) + (0,2) $) and ($ (U3.north) + (0,2) $) .. (U3.north);
    \draw [<-] (U2.north) .. controls ($ (U2.north) + (0,2) $) and ($ (U4.north) + (0,2) $) .. (U4.north);
\end{tikzpicture}}

\newcommand{\ConfigurationCrossingLR}[5]{%
  \begin{tikzpicture}[scale=#1,inner sep=.1pt,node distance=.7cm,
      >=latex',text height=1.5ex,text depth=.25ex,line width=1pt,baseline={(0,0.05)}]
    \draw node [font=\tiny,draw,shape=circle,circular drop shadow,fill=black!10,inner sep=2pt] (U1)               {#2};
    \draw node [font=\tiny,draw,shape=circle,circular drop shadow,fill=black!10,inner sep=2pt] [right of=U1] (U2) {#3};
    \draw node [font=\tiny,draw,shape=circle,circular drop shadow,fill=black!10,inner sep=2pt] [right of=U2] (U3) {#4};
    \draw node [font=\tiny,draw,shape=circle,circular drop shadow,fill=black!10,inner sep=2pt] [right of=U3] (U4) {#5};
    \draw [<-] (U1.north) .. controls ($ (U1.north) + (0,2) $) and ($ (U3.north) + (0,2) $) .. (U3.north);
    \draw [->] (U2.north) .. controls ($ (U2.north) + (0,2) $) and ($ (U4.north) + (0,2) $) .. (U4.north);
\end{tikzpicture}}

\newcommand{\PrecedenceLL}{%
  \begin{tikzpicture}[yscale=0.07,xscale=0.07,inner sep=.8pt,
      node distance=.25cm,>=latex',line width=1pt,baseline={(0,0.04)}]
    \draw node [draw,circle,fill=black] (U1)               {};
    \draw node [draw,circle,fill=black] [right of=U1] (U2) {};
    \draw node [draw,circle,fill=black] [right of=U2] (U3) {};
    \draw node [draw,circle,fill=black] [right of=U3] (U4) {};
    \draw [<-] (U1.north) .. controls ($ (U1.north) + (0,6) $) and ($ (U2.north) + (0,6) $) .. (U2.north);
    \draw [<-] (U3.north) .. controls ($ (U3.north) + (0,6) $) and ($ (U4.north) + (0,6) $) .. (U4.north);
  \end{tikzpicture}
}

\newcommand{\PrecedenceLR}{%
  \begin{tikzpicture}[yscale=0.07,xscale=0.07,inner sep=.8pt,
      node distance=.25cm,>=latex',line width=1pt,baseline={(0,0.04)}]
    \draw node [draw,circle,fill=black] (U1)               {};
    \draw node [draw,circle,fill=black] [right of=U1] (U2) {};
    \draw node [draw,circle,fill=black] [right of=U2] (U3) {};
    \draw node [draw,circle,fill=black] [right of=U3] (U4) {};
    \draw [<-] (U1.north) .. controls ($ (U1.north) + (0,6) $) and ($ (U2.north) + (0,6) $) .. (U2.north);
    \draw [->] (U3.north) .. controls ($ (U3.north) + (0,6) $) and ($ (U4.north) + (0,6) $) .. (U4.north);
  \end{tikzpicture}
}

\newcommand{\PrecedenceRL}{%
  \begin{tikzpicture}[yscale=0.07,xscale=0.07,inner sep=.8pt,
      node distance=.25cm,>=latex',line width=1pt,baseline={(0,0.04)}]
    \draw node [draw,circle,fill=black] (U1)               {};
    \draw node [draw,circle,fill=black] [right of=U1] (U2) {};
    \draw node [draw,circle,fill=black] [right of=U2] (U3) {};
    \draw node [draw,circle,fill=black] [right of=U3] (U4) {};
    \draw [->] (U1.north) .. controls ($ (U1.north) + (0,6) $) and ($ (U2.north) + (0,6) $) .. (U2.north);
    \draw [<-] (U3.north) .. controls ($ (U3.north) + (0,6) $) and ($ (U4.north) + (0,6) $) .. (U4.north);
  \end{tikzpicture}
}

\newcommand{\PrecedenceRR}{%
  \begin{tikzpicture}[yscale=0.07,xscale=0.07,inner sep=.8pt,
      node distance=.25cm,>=latex',line width=1pt,baseline={(0,0.04)}]
    \draw node [draw,circle,fill=black] (U1)               {};
    \draw node [draw,circle,fill=black] [right of=U1] (U2) {};
    \draw node [draw,circle,fill=black] [right of=U2] (U3) {};
    \draw node [draw,circle,fill=black] [right of=U3] (U4) {};
    \draw [->] (U1.north) .. controls ($ (U1.north) + (0,6) $) and ($ (U2.north) + (0,6) $) .. (U2.north);
    \draw [->] (U3.north) .. controls ($ (U3.north) + (0,6) $) and ($ (U4.north) + (0,6) $) .. (U4.north);
  \end{tikzpicture}
}

\newcommand{\LabeledPrecedenceLL}[5]{%
  \begin{tikzpicture}[scale=#1,inner sep=.1pt,node distance=.25cm,
      >=latex',text height=1.5ex,text depth=.25ex,line width=1pt,baseline={(0,0.05)}]
    \draw node [font=\tiny] (U1)               {#2};
    \draw node [font=\tiny] [right of=U1] (U2) {#3};
    \draw node [font=\tiny] [right of=U2] (U3) {#4};
    \draw node [font=\tiny] [right of=U3] (U4) {#5};
    \draw [<-] (U1.north) .. controls ($ (U1.north) + (0,6) $) and ($ (U2.north) + (0,6) $) .. (U2.north);
    \draw [<-] (U3.north) .. controls ($ (U3.north) + (0,6) $) and ($ (U4.north) + (0,6) $) .. (U4.north);
\end{tikzpicture}}

\newcommand{\LabeledPrecedenceLR}[5]{%
  \begin{tikzpicture}[scale=#1,inner sep=.1pt,node distance=.25cm,
      >=latex',text height=1.5ex,text depth=.25ex,line width=1pt,baseline={(0,0.05)}]
    \draw node [font=\tiny] (U1)               {#2};
    \draw node [font=\tiny] [right of=U1] (U2) {#3};
    \draw node [font=\tiny] [right of=U2] (U3) {#4};
    \draw node [font=\tiny] [right of=U3] (U4) {#5};
    \draw [<-] (U1.north) .. controls ($ (U1.north) + (0,6) $) and ($ (U2.north) + (0,6) $) .. (U2.north);
    \draw [->] (U3.north) .. controls ($ (U3.north) + (0,6) $) and ($ (U4.north) + (0,6) $) .. (U4.north);
\end{tikzpicture}}

\newcommand{\LabeledPrecedenceRL}[5]{%
  \begin{tikzpicture}[scale=#1,inner sep=.1pt,node distance=.25cm,
      >=latex',text height=1.5ex,text depth=.25ex,line width=1pt,baseline={(0,0.05)}]
    \draw node [font=\tiny] (U1)               {#2};
    \draw node [font=\tiny] [right of=U1] (U2) {#3};
    \draw node [font=\tiny] [right of=U2] (U3) {#4};
    \draw node [font=\tiny] [right of=U3] (U4) {#5};
    \draw [->] (U1.north) .. controls ($ (U1.north) + (0,6) $) and ($ (U2.north) + (0,6) $) .. (U2.north);
    \draw [<-] (U3.north) .. controls ($ (U3.north) + (0,6) $) and ($ (U4.north) + (0,6) $) .. (U4.north);
\end{tikzpicture}}

\newcommand{\LabeledPrecedenceRR}[5]{%
  \begin{tikzpicture}[scale=#1,inner sep=.1pt,node distance=.25cm,
      >=latex',text height=1.5ex,text depth=.25ex,line width=1pt,baseline={(0,0.05)}]
    \draw node [font=\tiny] (U1)               {#2};
    \draw node [font=\tiny] [right of=U1] (U2) {#3};
    \draw node [font=\tiny] [right of=U2] (U3) {#4};
    \draw node [font=\tiny] [right of=U3] (U4) {#5};
    \draw [->] (U1.north) .. controls ($ (U1.north) + (0,6) $) and ($ (U2.north) + (0,6) $) .. (U2.north);
    \draw [->] (U3.north) .. controls ($ (U3.north) + (0,6) $) and ($ (U4.north) + (0,6) $) .. (U4.north);
\end{tikzpicture}}

\begin{document}


\title{%
  Algorithmic and algebraic aspects \\
  of unshuffling permutations
}%
\tnotetext[t]{%
  This paper is an extended version of~\cite{GV16}.
}

\author{Samuele Giraudo}
\ead{samuele.giraudo@u-pem.fr}

\author{St\'ephane Vialette}
\ead{stephane.vialette@u-pem.fr}

\address{%
  Universit\'e Paris-Est, LIGM (UMR $8049$), CNRS,
  ENPC, ESIEE Paris, UPEM, F-$77454$, Marne-la-Vall\'ee, France
}


\begin{abstract}
  A permutation is said to be a square if it can be obtained by
  shuffling two order-isomorphic patterns. The definition is intended
  to be the natural counterpart to the ordinary shuffle of words and
  languages. In this paper, we tackle the problem of recognizing square
  permutations from both the point of view of algebra and algorithms.
  On the one hand, we present some algebraic and combinatorial
  properties of the shuffle product of permutations. We follow an
  unusual line consisting in defining the shuffle of permutations by
  means of an unshuffling operator, known as a coproduct. This
  strategy allows to obtain easy proofs for algebraic and combinatorial
  properties of our shuffle product. We besides exhibit a bijection
  between square $(213,231)$-avoiding permutations and square binary
  words. On the other hand, by using a pattern avoidance criterion on
  directed perfect matchings, we prove that recognizing square
  permutations is \NP-complete.
\end{abstract}

\maketitle

\tableofcontents

\section*{Introduction} \label{section:Introduction}
The {\em shuffle product}, denoted by $\shuffle$, is a well-known
operation on words first defined by Eilenberg and
Mac Lane~\cite{Eilenberg:MacLane:1953}. Given three words $u$, $v_1$,
and $v_2$, $u$ is said to be a \emph{shuffle} of $v_1$ and $v_2$ if it
can be formed by interleaving the letters from $v_1$ and $v_2$ in a way
that maintains the left-to-right ordering of the letters from each word.
Besides purely combinatorial questions, the shuffle product of words
naturally leads to the following computational problems:
\begin{enumerate}[label={\it (\roman*)},fullwidth]
\item \label{item:problem_1}
  Given two words $v_1$ and $v_2$, compute the set $v_1 \shuffle v_2$.
\item \label{item:problem_2}
  Given three words $u$, $v_1$, and $v_2$, decide if $u$ is a shuffle
  of $v_1$ and $v_2$.
\item \label{item:problem_3}
  Given words $u$, $v_1$, \dots, $v_k$, decide if $u$ is in
  $v_1 \shuffle \dots \shuffle v_k$.
\item \label{item:problem_4}
  Given a word $u$, decide if there is a word $v$ such that $u$ is
  in $v \shuffle v$.
\end{enumerate}
Even if these problems seem similar, they radically differ in terms of
time complexity. Let us now review some facts about these. In what
follows, $n$ denotes the size of $u$ and $m_i$ denotes the size of each
$v_i$. A solution to Problem~\ref{item:problem_1} can be computed in
\begin{equation}
  O\left((m_1 + m_2) \; \binom{m_1 + m_2}{m_1}\right)
\end{equation}
time~\cite{Spehner:TCS:1986}. An improvement and a generalization
of~Problem~\ref{item:problem_1} has been proposed
in~\cite{Allauzen:IGM:2000}, where it is proven that given words
$v_1$, \dots, $v_k$, the iterated shuffle
$v_1 \shuffle \dots \shuffle v_k$ can be computed in
\begin{equation}
  O\left(\binom{m_1 + \dots + m_k}{m_1, \dots, m_k}\right)
\end{equation}
time. Problem~\ref{item:problem_2} is in \P; it is indeed a
classical textbook exercise to design an efficient dynamic programming
algorithm solving it. It can be tested in $O\left(n^2 / \log(n)\right)$
time~\cite{Leeuwen:Nivat:IPL:1982}. To the best of our knowledge, the
first $O(n^2)$ time algorithm for this problem appeared
in~\cite{Mansfield:DAM:1983}. This algorithm can easily be extended to
check in polynomial-time whether a word is in the shuffle of any
fixed number of given words. Nevertheless, Problem~\ref{item:problem_3}
is \NP-complete~\cite{Mansfield:DAM:1983,Warmuth:Haussler:JCSS:1984}. This
remains true even if the ground alphabet has size
$3$~\cite{Warmuth:Haussler:JCSS:1984}. Of particular interest,
it is shown in~\cite{Warmuth:Haussler:JCSS:1984} that
Problem~\ref{item:problem_3} remains \NP-complete even if all the words $v_i$,
$i \in [k]$, are identical, thereby proving that, for two words $u$ and
$v$, it is \NP-complete to decide whether or not $u$ is in the iterated shuffle
of $v$. Again, this remains true even if the ground alphabet has size $3$.
Let us now finally focus on Problem~\ref{item:problem_4}. It is shown
in~\cite{Buss:Soltys:2014,Rizzi:Vialette:CSR:2013} that it is \NP-complete to
decide if a word $u$ is a \emph{square} (w.r.t. the shuffle), that is,
a word $u$ with the property that there exists a word $v$ such that $u$
is a shuffle of $v$ with itself. Hence, Problem~\ref{item:problem_4}
is \NP-complete.

This paper is intended to study a natural generalization of $\shuffle$,
denoted by $\SHUFFLE$, as a shuffle of permutations. Roughly speaking,
given three permutations $\pi$, $\sigma_1$, and $\sigma_2$, $\pi$ is
said to be a {\em shuffle} of $\sigma_1$ and $\sigma_2$ if $\pi$ (viewed
as a word) is a shuffle of two words that are order-isomorphic to
$\sigma_1$ and $\sigma_2$. This shuffle product was first introduced by
Vargas~\cite{Vargas:2014} under the name of {\em supershuffle}. Our
intention in this paper is to study this shuffle product of permutations
$\SHUFFLE$ both from a combinatorial and from a computational point of
view by focusing on {\em square} permutations, that are permutations
$\pi$ being in the shuffle of a permutation $\sigma$ with itself. Many
other shuffle products on permutations appear in the literature. For
instance, in~\cite{DHT:IJAC:2002}, the authors define the {\em convolution
  product} and the {\em shifted shuffle product}. For this last product,
$\pi$ is a shuffle of $\sigma_1$ and $\sigma_2$ if $\pi$ is in the shuffle,
as words, of $\sigma_1$ and the word obtained by incrementing all the
letters of $\sigma_2$ by the size of $\sigma_1$. It is a simple exercise
to prove that, given three permutations $\pi$, $\sigma_1$, and $\sigma_2$,
deciding if $\pi$ is in the shifted shuffle of $\sigma_1$ and $\sigma_2$
is in~\P.

This paper is organized as follows. In
Section~\ref{section:Shuffle product on permutations}, we provide a
precise definition of $\SHUFFLE$.
We shall define $\SHUFFLE$ in terms of what we call the unshuffling operator $\Delta$.
The operator $\Delta$ is in fact a coproduct,
endowing the linear span of all permutations with a coalgebra structure
(see~\cite{Joni:Rota:1979} or~\cite{Grinberg:Reiner:2014} for the
definition of these algebraic structures). By duality, the unshuffling
operator $\Delta$ leads to the definition of our shuffle operation on
permutations. This approach has many advantages. First, some
combinatorial properties of $\SHUFFLE$ depend on properties of $\Delta$
and those properties are easier
to prove on the coproduct side. Second, this approach allows us
to obtain a clear description of the multiplicities of the
elements appearing in the shuffle of two permutations,
which are of interest in their own right
from a combinatorial point of view.
Section~\ref{section:Binary square words and permutations} is devoted to
showing that the problems related to the shuffle of words has links with
the shuffle of permutations. In particular, we show that binary words
that are square are in one-to-one correspondence with square
permutations avoiding some patterns
(Proposition~\ref{prop:bijection_binary_to_permutations_squares}). Next,
Section~\ref{section:Algebraic issues} presents some algebraic and
combinatorial properties of~$\SHUFFLE$. We show that $\SHUFFLE$ is
associative and commutative
(Proposition~\ref{prop:shuffle_associative_commutative}), and that if a
permutation is a square, its mirror, complement, and inverse are also
squares (Proposition~\ref{prop:square_stability}). Finally,
Section~\ref{section:Algorithmic issues} presents the most important
result of this paper: the fact that deciding if a permutation is a
square is \NP-complete (Proposition~\ref{proposition:hardness}). This result is
obtained by exhibiting a reduction from the \NP-complete pattern involvement
problem~\cite{Bose:Buss:Lubiw:1998}.


\section{Notations and basic definitions} \label{section:Notations}


\subsection*{General notations}

If $S$ is a finite set, the cardinality of $S$ is denoted by $|S|$,
and if $P$ and $Q$ are two disjoint sets, $P \sqcup Q$ denotes the
disjoint union of $P$ and $Q$. For any nonnegative integer $n$, $[n]$
is the set $\{1, \dots, n\}$.


\subsection*{Words and permutations}

We follow the usual terminology on words~\cite{ChoffrutKarhumaki1997}.
Let us recall here the most important ones. Let $u$ be a word. The
length of $u$ (also called {\em size}) is denoted by $|u|$. The
{\em empty word}, the only word of null length, is denoted by $\epsilon$.
We denote by $\widetilde{u}$ the {\em mirror image} of $u$, that is the
word $u_{|u|} u_{|u| - 1} \dots u_1$. If $P$ is a subset of $[|u|]$,
$u_{|P}$ is the subword of $u$ consisting in the letters of $u$ at the
positions specified by the elements of $P$. If $u$ is a word of integers
and $k$ is an integer, we denote by $u[k]$ the word obtained by
incrementing by $k$ all letters of $u$. The {\em shuffle} of two words
$u$ and $v$ is the set recursively defined by
\begin{equation}
  u \shuffle \epsilon = \{u\} = \epsilon \shuffle u
\end{equation}
and
\begin{equation}
  ua \shuffle vb = (u \shuffle vb)a \; \cup \; (ua \shuffle v)b,
\end{equation}
were $a$ and $b$ are letters. For instance,
\begin{equation}
  01 \shuffle 20 = \{0120, 0210, 0201, 2010, 2001\}.
\end{equation}
A word $u$ is a {\em square} if there exists a word $v$ such that $u$
belongs to $v \shuffle v$. For example, $202101$ is a square since this
word belongs to the set $201 \shuffle 201$.

We denote by $S_n$ the set of permutations of size $n$ and by $S$ the
set of all permutations. In this paper, permutations of a size $n$ are
specified by words of length $n$ on the alphabet $[n]$ and without
multiple occurrences of a letter, so that all above definitions about
words remain valid on permutations. The only difference lies on the
fact that we shall denote by $\pi(i)$ (instead of $\pi_i$) the $i$-th
letter of any permutation $\pi$. For any nonnegative integer $n$, we
write $\nearrow_{n}$ (resp. $\searrow_{n}$) for the permutation
$1 2 \dots n$ (resp. $n\,(n-1) \dots 1$). If $\pi$ is a permutation of
$S_n$, we denote by $\bar \pi$ the {\em complement} of $\pi$, that is
the permutation satisfying $\bar \pi(i) = n - \pi(i) + 1$ for all
$i \in [n]$. The {\em inverse} of $\pi$ is denoted by $\pi^{-1}$.

If $u$ is a word of integers where no letter occurs more than once, we define the
\emph{standarization} of $u$, $\STD(u)$, to be the unique permutation of the same size
as $u$ such that for all $i, j \in [|u|]$, $u_i < u_j$ if and only if
$\STD(u)(i) < \STD(u)(j)$.
For instance,
\begin{equation}
  \STD(814637) = 613425.
\end{equation}
In particular, the image of the map $\STD$ is the set $S$ of all
permutations. Two words $u$ and $v$ having the same standarization are
{\em order-isomorphic}.
If $\sigma$ is a permutation, we say that $\sigma$ \emph{occurs} in $\pi$
if there is a set of indices $P$ of $[|\pi|]$ such that $\sigma$ and
$\pi_{|P}$ are order isomorphic.
When $\sigma$ does not occur in $\pi$, $\pi$ is said to {\em avoid} $\sigma$.
The set of permutations of size
$n$ avoiding $\sigma$ is denoted by $S_n(\sigma)$. The
{\em  pattern involvement problem} consists, given two permutations
$\pi$ and $\sigma$, in deciding if $\sigma$ occurs in $\pi$.
This problem is known to be \NP-complete~\cite{Bose:Buss:Lubiw:1998}.


\subsection*{Directed perfect matchings}

A \emph{directed graph} is an ordered pair $\mathcal{G} = (V, A)$ where
$V$ is a set whose elements are called \emph{vertices} and $A$ is a set
of ordered pairs of vertices, called \emph{arcs} (from a \emph{source}
vertex to a \emph{sink} vertex). In this paper, we shall exclusively use
$V \subset \NN$. Notice that the aforementioned definition does not allow
a directed graph to have multiple arcs with same source and target nodes.
We shall not allow directed loops (that is, arcs that connect vertices
with themselves). Two arcs are \emph{independent} if they do not have a
common vertex.
An arc $(i, i')$ \emph{contains} an arc $(j, j')$ if
$\min(i, i') < \min(j, j') < \max(j, j') < \max(i, i')$. If no
arc of $\mathcal{G}$ contains an other arc, we say that $\mathcal{G}$
is \emph{containment-free}.
Two arcs $(i, i')$ and $(j, j')$ are \emph{crossing}
if $\min(i, i') < \min(j, j') < \max(i, i') < \max(j, j')$.
If no arcs of $\mathcal{G}$ are crossing, we say that $\mathcal{G}$
is \emph{crossing-free}.
A directed graph is a
\emph{directed matching} if all its arcs are independent. A directed
matching is \emph{perfect} if every vertex is either a source or a sink.

For any permutation $\pi$ of an even size $2n$, a {\em directed perfect
matching on $\pi$} is a pair $\DMATCHING = (\mathcal{G}, \pi)$ where
$\mathcal{G}$ is a directed perfect matching on the set $[2n]$ of
vertices (see Figure~\ref{fig:permutation with directed perfect
matching}).
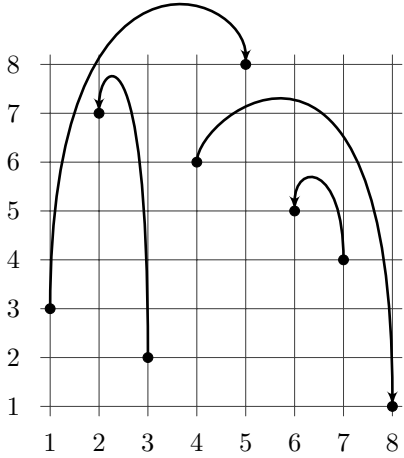
\begin{figure}[htbp!]
    \centering
    \begin{tikzpicture}[scale=.65,label/.style={anchor=base}]
        \draw[step=1cm,black!75,ultra thin,fill=black!50]
        (-0.2,-0.2) grid (7.2,7.2);
        \foreach \x/\y in {0/2,1/6,2/1,3/5,4/7,5/4,6/3,7/0} {
          \draw [fill=black] (\x,\y) circle (0.1);
        }
        \foreach \i in {1,...,8} {
          \node [align=center] at (\i-1,-0.75) {$\i$};
          \node [align=center] at (-0.75,\i-1) {$\i$};
        }
        \draw [black,line width=1pt,->,>=latex']
        (0,2) .. controls +(0,8) and +(0,1.5) .. (4,7);
        \draw [black,line width=1pt,<-,>=latex']
        (1,6) .. controls +(0,1) and +(0,7) .. (2,1);
        \draw [black,line width=1pt,->,>=latex']
        (3,5) .. controls +(0,1) and +(0,9) .. (7,0);
        \draw [black,line width=1pt,<-,>=latex']
        (5,4) .. controls +(0,1) and +(0,2) .. (6,3);
    \end{tikzpicture}
    \caption{%
    A directed perfect matching $\DMATCHING$ on the permutation
    \mbox{$\pi = 37268541$}, represented on the permutation matrix of
    $\pi$. The set of vertices of $\DMATCHING$ is $\{1, \dots, 8\}$ and
    the set of arcs of $\DMATCHING$ is $\{(1,5), (3,2), (4,8), (7,6)\}$.
    }
    \label{fig:permutation with directed perfect matching}
\end{figure}
The {\em word of sources} (resp. {\em word of sinks}) of
$\DMATCHING$ is the subword $\pi(i_1) \pi(i_2) \dots \pi(i_n)$ of $\pi$
where the indexes $i_1 < i_2 < \dots < i_n$ are the sources (resp.
sinks) of the arcs of $\DMATCHING$.
Figure~\ref{fig:example_directed_perfect_matching} shows an example for
these notions.
\begin{figure}[ht]
  \centering
  \begin{tikzpicture}
    [xscale=0.35,yscale=.3,inner sep=2pt,node distance=1.1cm]s
    \tikzstyle{Vertex}=[draw,shape=circle,minimum size=.8cm,
        circular drop shadow,inner sep=2pt,fill=white]
    \draw node[Vertex,fill=black!20](U01){$1$};
    \draw node[Vertex,right of=U01](U02){$2$};
    \draw node[Vertex,fill=black!20,right of=U02](U03){$3$};
    \draw node[Vertex,fill=black!20,right of=U03](U04){$4$};
    \draw node[Vertex,right of=U04](U05){$5$};
    \draw node[Vertex,right of=U05](U06){$6$};
    \draw node[Vertex,fill=black!20,right of=U06](U07){$7$};
    \draw node[Vertex,right of=U07](U08){$8$};
    \draw [thick,->,>=latex']
    (U01.north) .. controls ($ (U01.north) + (0,4.25) $)
    and ($ (U06.north) + (0,4.25) $) .. (U06.north);
    \draw [thick,->,>=latex']
    (U03.north) .. controls ($ (U03.north) + (0,4.25) $)
    and ($ (U08.north) + (0,4.25) $) .. (U08.north);
    \draw [thick,->,>=latex']
    (U04.north) .. controls ($ (U04.north) + (0,2.5) $)
    and ($ (U02.north) + (0,2.5) $) .. (U02.north);
    \draw [thick,->,>=latex']
    (U07.north) .. controls ($ (U07.north) + (0,2.5) $)
    and ($ (U05.north) + (0,2.5) $) .. (U05.north);
    \draw node[below of=U01]{$4$};
    \draw node[below of=U02]{$1$};
    \draw node[below of=U03]{$3$};
    \draw node[below of=U04]{$2$};
    \draw node[below of=U05]{$8$};
    \draw node[below of=U06]{$5$};
    \draw node[below of=U07]{$7$};
    \draw node[below of=U08]{$6$};
  \end{tikzpicture}
  \caption{A directed perfect matching $\DMATCHING$ on the permutation
    \mbox{$\pi = 41328576$}. The word of sources of $\DMATCHING$ is
    $\mathbf{4327}$ and its word of sinks is $1856$. Unlike in
    Figure~\ref{fig:permutation with directed perfect matching},
    $\DMATCHING$ is not drawn on the permutation matrix of~$\pi$.}
  \label{fig:example_directed_perfect_matching}
\end{figure}

We describe here two notions of patterns for directed perfect matchings
on permutations together with the notions of occurrences of patterns
accompanying them. Let $\pi$ be a permutation of size $2n$ and
$\DMATCHING = (\mathcal{G}, \pi)$ be a directed perfect matching on
$\pi$.
\begin{enumerate}[fullwidth]
    \item An {\em unlabeled pattern} is a directed perfect matching
    $\mathcal{U} = ([2k], A)$, where $k \leq n$. We say that
    $\DMATCHING$ contains an {\em unlabeled occurrence} of $\mathcal{U}$
    if there is an increasing map $\phi : [2k] \to [2n]$ ({\em i.e.},
    $i < j \in [2k]$ implies $\phi(i) < \phi(j)$) such that, if $(i, i')$
    is an arc of $\mathcal{U}$ then $(\phi(i), \phi(i'))$ is an arc of
    $\mathcal{G}$. Observe that this first notion of pattern occurrence
    does not depend on the permutation~$\pi$. In other words,
    $\DMATCHING$ contains an unlabeled occurrence of $\mathcal{U}$ if
    $\mathcal{G}$ contains a copy of $\mathcal{U}$ as a subgraph by
    changing some of its labels if necessary.

    \item A {\em labeled pattern} is a directed perfect matching
    $\mathcal{P} = (\mathcal{U}, \sigma)$ on a permutation $\sigma$
    of size $2k$. We say that $\DMATCHING$ contains a {\em labeled
    occurrence} of $\mathcal{P}$ if $\DMATCHING$ contains an unlabeled
    occurrence of the directed perfect matching $\mathcal{U} = ([2k], A)$
    such that $\STD(\pi(\phi(1))\pi(\phi(2))\ldots\pi(\phi(2k))) = \sigma$,
    where $\phi$ is a map defined as above. In other words, $\DMATCHING$
    contains a labeled occurrence of $\mathcal{P}$ if $\mathcal{G}$
    contains a copy of $\mathcal{U}$ as a subgraph and the word
    consisting in the letters of $\pi$ associated with each vertices of
    this copy in $\mathcal{G}$ is order-isomorphic to~$\sigma$.
\end{enumerate}
When $\DMATCHING$ does not contain any unlabeled occurrence (resp.
labeled occurrence) of an unlabeled pattern $\mathcal{U}$ (resp. labeled
pattern $\mathcal{P}$), we say that $\DMATCHING$ {\em avoids}
$\mathcal{U}$ (resp. $\mathcal{P}$). This definition naturally extends
to sets of patterns by setting that $\DMATCHING$ {\em avoids} the set
of unlabeled patterns (resp. labeled patterns)
$U = \{\mathcal{U}_1, \dots, \mathcal{U}_\ell\}$ (resp.
$P = \{\mathcal{P}_1, \dots, \mathcal{P}_\ell\}$) if $\DMATCHING$
avoids every $\mathcal{U}_i$ of~$U$ (resp. $\mathcal{P}_i$ of~$P$).

In this paper, we shall consider only patterns of size $4$. The set of
all unlabeled patterns of this size is
\begin{equation}
    \mathcal{P} =
    \mathcal{P}_{\mathrm{prec}}
    \cup
    \mathcal{P}_{\mathrm{cont}}
    \cup
    \mathcal{P}_{\mathrm{cros}},
\end{equation}
where
\begin{equation}
    \mathcal{P}_{\mathrm{prec}} =
    \left\{
    \begin{split}\PrecedenceLL\end{split},
    \begin{split}\PrecedenceLR\end{split},
    \begin{split}\PrecedenceRL\end{split},
    \begin{split}\PrecedenceRR\end{split}
    \right\},
\end{equation}
\begin{equation}
    \mathcal{P}_{\mathrm{cont}} =
    \left\{
    \begin{split}\InclusionLL\end{split},
    \begin{split}\InclusionLR\end{split},
    \begin{split}\InclusionRL\end{split},
    \begin{split}\InclusionRR\end{split}
    \right\},
\end{equation}
\begin{equation}
    \mathcal{P}_{\mathrm{cros}} =
    \left\{
    \begin{split}\CrossingLL\end{split},
    \begin{split}\CrossingLR\end{split},
    \begin{split}\CrossingRL\end{split},
    \begin{split}\CrossingRR\end{split}
    \right\}.
\end{equation}
In these drawings, the vertices of each pattern are implicitly indexed
from left to right by $1$ to $4$. Besides,  any labeled pattern
$\mathcal{P} = (\mathcal{U}, \sigma)$ is depicted by drawing
$\mathcal{U}$ and by labeling all its vertices $i$ by~$\sigma_i$.

To give
some examples of the previous notions, observe that a directed perfect
matching $\DMATCHING$ on a permutation contains an occurrence of the
unlabeled pattern $\InclusionRL$ if there are four vertices
$i_1 < i_2 < i_3 < i_4$ of $\DMATCHING$ such that $(i_1, i_4)$ and
$(i_3, i_2)$ are arcs of $\DMATCHING$. Moreover, $\DMATCHING$
is \emph{containment-free} (resp. \emph{crossing-free}) if it avoids all
patterns of $\mathcal{P}_{\mathrm{cont}}$ (resp.
$\mathcal{P}_{\mathrm{cros}}$).
For example, the directed perfect matching on
the permutation of Figure~\ref{fig:example_directed_perfect_matching}
\begin{itemize}
    \item contains exactly two unlabeled occurrences of the pattern
    $\InclusionRL$ corresponding to the arcs $(1, 6)$ and $(4, 2)$, or $(3, 8)$ and
    $(7, 5)$;
    \item contains exactly one unlabeled occurrence of $\CrossingLR$ corresponding to
    the arcs $(4, 2)$ and $(3, 8)$;
    \item avoids the unlabeled pattern~$\InclusionRR$.
\end{itemize}
The directed perfect matching on the permutation of
Figure~\ref{fig:permutation with directed perfect matching}
\begin{itemize}
    \item contains a labeled occurrence of the pattern
    $\LabeledInclusionRL{.24}{2}{3}{1}{4}$ corresponding to the arcs $(1,5)$ and
    $(3,2)$;
    \item contains a labeled occurrence of the pattern
    $\LabeledCrossingRR{.08}{2}{3}{4}{1}$ corresponding to the arcs $(1,5)$ and $(4,8)$;
    \item contains a labeled occurrence of the pattern
    $\LabeledPrecedenceRL{.08}{1}{4}{3}{2}$ corresponding to the arcs arcs $(1,5)$
    and~$(7,6)$;
    \item contains a labeled occurrence of the pattern
    $\LabeledPrecedenceLR{.08}{4}{2}{3}{1}$ corresponding to the arcs $(3,2)$
    and~$(4,8)$;
    \item contains a labeled occurrence of the pattern
    $\LabeledPrecedenceLL{.08}{4}{1}{3}{2}$ corresponding to the arcs $(3,2)$
    and~$(7,6)$;
    \item contains a labeled occurrence of the pattern
    $\LabeledInclusionRL{.24}{4}{3}{2}{1}$ corresponding to the arcs $(4,8)$
    and~$(7,6)$;
    \item avoids all other labeled patterns of size~$4$.
\end{itemize}


\section{Shuffle product on permutations}
\label{section:Shuffle product on permutations}

The main purpose of this section is to give a formal definition of the
shuffle product $\SHUFFLE$ on permutations. We shall define $\SHUFFLE$ by first defining a
co-product called the unshuffling operator $\Delta$ on permutations.
Then $\SHUFFLE$ is
defined to be the dual of $\Delta$. The reason that we define $\SHUFFLE$ in terms of $\Delta$ is
due to the fact that many properties of $\SHUFFLE$ depend on properties of $\Delta$ and
those properties are easier to prove on the co-product side.
We invite the reader unfamiliar with the
concepts of coproduct and duality to consult~\cite{Joni:Rota:1979}
or~\cite{Grinberg:Reiner:2014}.

Let us denote by $\QQ[S]$ the linear span of all permutations. We
define a linear coproduct $\Delta$ on $\QQ[S]$ in the following way. For
any permutation $\pi$, we set
\begin{equation} \label{equ:unshuffling_coproduct}
  \Delta(\pi) =
  \sum_{P_1 \sqcup P_2 = [|\pi|]}
  \STD\left(\pi_{|P_1}\right) \otimes \STD\left(\pi_{|P_2}\right).
\end{equation}
We call $\Delta$ the {\em unshuffling coproduct of permutations}. For
instance,
\begin{equation}
  \Delta(213) =
  \epsilon \otimes 213 + 2 \cdot 1 \otimes 12 +
  1 \otimes 21 + 2 \cdot 12 \otimes 1 + 21 \otimes 1
  + 213 \otimes \epsilon,
\end{equation}
\begin{equation}
  \Delta(1234) =
  \epsilon \otimes 1234 + 4 \cdot 1 \otimes 123 + 6 \cdot 12 \otimes 12 +
  4 \cdot 123 \otimes 1 + 1234 \otimes \epsilon,
\end{equation}
\begin{equation}\begin{split} \label{equ:example_unshuffling_coproduct}
    \Delta(1432) & =
    \epsilon \otimes 1432 + 3 \cdot {\bf 1 \otimes 132} + 1 \otimes 321 +
    3 \cdot 12 \otimes 21 \\ & + 3 \cdot 21 \otimes 12
    + 3 \cdot 132 \otimes 1 + 321 \otimes 1 + 1432 \otimes \epsilon.
\end{split}\end{equation}
Observe that the coefficient of the tensor $1 \otimes 132$ is $3$
in~\eqref{equ:example_unshuffling_coproduct} because there are exactly
three ways to extract from the permutation $1432$ two disjoint subwords which are,
respectively, order-isomorphic to the permutations $1$ and~$132$.

We can now define our shuffle product $\SHUFFLE$ as the product that corresponds to the
co-product $\Delta$ under duality.
From~\eqref{equ:unshuffling_coproduct}, for any permutation $\pi$, we
have
\begin{equation} \label{equ:schematic_coproduct}
  \Delta(\pi) =
  \sum_{\sigma, \nu \in S} \lambda_{\sigma, \nu}^\pi \;
  \sigma \otimes \nu,
\end{equation}
where the $\lambda_{\sigma, \nu}^\pi$ are nonnegative integers. By the
definition~\eqref{equ:unshuffling_coproduct} of $\Delta$, the
$\lambda_{\sigma, \nu}^\pi$ are equal to the number of different ways to
extract from $\pi$ two disjoint subwords respectively order-isomorphic
to $\sigma$ and $\nu$. Now, by definition of duality, the dual product
of $\Delta$, denoted by $\SHUFFLE$, is a linear binary product on
$\QQ[S]$. It satisfies, for any permutations $\sigma$ and~$\nu$,
\begin{equation}
  \sigma \SHUFFLE \nu =
  \sum_{\pi \in S}
  \lambda_{\sigma, \nu}^\pi \; \pi,
\end{equation}
where the coefficients $\lambda_{\sigma, \nu}^\pi$ are the ones
of~\eqref{equ:schematic_coproduct}. We call $\SHUFFLE$ the
{\em shuffle product of permutations}. For instance,
\begin{equation}\begin{split} \label{equ:example_shuffle_product}
    12 \SHUFFLE 21 & =
    1243 + 1324 + 2 \cdot 1342 + 2 \cdot 1423 + 3 \cdot {\bf 1432} +
    2134 + 2 \cdot 2314 \\
    & + 3 \cdot 2341 + 2413 + 2 \cdot 2431 + 2 \cdot 3124 + 3142 +
    3 \cdot 3214 + 2 \cdot 3241 \\
    & + 3421 + 3 \cdot 4123 + 2 \cdot 4132 + 2 \cdot 4213 + 4231 + 4312.
\end{split}\end{equation}
Observe that the coefficient $3$ of the permutation $1432$
in~\eqref{equ:example_shuffle_product} comes from the fact that the
coefficient of the tensor $12 \otimes 21$ is $3$
in~\eqref{equ:example_unshuffling_coproduct}.

Intuitively, the product $\SHUFFLE$ shuffles the values and the positions of the
letters of the permutations. One can observe that the empty permutation
$\epsilon$ is a unit for $\SHUFFLE$ and that this product is graded by
the sizes of the permutations ({\em i.e.}, the product of a permutation
of size $n$ with a permutation of size $m$ produces a sum of permutations
of size $n + m$).

We say that a permutation $\pi$ {\em appears} in the shuffle
$\sigma \SHUFFLE \nu$ of two permutations $\sigma$ and $\nu$ if the
coefficient $\lambda_{\sigma, \nu}^\pi$ defined above is different from
zero. In a more combinatorial way, this is equivalent to say that there
are two sets $P_1$ and $P_2$ of disjoints indexes of letters of $\pi$
satisfying $P_1 \sqcup P_2 = [|\pi|]$ such that the subword $\pi_{|P_1}$
is order-isomorphic to $\sigma$ and the subword $\pi_{|P_2}$ is
order-isomorphic to $\nu$.

A permutation $\pi$ is a {\em square} if there is a permutation
$\sigma$ such that $\pi$ appears in $\sigma \SHUFFLE \sigma$.
In this case, we say that $\sigma$ is a {\em square root} of $\pi$.
Equivalently, $\pi$ is a square with $\sigma$ as square root if and only
if in the expansion of $\Delta(\pi)$, there is a tensor
$\sigma \otimes \sigma$ with a nonzero coefficient. In a more
combinatorial way, this is equivalent to saying that there are two sets
$P_1$ and $P_2$ of disjoints indexes of letters of $\pi$ satisfying
$P_1 \sqcup P_2 = [|\pi|]$ such that the subwords $\pi_{|P_1}$ and
$\pi_{|P_2}$ are order-isomorphic. Computer experiments give us the
first numbers of square permutations with respects to their size, which
are, from size $0$ to $10$,
\begin{equation}
  1, 0, 2, 0, 20, 0, 504, 0, 21032, 0, 1293418.
\end{equation}
This sequence (and its subsequence obtained by removing the $0$'s) is for
the time being not listed in~\cite{Slo}. The first square permutations
are listed in Table~\ref{tab:first_square_permutations}.
\begin{table}[ht!]
  \centering
  \begin{tabular}{c|c|c}
    Size $0$ \, & Size $2$ & Size $4$ \\ \hline
    \multirow{2}{*}{\, $\epsilon$} &
    \multirow{2}{*}{\, $12$, $21$ \,} &
    $1234$, $1243$, $1423$, $1324$, $1342$, $4132$, $3124$, $3142$,
    $3412$, $4312$, \\
    & & $2134$, $2143$, $2413$, $4213$, $2314$, $2431$, $4231$,
    $3241$, $3421$, $4321$
  \end{tabular}
  \bigskip

  \caption{The square permutations of sizes $0$ to~$4$.}
  \label{tab:first_square_permutations}
\end{table}


\section{Binary square words and permutations}
\label{section:Binary square words and permutations}

In this section, we shall show that the square binary words are in
one-to-one correspondence with square permutations avoiding some
patterns. This property establishes a link between the shuffle of binary
words and our shuffle of permutations and allows us to obtain a new
description of square binary words.

Let $u$ be a binary word of length $n$ with $k$ occurrences of $0$.
We denote by $\BINTOPERM$ (\emph{Binary word To Permutation}) the map sending
any such word $u$ to the permutation obtained by replacing from left to
right each occurrence of $0$ in $u$ by $1$, $2$, \dots, $k$, and from
right to left each occurrence of $1$ in $u$ by $k + 1$, $k + 2$, \dots, $n$.
For instance,
\begin{equation}
  \BINTOPERM({\bf 1}00{\bf 1}0{\bf 1}{\bf 1}0{\bf 1}000) =
            {\bf C} 12 {\bf B} 3 {\bf A} {\bf 9} 4 {\bf 8} 567,
\end{equation}
where $\mathrm{A}$, $\mathrm{B}$, and $\mathrm{C}$ respectively stand
for $10$, $11$, and $12$. Observe that for any nonempty permutation
$\pi$ in the image of $\BINTOPERM$, there is exactly one binary word $u$
such that $\BINTOPERM(u0) = \BINTOPERM(u1) = \pi$. In support of this
observation, when $\pi$ has an even size, we denote by $\PERMTOBIN(\pi)$
(\emph{Permutation To Binary word}) the word $ua$ such that $|ua|_0$ and $|ua|_1$
are both even, where $a \in \{0, 1\}$. For instance,
\begin{equation}
    \PERMTOBIN({\bf 6}1{\bf 5}{\bf 4}23) = {\bf 1}0{\bf 1}{\bf 1}00
    \qquad \mbox{and} \qquad
    \PERMTOBIN(1{\bf 4}2{\bf 3}) = 0{\bf 1}0{\bf 1}.
\end{equation}

\begin{Proposition} \label{prop:bijection_binary_to_permutations_squares}
  For any $n \geq 0$, the map $\BINTOPERM$ restricted to the set of
  square binary words of length $2n$ is a bijection between this last
  set and the set of square permutations of size $2n$ avoiding the
  patterns $213$ and $231$.
\end{Proposition}

\begin{proof}
  [Proof of Proposition~\ref{prop:bijection_binary_to_permutations_squares}]
  The statement of the proposition is a consequence of the
  following claims implying that $\PERMTOBIN$ is the inverse
  map of $\BINTOPERM$ over the set of square binary words.

  \begin{Claim} \label{claim:binary_to_permutation_avoiding}
    The image of $\BINTOPERM$ is the set of all permutations
    avoiding $213$ and $231$.
  \end{Claim}
  \begin{proof}[Proof of Claim~\ref{claim:binary_to_permutation_avoiding}]
    Let us first show that the image of $\BINTOPERM$ contains only
    permutations avoiding $213$ and $231$. Let $u$ be a binary word,
    $\pi = \BINTOPERM(u)$, and $P_0$ (resp. $P_1$) be the set of the
    positions of the occurrences of $0$ (resp. $1$) in $u$. By
    definition of $\BINTOPERM$, from left to right, the subword
    $v = \pi_{|P_0}$ is increasing and the subword $w = \pi_{|P_1}$
    is decreasing, and all letters of $w$ are greater than those
    of~$v$. Now, assume that $123$ occurs in $\pi$.
    Then, since $v$ is increasing and $w$ is decreasing, there is an
    occurrence of $3$ (resp. $13$, $23$) in $v$ and a relative
    occurrence of $21$ (resp. $2$, $1$) in $w$. All these three cases
    contradict the fact that all letters of $w$ are greater than
    those of $v$. A similar argument shows that $\pi$ avoids~$231$
    as well.

    Finally, observe that any permutation $\pi$ avoiding $213$ and
    $231$ necessarily starts by the smallest possible letter or the
    greatest possible letter. This property is then true for the
    suffix of $\pi$ obtained by deleting its first letter,
    and so on for all of its suffixes. Thus, by
    replacing each letter $a$ of $\pi$ by $0$ (resp. $1$) if $a$ has the
    role of a smallest (resp. greatest) letter, one obtains a binary
    word $u$ such that $\BINTOPERM(u) = \pi$. Hence, all permutations
    avoiding $213$ and $231$ are in the image of~$\BINTOPERM$.
  \end{proof}

  \begin{Claim} \label{claim:square_binary_to_square_permutation}
    If $u$ is a square binary word, $\BINTOPERM(u)$ is a square
    permutation.
  \end{Claim}
  \begin{proof}
    [Proof of Claim~\ref{claim:square_binary_to_square_permutation}]
    Since $u$ is a square binary word, there is a binary word $v$
    such that $u \in v \shuffle v$. Then, there are two disjoint
    sets $P$ and $Q$ of positions of letters of $u$ such that
    $u_{|P} = v = u_{|Q}$. Now, by definition of $\BINTOPERM$, the
    words $\BINTOPERM(u)_{|P}$ and $\BINTOPERM(u)_{|Q}$ have the
    same standarization $\sigma$. Hence, and by definition of
    the shuffle product of permutations, $\BINTOPERM(u)$ appears in
    $\sigma \SHUFFLE \sigma$, showing that $\BINTOPERM(u)$ is a
    square permutation.
  \end{proof}

  \begin{Claim} \label{claim:square_permutation_to_square_binary}
    If $\pi$ is a square permutation avoiding $213$ and $231$,
    $\PERMTOBIN(\pi)$ is a square binary word.
  \end{Claim}
  \begin{proof}
    [Proof of Claim~\ref{claim:square_permutation_to_square_binary}]
    Let $\pi$ be a square permutation avoiding $213$ and $231$. By
    Claim~\ref{claim:binary_to_permutation_avoiding}, $\pi$ is in
    the image of $\BINTOPERM$ and hence, $u = \PERMTOBIN(\pi)$ is a
    well-defined binary word. Since $\pi$ is a square permutation,
    there are two disjoint sets $P_1$ and $P_2$ of indexes of letters
    of $\pi$ such that $\pi_{|P_1}$ and $\pi_{|P_2}$ are
    order-isomorphic. This implies, by the definitions of $\BINTOPERM$
    and $\PERMTOBIN$, that $u_{|P_1} = u_{|P_2}$, showing that $u$
    is a square binary word.
  \end{proof}

  This ends the proof of
  Proposition~\ref{prop:bijection_binary_to_permutations_squares}
\end{proof}

The number of square binary words is Sequence \OEIS{A191755} of~\cite{Slo}
beginning by
\begin{equation}
  1, 0, 2, 0, 6, 0, 22, 0, 82, 0, 320, 0, 1268, 0, 5102, 0, 020632.
\end{equation}
According to
Proposition~\ref{prop:bijection_binary_to_permutations_squares}, this
is also the sequence enumerating square permutations avoiding $213$
and~$231$.
Notice that it is conjectured in \cite{Henshall:Rampersad:Shallit:2011} that the
number of square binary words of length $2n$ is
$\binom{2n}{n} \frac{2^n}{n+1} - \binom{2n-1}{n+1}2^{n-1} + O(2^{n-2})$.


\section{Algebraic issues}
\label{section:Algebraic issues}

The aim of this section is to establish some of properties of the
shuffle product of permutations $\SHUFFLE$. It is worth to note that, as
we will see, algebraic properties of the unshuffling coproduct $\Delta$
of permutations defined in
Section~\ref{section:Shuffle product on permutations} lead to
combinatorial properties of $\SHUFFLE$.

\begin{Proposition} \label{prop:shuffle_associative_commutative}
  The shuffle product $\SHUFFLE$ of permutations is associative and
  commutative.
\end{Proposition}
\begin{proof}
  [Proof of Proposition~\ref{prop:shuffle_associative_commutative}]
  To prove the associativity of $\SHUFFLE$, it is convenient to show
  that its dual coproduct $\Delta$ is coassociative, that is
  \begin{equation}
    (\Delta \otimes I) \Delta = (I \otimes \Delta) \Delta,
  \end{equation}
  where $I$ denotes the identity map. This strategy relies on the fact
  that a product is associative if and only if its dual coproduct is
  coassociative. For any permutation $\pi$, we have
  \begin{equation} \begin{split}
      \label{equ:shuffle_associative_commutative}
      (\Delta \otimes I) \Delta(\pi) & =
      (\Delta \otimes I)
      \sum_{P_1 \sqcup P_2 = [|\pi|]}
      \STD\left(\pi_{|P_1}\right) \otimes \STD\left(\pi_{|P_2}\right) \\
      & =
      \sum_{P_1 \sqcup P_2 = [|\pi|]}
      \Delta\left(\STD\left(\pi_{|P_1}\right)\right)
      \otimes I\left(\STD\left(\pi_{|P_2}\right)\right) \\
      & =
      \sum_{P_1 \sqcup P_2 = [|\pi|]} \;
      \sum_{Q_1 \sqcup Q_2 = [|P_1|]}
      \STD\left(\STD\left(\pi_{|P_1}\right)_{|Q_1}\right)
      \otimes
      \STD\left(\STD\left(\pi_{|P_1}\right)_{|Q_2}\right)
      \otimes \STD\left(\pi_{|P_2}\right) \\
      & =
      \sum_{P_1 \sqcup P_2 \sqcup P_3 = [|\pi|]}
      \STD\left(\pi_{|P_1}\right) \otimes
      \STD\left(\pi_{|P_2}\right) \otimes
      \STD\left(\pi_{|P_3}\right).
  \end{split} \end{equation}
  An analogous computation shows that $(I \otimes \Delta) \Delta(\pi)$
  is equal to the last member
  of~\eqref{equ:shuffle_associative_commutative}, whence the
  associativity of $\SHUFFLE$.

  Finally, to prove the commutativity of $\SHUFFLE$, we shall show
  that $\Delta$ is cocommutative, that is for any permutation $\pi$,
  if in the expansion of $\Delta(\pi)$ there is a tensor
  $\sigma \otimes \nu$ with a coefficient $\lambda$, there is in the
  same expansion the tensor $\nu \otimes \sigma$ with the same
  coefficient $\lambda$. Clearly, a product is commutative if and only
  if its dual coproduct is cocommutative. Now, from the
  definition~\eqref{equ:unshuffling_coproduct} of $\Delta$, one
  observes that if the pair $(P_1, P_2)$ of subsets of $[|\pi|]$
  contributes to the coefficient of
  $\STD\left(\pi_{|P_1}\right) \otimes \STD\left(\pi_{|P_2}\right)$,
  the pair $(P_2, P_1)$ contributes to the coefficient of
  $\STD\left(\pi_{|P_2}\right) \otimes \STD\left(\pi_{|P_1}\right)$.
  This shows that $\Delta$ is cocommutative and hence, that $\SHUFFLE$
  is commutative.
\end{proof}

Proposition~\ref{prop:shuffle_associative_commutative} shows that
$\QQ[S]$ under the unshuffling coproduct $\Delta$ is a co-associative co-commutative
coalgebra which implies, by duality, that $\QQ[S]$ under $\SHUFFLE$ is an associative
commutative algebra

\begin{Lemma} \label{lem:endomorphisms}
  The three linear maps
  \begin{equation}
    \phi_1, \phi_2, \phi_3 : \QQ[S] \to \QQ[S]
  \end{equation}
  linearly sending a permutation $\pi$ to, respectively,
  $\widetilde{\pi}$, $\bar \pi$, and $\pi^{-1}$ are endomorphisms of
  associative algebras.
\end{Lemma}
\begin{proof}[Proof of Lemma~\ref{lem:endomorphisms}]
  To prove, for $j = 1, 2, 3$, that $\phi_j$ is a morphism of
  associative algebras, we have to prove that for all permutations
  $\sigma$ and $\nu$,
  \begin{equation}
    \phi_j(\sigma \SHUFFLE \nu) =
    \phi_j(\sigma) \SHUFFLE \phi_j(\nu).
  \end{equation}
  By duality, this is equivalent
  to showing that $\phi_j$ is a morphism of coalgebras, that is,
  \begin{equation}
    \Delta \phi_j = (\phi_j \otimes \phi_j) \Delta.
  \end{equation}
  In the sequel, $\pi$ is a permutation.

  If $P$ is a set of indexes of letters of $\pi$, we denote by
  $\widetilde{P}$ the set $\{|\pi| - i  + 1 : i \in P\}$. Now, since
  the operation $\widetilde{\,}$ defines a bijection on the set of the
  subsets of $[|\pi|]$, and since the standardization operation
  commutes with the mirror operation on words without multiple
  occurrence of a letter, we have
  \begin{equation} \begin{split}
      \Delta(\phi_1(\pi))
      & = \sum_{P_1 \sqcup P_2 = [|\pi|]}
      \STD\left(\phi_1(\pi)_{|P_1}\right)
      \otimes \STD\left(\phi_1(\pi)_{|P_2}\right) \\
      & = \sum_{P_1 \sqcup P_2 = [|\pi|]}
      \STD\left(\widetilde{\pi}_{|P_1}\right)
      \otimes \STD\left(\widetilde{\pi}_{|P_2}\right) \\
      & = \sum_{P_1 \sqcup P_2 = [|\pi|]}
      \STD\left(\widetilde{\pi}_{|\widetilde{P_1}}\right)
      \otimes \STD\left(\widetilde{\pi}_{|\widetilde{P_2}}\right) \\
      & = \sum_{P_1 \sqcup P_2 = [|\pi|]}
      \widetilde{\STD\left(\pi_{|P_1}\right)}
      \otimes \widetilde{\STD\left(\pi_{|P_2}\right)} \\
      & = \sum_{P_1 \sqcup P_2 = [|\pi|]}
      \phi_1\left(\STD\left(\pi_{|P_1}\right)\right)
      \otimes \phi_1\left(\STD\left(\pi_{|P_2}\right)\right) \\
      & = (\phi_1 \otimes \phi_1) \Delta(\pi).
  \end{split} \end{equation}
  This shows that $\phi_1$ is a morphism of coalgebras and hence, that
  $\phi_1$ is a morphism of associative algebras.

  Next, since by definition of the complementation operation on
  permutations, for any permutation $\tau$ and any indexes $i$ and $k$,
  we have $\tau(i) < \tau(k)$ if and only if $\bar \tau(i) > \bar \tau(k)$,
  we have
  \begin{equation} \begin{split}
      \Delta(\phi_2(\pi))
      & = \sum_{P_1 \sqcup P_2 = [|\pi|]}
      \STD\left(\phi_2(\pi)_{|P_1}\right)
      \otimes \STD\left(\phi_2(\pi)_{|P_2}\right) \\
      & = \sum_{P_1 \sqcup P_2 = [|\pi|]}
      \STD\left(\bar \pi_{|P_1}\right)
      \otimes \STD\left(\bar \pi_{|P_2}\right) \\
      & = \sum_{P_1 \sqcup P_2 = [|\pi|]}
      \phi_2\left(\STD\left(\pi_{|P_1}\right)\right)
      \otimes \phi_2\left(\STD\left(\pi_{|P_2}\right)\right) \\
      & = (\phi_2 \otimes \phi_2) \Delta(\pi).
  \end{split} \end{equation}
  This shows that $\phi_2$ is a morphism of coalgebras and hence, that
  $\phi_2$ is a morphism of associative algebras.

  Finally, for any permutation $\tau$, if $P$ is a set of indexes of
  letters of $\tau$, we denote by $P_\tau^{-1}$ the set
  $\{\tau(i) : i \in P\}$. Since the map sending a subset $P$ of
  $[|\pi|]$ to $P_\pi^{-1}$ is a bijection, and since
  $\STD\left(\pi_{|P}\right)^{-1} = \STD\left(\pi^{-1}_{|P_\pi^{-1}}\right)$,
  we have
  \begin{equation} \begin{split}
      \Delta(\phi_3(\pi))
      & = \sum_{P_1 \sqcup P_2 = [|\pi|]}
      \STD\left(\phi_3(\pi)_{|P_1}\right)
      \otimes \STD\left(\phi_3(\pi)_{|P_2}\right) \\
      & = \sum_{P_1 \sqcup P_2 = [|\pi|]}
      \STD\left(\pi^{-1}_{|P_1}\right)
      \otimes \STD\left(\pi^{-1}_{|P_2}\right) \\
      & = \sum_{P_1 \sqcup P_2 = [|\pi|]}
      \STD\left(\pi^{-1}_{|{P_1}_\pi^{-1}}\right)
      \otimes \STD\left(\pi^{-1}_{|{P_2}_\pi^{-1}}\right) \\
      & = \sum_{P_1 \sqcup P_2 = [|\pi|]}
      \STD\left(\pi_{|P_1}\right)^{-1}
      \otimes \STD\left(\pi_{|P_2}\right)^{-1} \\
      & = \sum_{P_1 \sqcup P_2 = [|\pi|]}
      \phi_3\left(\STD\left(\pi_{|P_1}\right)\right)
      \otimes \phi_3\left(\STD\left(\pi_{|P_2}\right)\right) \\
      & = (\phi_3 \otimes \phi_3) \Delta(\pi).
  \end{split} \end{equation}
  This shows that $\phi_3$ is a morphism of coalgebras and hence, that
  $\phi_3$ is a morphism of associative algebras.
\end{proof}

We now use the algebraic properties of $\SHUFFLE$ exhibited by
Lemma~\ref{lem:endomorphisms} to obtain combinatorial properties
of square permutations.

\begin{Proposition} \label{prop:square_stability}
  Let $\pi$ be a square permutation and $\sigma$ be a square root of
  $\pi$. Then,
  \begin{enumerate}[label={\it (\roman*)},fullwidth]
  \item \label{item:square_stability_1}
    the permutation $\widetilde{\pi}$ is a square and
    $\widetilde{\sigma}$ is one of its square roots;
  \item \label{item:square_stability_2}
    the permutation $\bar \pi$ is a square and $\bar \sigma$ is one of
    its square roots;
  \item \label{item:square_stability_3}
    the permutation $\pi^{-1}$ is a square and $\sigma^{-1}$ is one of
    its square roots.
  \end{enumerate}
\end{Proposition}
\begin{proof}[Proof of Proposition~\ref{prop:square_stability}]
  All statements~\ref{item:square_stability_1},
  \ref{item:square_stability_2}, and~\ref{item:square_stability_3} are
  consequences of Lemma~\ref{lem:endomorphisms}. Indeed,
  since $\pi$ is a square permutation and $\sigma$ is a square root of
  $\pi$, by definition, $\pi$ appears in the product
  $\sigma \SHUFFLE \sigma$. Now, by Lemma~\ref{lem:endomorphisms},
  for any $j = 1, 2, 3$, since $\phi_j$ is a morphism of associative
  algebras from $\QQ[S]$ to $\QQ[S]$, $\phi_j$ commutes with the
  shuffle product of permutations $\SHUFFLE$. Hence, in particular,
  one has
  \begin{equation}
    \phi_j(\sigma \SHUFFLE \sigma) =
    \phi_j(\sigma) \SHUFFLE \phi_j(\sigma).
  \end{equation}
  Then, since $\pi$ appears in $\sigma \SHUFFLE \sigma$, $\phi_j(\pi)$
  appears in $\phi_j(\sigma \SHUFFLE \sigma)$ and appears also in
  $\phi_j(\sigma) \SHUFFLE \phi_j(\sigma)$. This shows that
  $\phi_j(\sigma)$ is a square root of $\phi_j(\pi)$ and
  implies~\ref{item:square_stability_1}, \ref{item:square_stability_2},
  and~\ref{item:square_stability_3}.
\end{proof}

Let us make an observation about Wilf-equivalence classes of permutations
restrained on square permutations. Recall that two permutations $\sigma$
and $\nu$ of the same size are {\em Wilf equivalent} if
$|S_n(\sigma)| = |S_n(\nu)|$ for all $n \geq 0$. The
well-known~\cite{Simion:Schmidt:EJC:1985} fact that there is a single
Wilf-equivalence class of permutations of size $3$ together with
Proposition~\ref{prop:square_stability} imply that $123$ and $321$ are
in the same Wilf-equivalence class of square permutations, and that
$132$, $213$, $231$, and $312$ are in the same Wilf-equivalence class of
square permutations. Computer experiments show us that there are two
Wilf-equivalence classes of square permutations of size $3$. Indeed, the
number of square permutations avoiding $123$ begins by
\begin{equation} \label{equ:sequence_square_123}
  1, 0, 2, 0, 12, 0, 118, 0, 1218, 0, 14272,
\end{equation}
while the number of square permutations avoiding $132$ begins by
\begin{equation} \label{equ:sequence_square_132}
  1, 0, 2, 0, 11, 0, 84, 0, 743, 0, 7108.
\end{equation}

Another consequence of Proposition~\ref{prop:square_stability}
is that its makes sense to enumerate the sets of square permutations
quotiented by the operations of mirror image, complement, and
inverse. The sequence enumerating these sets begins by
\begin{equation} \label{equ:sequence_square_classes}
  1, 0, 1, 0, 6, 0, 81, 0, 2774, 0, 162945.
\end{equation}

All Sequences~\eqref{equ:sequence_square_123}, \eqref{equ:sequence_square_132},
and~\eqref{equ:sequence_square_classes} (and their subsequences obtained
by removing the $0$s) are for the time being not listed in~\cite{Slo}.


\section{Algorithmic issues}
\label{section:Algorithmic issues}

This section is devoted to proving the \NP-hardness of recognizing square
permutations.
As in the case of words, we shall use a
linear graph framework where deciding whether a permutation is a square
reduces to computing some specific matching in the associated linear
graph~\cite{Rizzi:Vialette:CSR:2013,Buss:Soltys:2014}. We have, however,
to deal with directed graphs/perfect matchings satisfying some precise properties.
Let us first define two properties.

\begin{Definition}[Property $\mathbf{P_1}$]
  \label{definition:Property P_1}
  Let $\pi$ be a permutation. A directed perfect matching
  $\DMATCHING$ on $\pi$ is said to have property $\mathbf{P_1}$ if it
  avoids the following set of unlabeled patterns:
  \begin{equation} \label{equ:property_P_1}
    \mathcal{P}_1 =
    \left\{
    \begin{split}\InclusionLL\end{split},\;
    \begin{split}\InclusionLR\end{split},\;
    \begin{split}\InclusionRL\end{split},\;
    \begin{split}\InclusionRR\end{split},\;
    \begin{split}\CrossingLR\end{split},\;
    \begin{split}\CrossingRL\end{split}
    \right\}.
  \end{equation}
\end{Definition}
Observe that the unlabeled patterns of $\mathcal{P}_1$ are the four
of $\mathcal{P}_{\mathrm{cont}}$ and the two of $\mathcal{P}_{\mathrm{cros}}$
that have crossing edges in the opposite directions.

\begin{Definition}[Property $\mathbf{P_2}$]
  \label{definition:Property P_2}
  Let $\pi$ be a permutation. A directed perfect matching $\DMATCHING$
  on $\pi$ is said to have property $\mathbf{P_2}$ if, for any two
  distinct arcs $(i, i')$ and $(j, j')$ of $\DMATCHING$, we have
  $\pi(i) < \pi(j)$ if and only if $\pi(i') < \pi(j')$.
\end{Definition}

The rationale for introducing properties $\mathbf{P_1}$ and
$\mathbf{P_2}$ stems from the following lemma.

\begin{Lemma}
  \label{lemma:matching}
  Let $\pi$ be a permutation. The following statements are equivalent:
  \begin{enumerate}[label={\it (\roman*)},fullwidth]
  \item \label{item:matching_1}
    The permutation $\pi$ is a square.
  \item \label{item:matching_2}
    There exists a directed perfect matching $\DMATCHING$
    on $\pi$ satisfying properties~$\mathbf{P_1}$ and~$\mathbf{P_2}$.
  \end{enumerate}
\end{Lemma}
\begin{proof}[Proof of Lemma~\ref{lemma:matching}]
  Assume that~\ref{item:matching_1} holds. Since $\pi$ is a square,
  $\pi$ has a square root, say~$\sigma$. Let $2n = |\pi|$ (and
  hence $|\sigma| = n$). Then, by definition, there exist two sets
  \begin{equation}
    I^1 = \left\{i^1_1 < i^1_2 < \dots < i^1_n\right\}
    \quad \mbox{and} \quad
    I^2 = \left\{i^2_1 < i^2_2 < \dots < i^2_n\right\}
  \end{equation}
  of disjoint indexes of letters of $\pi$ such that $\pi_{|I^1}$ and
  $\pi_{|I^2}$ are both order-isomorphic to $\sigma$. Let
  $\mathcal{G} = (V, E)$ be the directed graph such that $V = [2n]$
  and
  \begin{math}
    E = \left\{\left(i^1_j, i^2_j\right) : j \in [n]\right\}.
  \end{math}
  It is easily seen that $\DMATCHING = (\mathcal{G}, \pi)$ is a
  directed perfect matching since $I^1 \cap I^2 = \emptyset$ and
  $I^1 \cup I^2 = [2n]$. We first show that $\DMATCHING$ avoids the
  unlabeled patterns of $\mathcal{P}_{\mathrm{cont}}$. Indeed, suppose, aiming at a
  contradiction, that such an occurrence appears for, say, arcs
  $\left(i^1_j, i^2_j\right)$ and $\left(i^1_k, i^2_k\right)$ of
  $\DMATCHING$. Assuming without loss of generality $i^1_j < i^1_k$,
  we are left with the four configurations
  
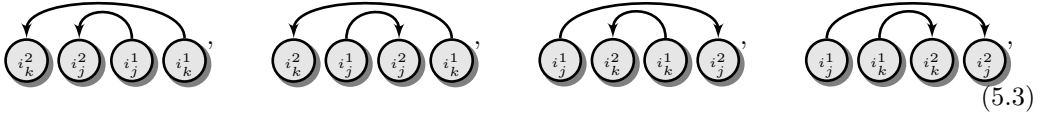
\begin{equation}
    \begin{split}
      \ConfigurationInclusionLL{.3}{$i^2_k$}{$i^2_j$}{$i^1_j$}{$i^1_k$}
    \end{split},
    \qquad
    \begin{split}
      \ConfigurationInclusionLR{.3}{$i^2_k$}{$i^1_j$}{$i^2_j$}{$i^1_k$}
    \end{split},
    \qquad
    \begin{split}
      \ConfigurationInclusionRL{.3}{$i^1_j$}{$i^2_k$}{$i^1_k$}{$i^2_j$}
    \end{split},
    \qquad
    \begin{split}
      \ConfigurationInclusionRR{.3}{$i^1_j$}{$i^1_k$}{$i^2_k$}{$i^2_j$}
    \end{split},
  \end{equation}
  where shadow nodes give the position in the permutation $\pi$.
  Then it follows that $i^2_j > i^2_k$. This is a contradiction since
  $i^1_j < i^1_k$ implies $j < k$, and hence, $i^2_j < i^2_k$.
  We now turn to proving that $\DMATCHING$ also avoids the unlabeled patterns
  $\CrossingLR$ and $\CrossingRL$.
  Indeed, suppose, aiming at a contradiction, that such
  an occurrence appears for, say, arcs
  $\left(i^1_j, i^2_j\right)$ and $\left(i^1_k, i^2_k\right)$ of
  $\DMATCHING$. Assuming without loss of generality $i^1_j < i^1_k$,
  we are left with the two configurations
  
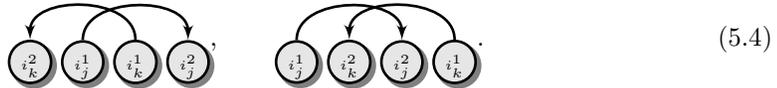
\begin{equation}
    \begin{split}
      \ConfigurationCrossingLR{.3}{$i^2_k$}{$i^1_j$}{$i^1_k$}{$i^2_j$}
    \end{split},
    \qquad
    \begin{split}
      \ConfigurationCrossingRL{.3}{$i^1_j$}{$i^2_k$}{$i^2_j$}{$i^1_k$}
    \end{split}.
  \end{equation}
  Then it follows that $i^2_j > i^2_k$. Again, this is a contradiction
  since $i^1_j < i^1_k$ implies $j < k$, and hence, $i^2_j < i^2_k$.
  Finally, for any two distinct arcs
  $\left(i^1_j, i^2_j\right)$ and $\left(i^1_k, i^2_k\right)$ of
  $\DMATCHING$, we have $\pi\left(i^1_j\right) < \pi\left(i^1_k\right)$
  if and only if $\pi\left(i^2_j\right) < \pi\left(i^2_k\right)$ since
  we are comparing in both cases two elements (at positions $j$ and
  $k$) in two patterns that are order-isomorphic to $\sigma$.
  Therefore, $\DMATCHING$ satisfies properties~$\mathbf{P_1}$ and $\mathbf{P_2}$,
  so that~\ref{item:matching_2} holds.

  Assume now that~\ref{item:matching_2} holds. Let
  \begin{equation}
    I^1 = \left\{i^1_1 < i^1_2 < \dots < i^1_n\right\}
    \quad \mbox{and} \quad
    I^2 = \left\{i^2_1 < i^2_2 < \dots < i^2_n\right\}
  \end{equation}
  such that $I^1$ is the set the sources of the arcs of $\DMATCHING$
  and $I^2$ is the set of the sinks of the arcs of $\DMATCHING$.
  Let us first show that, for every $j \in [n]$,
  $\left(i^1_j, i^2_j\right)$ is an arc of $\DMATCHING$. For that,
  we show that $\left(i^1_n, i^2_n\right)$ is an arc of $\DMATCHING$.
  Suppose, aiming at a contradiction that this is false. Then, there
  exist two vertices $i^2_p$ and $i^1_q$ of $\DMATCHING$ such
  that $\left(i^1_n, i^2_p\right)$ and $\left(i^1_q, i^2_n\right)$
  are arcs of $\DMATCHING$. Since $p < n$ and $q < n$, there is in
  $\DMATCHING$ one of the four configurations
  
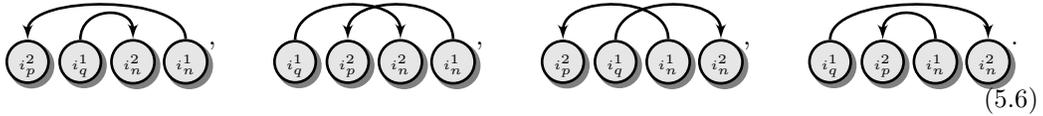
\begin{equation}
    \begin{split}
      \ConfigurationInclusionLR{.3}{$i^2_p$}{$i^1_q$}{$i^2_n$}{$i^1_n$}
    \end{split},
    \qquad
    \begin{split}
      \ConfigurationCrossingRL{.3}{$i^1_q$}{$i^2_p$}{$i^2_n$}{$i^1_n$}
    \end{split},
    \qquad
    \begin{split}
      \ConfigurationCrossingLR{.3}{$i^2_p$}{$i^1_q$}{$i^1_n$}{$i^2_n$}
    \end{split},
    \qquad
    \begin{split}
      \ConfigurationInclusionRL{.3}{$i^1_q$}{$i^2_p$}{$i^1_n$}{$i^2_n$}
    \end{split}.
  \end{equation}
  This is a contradiction since $\DMATCHING$ satisfies property~$\mathbf{P_1}$
  and hence avoids the unlabeled patterns $\InclusionLR$, $\CrossingRL$, $\CrossingLR$,
  and $\InclusionRL$. Therefore, $\left(i^1_n, i^2_n\right)$ is an
  arc of $\DMATCHING$. By iteratively applying the same reasoning,
  this also shows that all $\left(i^1_j, i^2_j\right)$, $j \in [n - 1]$,
  are arcs of $\DMATCHING$. Now, let
  $p^1$ be the word of sources and $p^2$ be the word of sinks of $\DMATCHING$.
  Clearly $p^1$ and $p^2$
  are disjoint in $\pi$ (since $\DMATCHING$ is a matching) and cover
  $\pi$ (since $\DMATCHING$ is perfect). Moreover, the fact that
  $\DMATCHING$ satisfies $\mathbf{P_2}$ implies immediately that
  $p^1$ and $p^2$ are order-isomorphic. Hence, this shows that $\pi$
  is a square, so that~\ref{item:matching_1} holds.
\end{proof}

Observe that, given a square permutation $\pi \in S_{2n}$ and a directed
perfect matching $\DMATCHING$ on $\pi$ satisfying properties~$\mathbf{P_1}$ and
$\mathbf{P_2}$, one can recover a square root of $\pi$ by considering
the standarization permutation of the word of sources (or, equivalently,
the word of sinks) of $\DMATCHING$.
Figure~\ref{fig:example containment-free matching} provides an
illustration of Lemma~\ref{lemma:matching} and of this observation.
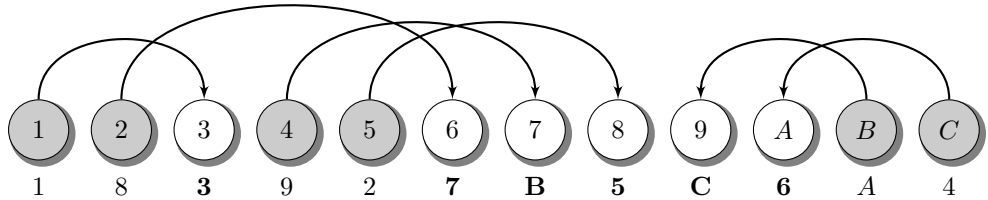
\begin{figure}[ht!]
  \centering
  \begin{tikzpicture}
    [xscale=0.35,yscale=.3,inner sep=2pt,node distance=1.1cm]
    \tikzstyle{Vertex}=[draw,shape=circle,minimum size=.8cm,
        circular drop shadow,inner sep=2pt]
    \draw node[Vertex,fill=black!20](U01){$1$};
    \draw node[Vertex,fill=black!20,right of=U01](U02){$2$};
    \draw node[Vertex,fill=white,right of=U02](U03){$3$};
    \draw node[Vertex,fill=black!20,right of=U03](U04){$4$};
    \draw node[Vertex,fill=black!20,right of=U04](U05){$5$};
    \draw node[Vertex,fill=white,right of=U05](U06){$6$};
    \draw node[Vertex,fill=white,right of=U06](U07){$7$};
    \draw node[Vertex,fill=white,right of=U07](U08){$8$};
    \draw node[Vertex,fill=white,right of=U08](U09){$9$};
    \draw node[Vertex,fill=white,right of=U09](U10){$A$};
    \draw node[Vertex,fill=black!20,right of=U10](U11){$B$};
    \draw node[Vertex,fill=black!20,right of=U11](U12){$C$};
    \draw [thick,->,>=latex']
    (U01.north) .. controls ($ (U01.north) + (0,3.5) $) and ($ (U03.north) + (0,3.5) $) .. (U03.north);
    \draw [thick,->,>=latex']
    (U02.north) .. controls ($ (U02.north) + (0,5.5) $) and ($ (U06.north) + (0,5.5) $) .. (U06.north);
    \draw [thick,->,>=latex']
    (U04.north) .. controls ($ (U04.north) + (0,4.5) $) and ($ (U07.north) + (0,4.5) $) .. (U07.north);
    \draw [thick,->,>=latex']
    (U05.north) .. controls ($ (U05.north) + (0,4.5) $) and ($ (U08.north) + (0,4.5) $) .. (U08.north);
    \draw [thick,->,>=latex']
    (U11.north) .. controls ($ (U11.north) + (0,3.5) $) and ($ (U09.north) + (0,3.5) $) .. (U09.north);
    \draw [thick,->,>=latex']
    (U12.north) .. controls ($ (U12.north) + (0,3.5) $) and ($ (U10.north) + (0,3.5) $) .. (U10.north);
    \draw node[below of=U01,node distance=.75cm]{$1$};
    \draw node[below of=U02,node distance=.75cm]{$8$};
    \draw node[below of=U03,node distance=.75cm]{$\mathbf{3}$};
    \draw node[below of=U04,node distance=.75cm]{$9$};
    \draw node[below of=U05,node distance=.75cm]{$2$};
    \draw node[below of=U06,node distance=.75cm]{$\mathbf{7}$};
    \draw node[below of=U07,node distance=.75cm]{$\mathbf{B}$};
    \draw node[below of=U08,node distance=.75cm]{$\mathbf{5}$};
    \draw node[below of=U09,node distance=.75cm]{$\mathbf{C}$};
    \draw node[below of=U10,node distance=.75cm]{$\mathbf{6}$};
    \draw node[below of=U11,node distance=.75cm]{$A$};
    \draw node[below of=U12,node distance=.75cm]{$4$};
  \end{tikzpicture}
  \caption{\label{fig:example containment-free matching}%
    A directed perfect matching $\DMATCHING$ on the permutation
    \mbox{$\pi = 183927B5C6A4$} satisfying the properties $\mathbf{P_1}$
    and $\mathbf{P_2}$. From $\DMATCHING$, it follows that $\pi$ is
    a square as it appears in the shuffle of
    ${1892A4}$ and
    $\mathbf{37B5C6}$, both being order-isomorphic to
    $145263$, a square root of~$\pi$.}
\end{figure}

Let $\pi$ be a permutation. For the sake of clarity, we will say that a
bunch of consecutive positions $P$ of $\pi$ is \emph{above} (resp.
\emph{below}) another bunch of consecutive positions $P'$ in $\pi$
if $\pi(i) > \pi(j)$ (resp. $\pi(i) < \pi(j)$) for every $i \in P$ and
every $j \in P'$. For example, $\sigma_1$ is above $\sigma_2$ (in an
equivalent manner, $\sigma_2$ is below $\sigma_1$) in
Figure~\ref{subfig:increasing before and above decreasing}, whereas
$\sigma_1$ is below $\sigma_2$ (in an equivalent manner, $\sigma_2$ is
above $\sigma_1$) in Figure~\ref{subfig:decreasing before and below
increasing}.

Moreover, if $\pi$ is a permutation satisfying
$\pi = \pi_1 \sigma_1 \pi_2 \sigma_2 \pi_3$ and $\DMATCHING$ is a
directed perfect matching on $\pi$, a {\em $(\sigma_1, \sigma_2)$-arc}
(resp. {\em $(\sigma_2, \sigma_1)$-arc}) of $\DMATCHING$ is any arc
$(i, j)$ (resp. $(j, i)$) of $\DMATCHING$ such that the $i$-th letter
of $\pi$ belongs to $\sigma_1$ and the $j$-th letter of $\pi$ belongs
to~$\sigma_2$.

Let us now state and prove some lemmas that will prove extremely useful
for simplifying the proof of upcoming
Proposition~\ref{proposition:hardness}.
First, whereas Lemma~\ref{lemma:matching} states that a directed perfect
matching on a permutation with
Property $\mathbf{P_1}$ avoids some unlabeled patterns of length $4$
(more specifically, it avoids the unlabeled patterns of $\mathcal{P}_1$),
the following two lemmas state that a directed perfect
matching on a permutation with Property
$\mathbf{P_2}$ also avoids some additional labeled patterns.
These two lemmas are easily proved by requiring Property~$\mathbf{P_2}$.
For example, an occurrence of the labeled pattern
$\LabeledCrossingRR{0.08}{3}{4}{2}{1}$ induces the existence of two arcs $(i_1, i_3)$ and
$(i_2, i_4)$ with $i_1 < i_2 < i_3 < i_4$ and
$\pi(i_4) < \pi(i_3) < \pi(i_2) < \pi(i_1)$.

\begin{figure}[htbp!]
  \centering
 \begin{tabular}{cccc}
    \begin{tikzpicture}[scale=0.165,label/.style={anchor=base}]
      \begin{scope}[]
        \draw[step=1cm,black!75,ultra thin,fill=black!50]
        (-0.2,-0.2) grid (8.2,8.2);
        \foreach \x/\y in {1/1,3/3,5/7,7/5} {
          \draw [fill=black] (\x,\y) circle (0.2);
        }
        \draw [black,line width=1pt,->,>=latex']
        (1,1) .. controls +(0,10) and +(0,4) .. (5,7);
        \draw [black,line width=1pt,->,>=latex']
        (3,3) .. controls +(0,8) and +(0,6) .. (7,5);
      \end{scope}
      \begin{scope}[xshift=10cm]
        \draw[step=1cm,black!75,ultra thin,fill=black!50]
        (-0.2,-0.2) grid (8.2,8.2);
        \foreach \x/\y in {1/1,3/3,5/7,7/5} {
          \draw [fill=black] (\x,\y) circle (0.2);
        }
        \draw [black,line width=1pt,<-,>=latex']
        (1,1) .. controls +(0,10) and +(0,4) .. (5,7);
        \draw [black,line width=1pt,<-,>=latex']
        (3,3) .. controls +(0,8) and +(0,6) .. (7,5);
      \end{scope}
      %
      \node [align=center] at (9,-2.5) {$1243$};
    \end{tikzpicture}
    &
    \begin{tikzpicture}[scale=0.165,label/.style={anchor=base}]
      \begin{scope}[]
        \draw[step=1cm,black!75,ultra thin,fill=black!50]
        (-0.2,-0.2) grid (8.2,8.2);
        \foreach \x/\y in {1/1,3/5,5/7,7/3} {
          \draw [fill=black] (\x,\y) circle (0.2);
        }
        \draw [black,line width=1pt,->,>=latex']
        (1,1) .. controls +(0,10) and +(0,4) .. (5,7);
        \draw [black,line width=1pt,->,>=latex']
        (3,5) .. controls +(0,4) and +(0,10) .. (7,3);
      \end{scope}
      \begin{scope}[xshift=10cm]
        \draw[step=1cm,black!75,ultra thin,fill=black!50]
        (-0.2,-0.2) grid (8.2,8.2);
        \foreach \x/\y in {1/1,3/5,5/7,7/3} {
          \draw [fill=black] (\x,\y) circle (0.2);
        }
        \draw [black,line width=1pt,<-,>=latex']
        (1,1) .. controls +(0,10) and +(0,4) .. (5,7);
        \draw [black,line width=1pt,<-,>=latex']
        (3,5) .. controls +(0,4) and +(0,10) .. (7,3);
      \end{scope}
      %
      \node [align=center] at (9,-2.5) {$1342$};
    \end{tikzpicture}
    &
        \begin{tikzpicture}[scale=0.165,label/.style={anchor=base}]
      \begin{scope}[]
        \draw[step=1cm,black!75,ultra thin,fill=black!50]
        (-0.2,-0.2) grid (8.2,8.2);
        \foreach \x/\y in {1/1,3/7,5/5,7/3} {
          \draw [fill=black] (\x,\y) circle (0.2);
        }
        \draw [black,line width=1pt,->,>=latex']
        (1,1) .. controls +(0,12) and +(0,6) .. (5,5);
        \draw [black,line width=1pt,->,>=latex']
        (3,7) .. controls +(0,4) and +(0,10) .. (7,3);
      \end{scope}
      \begin{scope}[xshift=10cm]
        \draw[step=1cm,black!75,ultra thin,fill=black!50]
        (-0.2,-0.2) grid (8.2,8.2);
        \foreach \x/\y in {1/1,3/7,5/5,7/3} {
          \draw [fill=black] (\x,\y) circle (0.2);
        }
        \draw [black,line width=1pt,<-,>=latex']
        (1,1) .. controls +(0,12) and +(0,6) .. (5,5);
        \draw [black,line width=1pt,<-,>=latex']
        (3,7) .. controls +(0,4) and +(0,10) .. (7,3);
      \end{scope}
      %
      \node [align=center] at (9,-2.5) {$1432$};
    \end{tikzpicture}
    &
        \begin{tikzpicture}[scale=0.165,label/.style={anchor=base}]
      \begin{scope}[]
        \draw[step=1cm,black!75,ultra thin,fill=black!50]
        (-0.2,-0.2) grid (8.2,8.2);
        \foreach \x/\y in {1/3,3/1,5/5,7/7} {
          \draw [fill=black] (\x,\y) circle (0.2);
        }
        \draw [black,line width=1pt,->,>=latex']
        (1,3) .. controls +(0,8) and +(0,6) .. (5,5);
        \draw [black,line width=1pt,->,>=latex']
        (3,1) .. controls +(0,10) and +(0,4) .. (7,7);
      \end{scope}
      \begin{scope}[xshift=10cm]
        \draw[step=1cm,black!75,ultra thin,fill=black!50]
        (-0.2,-0.2) grid (8.2,8.2);
        \foreach \x/\y in {1/3,3/1,5/5,7/7} {
          \draw [fill=black] (\x,\y) circle (0.2);
        }
        \draw [black,line width=1pt,<-,>=latex']
        (1,3) .. controls +(0,8) and +(0,6) .. (5,5);
        \draw [black,line width=1pt,<-,>=latex']
        (3,1) .. controls +(0,10) and +(0,4) .. (7,7);
      \end{scope}
      %
      \node [align=center] at (9,-2.5) {$2134$};
    \end{tikzpicture} \\
    \begin{tikzpicture}[scale=0.165,label/.style={anchor=base}]
      \begin{scope}[]
        \draw[step=1cm,black!75,ultra thin,fill=black!50]
        (-0.2,-0.2) grid (8.2,8.2);
        \foreach \x/\y in {1/3,3/5,5/7,7/1} {
          \draw [fill=black] (\x,\y) circle (0.2);
        }
        \draw [black,line width=1pt,->,>=latex']
        (1,3) .. controls +(0,10) and +(0,4) .. (5,7);
        \draw [black,line width=1pt,->,>=latex']
        (3,5) .. controls +(0,6) and +(0,12) .. (7,1);
      \end{scope}
      \begin{scope}[xshift=10cm]
        \draw[step=1cm,black!75,ultra thin,fill=black!50]
        (-0.2,-0.2) grid (8.2,8.2);
        \foreach \x/\y in {1/3,3/5,5/7,7/1} {
          \draw [fill=black] (\x,\y) circle (0.2);
        }
        \draw [black,line width=1pt,<-,>=latex']
        (1,3) .. controls +(0,10) and +(0,4) .. (5,7);
        \draw [black,line width=1pt,<-,>=latex']
        (3,5) .. controls +(0,6) and +(0,12) .. (7,1);
      \end{scope}
      %
      \node [align=center] at (9,-2.5) {$2341$};
    \end{tikzpicture}
    &
        \begin{tikzpicture}[scale=0.165,label/.style={anchor=base}]
      \begin{scope}[]
        \draw[step=1cm,black!75,ultra thin,fill=black!50]
        (-0.2,-0.2) grid (8.2,8.2);
        \foreach \x/\y in {1/3,3/7,5/5,7/1} {
          \draw [fill=black] (\x,\y) circle (0.2);
        }
        \draw [black,line width=1pt,->,>=latex']
        (1,3) .. controls +(0,10) and +(0,4) .. (5,5);
        \draw [black,line width=1pt,->,>=latex']
        (3,7) .. controls +(0,4) and +(0,10) .. (7,1);
      \end{scope}
      \begin{scope}[xshift=10cm]
        \draw[step=1cm,black!75,ultra thin,fill=black!50]
        (-0.2,-0.2) grid (8.2,8.2);
        \foreach \x/\y in {1/3,3/7,5/5,7/1} {
          \draw [fill=black] (\x,\y) circle (0.2);
        }
        \draw [black,line width=1pt,<-,>=latex']
        (1,3) .. controls +(0,10) and +(0,4) .. (5,5);
        \draw [black,line width=1pt,<-,>=latex']
        (3,7) .. controls +(0,4) and +(0,10) .. (7,1);
      \end{scope}
      %
      \node [align=center] at (9,-2.5) {$2431$};
    \end{tikzpicture}
    &
        \begin{tikzpicture}[scale=0.165,label/.style={anchor=base}]
      \begin{scope}[]
        \draw[step=1cm,black!75,ultra thin,fill=black!50]
        (-0.2,-0.2) grid (8.2,8.2);
        \foreach \x/\y in {1/5,3/1,5/3,7/7} {
          \draw [fill=black] (\x,\y) circle (0.2);
        }
        \draw [black,line width=1pt,->,>=latex']
        (1,5) .. controls +(0,4) and +(0,10) .. (5,3);
        \draw [black,line width=1pt,->,>=latex']
        (3,1) .. controls +(0,10) and +(0,4) .. (7,7);
      \end{scope}
      \begin{scope}[xshift=10cm]
        \draw[step=1cm,black!75,ultra thin,fill=black!50]
        (-0.2,-0.2) grid (8.2,8.2);
        \foreach \x/\y in {1/5,3/1,5/3,7/7} {
          \draw [fill=black] (\x,\y) circle (0.2);
        }
        \draw [black,line width=1pt,<-,>=latex']
        (1,5) .. controls +(0,4) and +(0,10) .. (5,3);
        \draw [black,line width=1pt,<-,>=latex']
        (3,1) .. controls +(0,10) and +(0,4) .. (7,7);
      \end{scope}
      %
      \node [align=center] at (9,-2.5) {$3124$};
    \end{tikzpicture}
    &
        \begin{tikzpicture}[scale=0.165,label/.style={anchor=base}]
      \begin{scope}[]
        \draw[step=1cm,black!75,ultra thin,fill=black!50]
        (-0.2,-0.2) grid (8.2,8.2);
        \foreach \x/\y in {1/5,3/3,5/1,7/7} {
          \draw [fill=black] (\x,\y) circle (0.2);
        }
        \draw [black,line width=1pt,->,>=latex']
        (1,5) .. controls +(0,4) and +(0,10) .. (5,1);
        \draw [black,line width=1pt,->,>=latex']
        (3,3) .. controls +(0,10) and +(0,4) .. (7,7);
      \end{scope}
      \begin{scope}[xshift=10cm]
        \draw[step=1cm,black!75,ultra thin,fill=black!50]
        (-0.2,-0.2) grid (8.2,8.2);
        \foreach \x/\y in {1/5,3/3,5/1,7/7} {
          \draw [fill=black] (\x,\y) circle (0.2);
        }
        \draw [black,line width=1pt,<-,>=latex']
        (1,5) .. controls +(0,4) and +(0,10) .. (5,1);
        \draw [black,line width=1pt,<-,>=latex']
        (3,3) .. controls +(0,10) and +(0,4) .. (7,7);
      \end{scope}
      %
      \node [align=center] at (9,-2.5) {$3214$};
    \end{tikzpicture} \\
        \begin{tikzpicture}[scale=0.165,label/.style={anchor=base}]
      \begin{scope}[]
        \draw[step=1cm,black!75,ultra thin,fill=black!50]
        (-0.2,-0.2) grid (8.2,8.2);
        \foreach \x/\y in {1/5,3/7,5/3,7/1} {
          \draw [fill=black] (\x,\y) circle (0.2);
        }
        \draw [black,line width=1pt,->,>=latex']
        (1,5) .. controls +(0,6) and +(0,8) .. (5,3);
        \draw [black,line width=1pt,->,>=latex']
        (3,7) .. controls +(0,4) and +(0,10) .. (7,1);
      \end{scope}
      \begin{scope}[xshift=10cm]
        \draw[step=1cm,black!75,ultra thin,fill=black!50]
        (-0.2,-0.2) grid (8.2,8.2);
        \foreach \x/\y in {1/5,3/7,5/3,7/1} {
          \draw [fill=black] (\x,\y) circle (0.2);
        }
        \draw [black,line width=1pt,<-,>=latex']
        (1,5) .. controls +(0,6) and +(0,8) .. (5,3);
        \draw [black,line width=1pt,<-,>=latex']
        (3,7) .. controls +(0,4) and +(0,10) .. (7,1);
      \end{scope}
      %
      \node [align=center] at (9,-2.5) {$3421$};
    \end{tikzpicture}
    &
        \begin{tikzpicture}[scale=0.165,label/.style={anchor=base}]
      \begin{scope}[]
        \draw[step=1cm,black!75,ultra thin,fill=black!50]
        (-0.2,-0.2) grid (8.2,8.2);
        \foreach \x/\y in {1/7,3/1,5/3,7/5} {
          \draw [fill=black] (\x,\y) circle (0.2);
        }
        \draw [black,line width=1pt,->,>=latex']
        (1,7) .. controls +(0,4) and +(0,10) .. (5,3);
        \draw [black,line width=1pt,->,>=latex']
        (3,1) .. controls +(0,10) and +(0,4) .. (7,5);
      \end{scope}
      \begin{scope}[xshift=10cm]
        \draw[step=1cm,black!75,ultra thin,fill=black!50]
        (-0.2,-0.2) grid (8.2,8.2);
        \foreach \x/\y in {1/7,3/1,5/3,7/5} {
          \draw [fill=black] (\x,\y) circle (0.2);
        }
        \draw [black,line width=1pt,<-,>=latex']
        (1,7) .. controls +(0,4) and +(0,10) .. (5,3);
        \draw [black,line width=1pt,<-,>=latex']
        (3,1) .. controls +(0,10) and +(0,4) .. (7,5);
      \end{scope}
      %
      \node [align=center] at (9,-2.5) {$4123$};
    \end{tikzpicture}
    &
        \begin{tikzpicture}[scale=0.165,label/.style={anchor=base}]
      \begin{scope}[]
        \draw[step=1cm,black!75,ultra thin,fill=black!50]
        (-0.2,-0.2) grid (8.2,8.2);
        \foreach \x/\y in {1/7,3/3,5/1,7/5} {
          \draw [fill=black] (\x,\y) circle (0.2);
        }
        \draw [black,line width=1pt,->,>=latex']
        (1,7) .. controls +(0,4) and +(0,10) .. (5,1);
        \draw [black,line width=1pt,->,>=latex']
        (3,3) .. controls +(0,10) and +(0,4) .. (7,5);
      \end{scope}
      \begin{scope}[xshift=10cm]
        \draw[step=1cm,black!75,ultra thin,fill=black!50]
        (-0.2,-0.2) grid (8.2,8.2);
        \foreach \x/\y in {1/7,3/3,5/1,7/5} {
          \draw [fill=black] (\x,\y) circle (0.2);
        }
        \draw [black,line width=1pt,<-,>=latex']
        (1,7) .. controls +(0,4) and +(0,10) .. (5,1);
        \draw [black,line width=1pt,<-,>=latex']
        (3,3) .. controls +(0,10) and +(0,4) .. (7,5);
      \end{scope}
      %
      \node [align=center] at (9,-2.5) {$4213$};
    \end{tikzpicture}
    &
        \begin{tikzpicture}[scale=0.165,label/.style={anchor=base}]
      \begin{scope}[]
        \draw[step=1cm,black!75,ultra thin,fill=black!50]
        (-0.2,-0.2) grid (8.2,8.2);
        \foreach \x/\y in {1/7,3/5,5/1,7/3} {
          \draw [fill=black] (\x,\y) circle (0.2);
        }
        \draw [black,line width=1pt,->,>=latex']
        (1,7) .. controls +(0,4) and +(0,10) .. (5,1);
        \draw [black,line width=1pt,->,>=latex']
        (3,5) .. controls +(0,6) and +(0,8) .. (7,3);
      \end{scope}
      \begin{scope}[xshift=10cm]
        \draw[step=1cm,black!75,ultra thin,fill=black!50]
        (-0.2,-0.2) grid (8.2,8.2);
        \foreach \x/\y in {1/7,3/5,5/1,7/3} {
          \draw [fill=black] (\x,\y) circle (0.2);
        }
        \draw [black,line width=1pt,<-,>=latex']
        (1,7) .. controls +(0,4) and +(0,10) .. (5,1);
        \draw [black,line width=1pt,<-,>=latex']
        (3,5) .. controls +(0,6) and +(0,8) .. (7,3);
      \end{scope}
      %
      \node [align=center] at (9,-2.5) {$4312$};
    \end{tikzpicture}
  \end{tabular}
  \caption{\label{fig:forbidden patterns - crossing}%
  The labeled patterns with crossing edges avoided by any directed
  perfect matching on a permutation
  satisfying Property~$\mathbf{P_2}$.}
\end{figure}
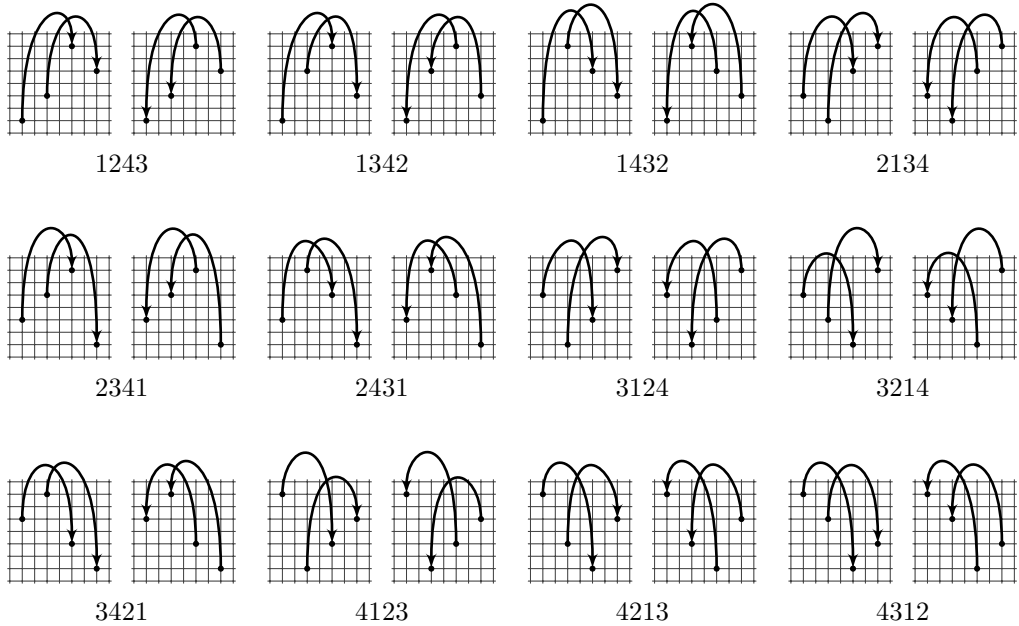


\begin{Lemma}[Forbidden crossing patterns]
  \label{lemma:Forbidden patterns crossing}
  Let $\pi$ be a permutation and $\DMATCHING$ be a directed perfect
  matching on $\pi$ satisfying Property~$\mathbf{P_2}$.
  Then $\DMATCHING$ avoids the following labeled patterns
  \begin{equation}
    \label{equ:Forbidden crossing patterns}
    \begin{array}{cccccccccc}
      \LabeledCrossingRR{0.1}{1}{2}{4}{3}\;,\;\;& \LabeledCrossingLL{0.1}{1}{2}{4}{3}\;,\;\;&
      \LabeledCrossingRR{0.1}{1}{3}{4}{2}\;,\;\;& \LabeledCrossingLL{0.1}{1}{3}{4}{2}\;,\;\;&
      \LabeledCrossingRR{0.1}{1}{4}{3}{2}\;,\;\;& \LabeledCrossingLL{0.1}{1}{4}{3}{2}\;,\;\;&
      \LabeledCrossingRR{0.1}{2}{1}{3}{4}\;,\;\;& \LabeledCrossingLL{0.1}{2}{1}{3}{4}\;,\;\; \\
      \LabeledCrossingRR{0.1}{2}{3}{4}{1}\;,\;\;& \LabeledCrossingLL{0.1}{2}{3}{4}{1}\;,\;\;&
      \LabeledCrossingRR{0.1}{2}{4}{3}{1}\;,\;\;& \LabeledCrossingLL{0.1}{2}{4}{3}{1}\;,\;\;&
      \LabeledCrossingRR{0.1}{3}{1}{2}{4}\;,\;\;& \LabeledCrossingLL{0.1}{3}{1}{2}{4}\;,\;\;&
      \LabeledCrossingRR{0.1}{3}{2}{1}{4}\;,\;\;& \LabeledCrossingLL{0.1}{3}{2}{1}{4}\;,\;\; \\
      \LabeledCrossingRR{0.1}{3}{4}{2}{1}\;,\;\;& \LabeledCrossingLL{0.1}{3}{4}{2}{1}\;,\;\;&
      \LabeledCrossingRR{0.1}{4}{1}{2}{3}\;,\;\;& \LabeledCrossingLL{0.1}{4}{1}{2}{3}\;,\;\;&
      \LabeledCrossingRR{0.1}{4}{2}{1}{3}\;,\;\;& \LabeledCrossingLL{0.1}{4}{2}{1}{3}\;,\;\;&
      \LabeledCrossingRR{0.1}{4}{3}{1}{2}\;,\;\;& \LabeledCrossingLL{0.1}{4}{3}{1}{2}\;;
    \end{array}
  \end{equation}
  see Figure~\ref{fig:forbidden patterns - crossing}.
\end{Lemma}

\begin{figure}[htbp!]
  \centering
  \begin{tabular}{cccc}
    \begin{tikzpicture}[scale=0.165,label/.style={anchor=base}]
      \begin{scope}[]
        \draw[step=1cm,black!75,ultra thin,fill=black!50]
        (-0.2,-0.2) grid (8.2,8.2);
        \foreach \x/\y in {1/1,3/7,5/3,7/5} {
          \draw [fill=black] (\x,\y) circle (0.2);
        }
        \draw [black,line width=1pt,->,>=latex']
        (1,1) .. controls +(0,10) and +(0,4) .. (3,7);
        \draw [black,line width=1pt,->,>=latex']
        (5,3) .. controls +(0,6) and +(0,4) .. (7,5);
      \end{scope}
      \begin{scope}[yshift=12cm]
        \draw[step=1cm,black!75,ultra thin,fill=black!50]
        (-0.2,-0.2) grid (8.2,8.2);
        \foreach \x/\y in {1/1,3/7,5/3,7/5} {
          \draw [fill=black] (\x,\y) circle (0.2);
        }
        \draw [black,line width=1pt,->,>=latex']
        (1,1) .. controls +(0,10) and +(0,4) .. (3,7);
        \draw [black,line width=1pt,<-,>=latex']
        (5,3) .. controls +(0,6) and +(0,4) .. (7,5);
      \end{scope}
      \begin{scope}[xshift=10cm]
        \draw[step=1cm,black!75,ultra thin,fill=black!50]
        (-0.2,-0.2) grid (8.2,8.2);
        \foreach \x/\y in {1/1,3/7,5/3,7/5} {
          \draw [fill=black] (\x,\y) circle (0.2);
        }
        \draw [black,line width=1pt,<-,>=latex']
        (1,1) .. controls +(0,10) and +(0,4) .. (3,7);
        \draw [black,line width=1pt,->,>=latex']
        (5,3) .. controls +(0,6) and +(0,4) .. (7,5);
      \end{scope}
      \begin{scope}[xshift=10cm,yshift=12cm]
        \draw[step=1cm,black!75,ultra thin,fill=black!50]
        (-0.2,-0.2) grid (8.2,8.2);
        \foreach \x/\y in {1/1,3/7,5/3,7/5} {
          \draw [fill=black] (\x,\y) circle (0.2);
        }
        \draw [black,line width=1pt,<-,>=latex']
        (1,1) .. controls +(0,10) and +(0,4) .. (3,7);
        \draw [black,line width=1pt,<-,>=latex']
        (5,3) .. controls +(0,6) and +(0,4) .. (7,5);
      \end{scope}
      %
      \node [align=center] at (9,-4.5) {$1423$};
    \end{tikzpicture}
    &
    \begin{tikzpicture}[scale=0.165,label/.style={anchor=base}]
      \begin{scope}[]
        \draw[step=1cm,black!75,ultra thin,fill=black!50]
        (-0.2,-0.2) grid (8.2,8.2);
        \foreach \x/\y in {1/1,3/7,5/5,7/3} {
          \draw [fill=black] (\x,\y) circle (0.2);
        }
        \draw [black,line width=1pt,->,>=latex']
        (1,1) .. controls +(0,12) and +(0,4) .. (3,7);
        \draw [black,line width=1pt,->,>=latex']
        (5,5) .. controls +(0,4) and +(0,6) .. (7,3);
      \end{scope}
      \begin{scope}[yshift=12cm]
        \draw[step=1cm,black!75,ultra thin,fill=black!50]
        (-0.2,-0.2) grid (8.2,8.2);
        \foreach \x/\y in {1/1,3/7,5/5,7/3} {
          \draw [fill=black] (\x,\y) circle (0.2);
        }
        \draw [black,line width=1pt,->,>=latex']
        (1,1) .. controls +(0,12) and +(0,4) .. (3,7);
        \draw [black,line width=1pt,<-,>=latex']
        (5,5) .. controls +(0,4) and +(0,6) .. (7,3);
      \end{scope}
      \begin{scope}[xshift=10cm]
        \draw[step=1cm,black!75,ultra thin,fill=black!50]
        (-0.2,-0.2) grid (8.2,8.2);
        \foreach \x/\y in {1/1,3/7,5/5,7/7} {
          \draw [fill=black] (\x,\y) circle (0.2);
        }
        \draw [black,line width=1pt,<-,>=latex']
        (1,1) .. controls +(0,12) and +(0,4) .. (3,7);
        \draw [black,line width=1pt,->,>=latex']
        (5,5) .. controls +(0,4) and +(0,6) .. (7,3);
      \end{scope}
      \begin{scope}[xshift=10cm,yshift=12cm]
        \draw[step=1cm,black!75,ultra thin,fill=black!50]
        (-0.2,-0.2) grid (8.2,8.2);
        \foreach \x/\y in {1/1,3/7,5/5,7/3} {
          \draw [fill=black] (\x,\y) circle (0.2);
        }
        \draw [black,line width=1pt,<-,>=latex']
        (1,1) .. controls +(0,12) and +(0,4) .. (3,7);
        \draw [black,line width=1pt,<-,>=latex']
        (5,5) .. controls +(0,4) and +(0,6) .. (7,3);
      \end{scope}
      %
      \node [align=center] at (9,-4.5) {$1432$};
    \end{tikzpicture}
    &
        \begin{tikzpicture}[scale=0.165,label/.style={anchor=base}]
      \begin{scope}[]
        \draw[step=1cm,black!75,ultra thin,fill=black!50]
        (-0.2,-0.2) grid (8.2,8.2);
        \foreach \x/\y in {1/3,3/5,5/1,7/7} {
          \draw [fill=black] (\x,\y) circle (0.2);
        }
        \draw [black,line width=1pt,->,>=latex']
        (1,3) .. controls +(0,6) and +(0,4) .. (3,5);
        \draw [black,line width=1pt,->,>=latex']
        (5,1) .. controls +(0,12) and +(0,4) .. (7,7);
      \end{scope}
      \begin{scope}[yshift=12cm]
        \draw[step=1cm,black!75,ultra thin,fill=black!50]
        (-0.2,-0.2) grid (8.2,8.2);
        \foreach \x/\y in {1/3,3/5,5/1,7/7} {
          \draw [fill=black] (\x,\y) circle (0.2);
        }
        \draw [black,line width=1pt,->,>=latex']
        (1,3) .. controls +(0,6) and +(0,4) .. (3,5);
        \draw [black,line width=1pt,<-,>=latex']
        (5,1) .. controls +(0,12) and +(0,4) .. (7,7);
      \end{scope}
      \begin{scope}[xshift=10cm]
        \draw[step=1cm,black!75,ultra thin,fill=black!50]
        (-0.2,-0.2) grid (8.2,8.2);
        \foreach \x/\y in {1/3,3/5,5/1,7/7} {
          \draw [fill=black] (\x,\y) circle (0.2);
        }
        \draw [black,line width=1pt,<-,>=latex']
        (1,3) .. controls +(0,6) and +(0,4) .. (3,5);
        \draw [black,line width=1pt,->,>=latex']
        (5,1) .. controls +(0,12) and +(0,4) .. (7,7);
      \end{scope}
      \begin{scope}[xshift=10cm,yshift=12cm]
        \draw[step=1cm,black!75,ultra thin,fill=black!50]
        (-0.2,-0.2) grid (8.2,8.2);
        \foreach \x/\y in {1/3,3/5,5/1,7/7} {
          \draw [fill=black] (\x,\y) circle (0.2);
        }
        \draw [black,line width=1pt,<-,>=latex']
        (1,3) .. controls +(0,6) and +(0,4) .. (3,5);
        \draw [black,line width=1pt,<-,>=latex']
        (5,1) .. controls +(0,12) and +(0,4) .. (7,7);
      \end{scope}
      %
      \node [align=center] at (9,-4.5) {$2314$};
    \end{tikzpicture}
    &
        \begin{tikzpicture}[scale=0.165,label/.style={anchor=base}]
      \begin{scope}[]
        \draw[step=1cm,black!75,ultra thin,fill=black!50]
        (-0.2,-0.2) grid (8.2,8.2);
        \foreach \x/\y in {1/3,3/5,5/7,7/1} {
          \draw [fill=black] (\x,\y) circle (0.2);
        }
        \draw [black,line width=1pt,->,>=latex']
        (1,3) .. controls +(0,6) and +(0,4) .. (3,5);
        \draw [black,line width=1pt,->,>=latex']
        (5,7) .. controls +(0,4) and +(0,12) .. (7,1);
      \end{scope}
      \begin{scope}[yshift=12cm]
        \draw[step=1cm,black!75,ultra thin,fill=black!50]
        (-0.2,-0.2) grid (8.2,8.2);
        \foreach \x/\y in {1/3,3/5,5/7,7/1} {
          \draw [fill=black] (\x,\y) circle (0.2);
        }
        \draw [black,line width=1pt,->,>=latex']
        (1,3) .. controls +(0,6) and +(0,4) .. (3,5);
        \draw [black,line width=1pt,<-,>=latex']
        (5,7) .. controls +(0,4) and +(0,12) .. (7,1);
      \end{scope}
      \begin{scope}[xshift=10cm]
        \draw[step=1cm,black!75,ultra thin,fill=black!50]
        (-0.2,-0.2) grid (8.2,8.2);
        \foreach \x/\y in {1/3,3/5,5/7,7/1} {
          \draw [fill=black] (\x,\y) circle (0.2);
        }
        \draw [black,line width=1pt,<-,>=latex']
        (1,3) .. controls +(0,6) and +(0,4) .. (3,5);
        \draw [black,line width=1pt,->,>=latex']
        (5,7) .. controls +(0,4) and +(0,12) .. (7,1);
      \end{scope}
      \begin{scope}[xshift=10cm,yshift=12cm]
        \draw[step=1cm,black!75,ultra thin,fill=black!50]
        (-0.2,-0.2) grid (8.2,8.2);
        \foreach \x/\y in {1/3,3/5,5/7,7/1} {
          \draw [fill=black] (\x,\y) circle (0.2);
        }
        \draw [black,line width=1pt,<-,>=latex']
        (1,3) .. controls +(0,6) and +(0,4) .. (3,5);
        \draw [black,line width=1pt,<-,>=latex']
        (5,7) .. controls +(0,4) and +(0,12) .. (7,1);
      \end{scope}
      %
      \node [align=center] at (9,-4.5) {$2341$};
    \end{tikzpicture}
    \\
    \begin{tikzpicture}[scale=0.165,label/.style={anchor=base}]
      \begin{scope}[]
        \draw[step=1cm,black!75,ultra thin,fill=black!50]
        (-0.2,-0.2) grid (8.2,8.2);
        \foreach \x/\y in {1/5,3/5,5/1,7/7} {
          \draw [fill=black] (\x,\y) circle (0.2);
        }
        \draw [black,line width=1pt,->,>=latex']
        (1,5) .. controls +(0,4) and +(0,6) .. (3,3);
        \draw [black,line width=1pt,->,>=latex']
        (5,1) .. controls +(0,12) and +(0,4) .. (7,7);
      \end{scope}
      \begin{scope}[yshift=12cm]
        \draw[step=1cm,black!75,ultra thin,fill=black!50]
        (-0.2,-0.2) grid (8.2,8.2);
        \foreach \x/\y in {1/5,3/3,5/1,7/7} {
          \draw [fill=black] (\x,\y) circle (0.2);
        }
        \draw [black,line width=1pt,->,>=latex']
        (1,5) .. controls +(0,4) and +(0,6) .. (3,3);
        \draw [black,line width=1pt,<-,>=latex']
        (5,1) .. controls +(0,12) and +(0,4) .. (7,7);
      \end{scope}
      \begin{scope}[xshift=10cm]
        \draw[step=1cm,black!75,ultra thin,fill=black!50]
        (-0.2,-0.2) grid (8.2,8.2);
        \foreach \x/\y in {1/5,3/3,5/1,7/7} {
          \draw [fill=black] (\x,\y) circle (0.2);
        }
        \draw [black,line width=1pt,<-,>=latex']
        (1,5) .. controls +(0,4) and +(0,6) .. (3,3);
        \draw [black,line width=1pt,->,>=latex']
        (5,1) .. controls +(0,12) and +(0,4) .. (7,7);
      \end{scope}
      \begin{scope}[xshift=10cm,yshift=12cm]
        \draw[step=1cm,black!75,ultra thin,fill=black!50]
        (-0.2,-0.2) grid (8.2,8.2);
        \foreach \x/\y in {1/5,3/3,5/1,7/7} {
          \draw [fill=black] (\x,\y) circle (0.2);
        }
        \draw [black,line width=1pt,<-,>=latex']
        (1,5) .. controls +(0,4) and +(0,6) .. (3,3);
        \draw [black,line width=1pt,<-,>=latex']
        (5,1) .. controls +(0,12) and +(0,4) .. (7,7);
      \end{scope}
      %
      \node [align=center] at (9,-4.5) {$3214$};
    \end{tikzpicture}
    &
    \begin{tikzpicture}[scale=0.165,label/.style={anchor=base}]
      \begin{scope}[]
        \draw[step=1cm,black!75,ultra thin,fill=black!50]
        (-0.2,-0.2) grid (8.2,8.2);
        \foreach \x/\y in {1/5,3/3,5/7,7/1} {
          \draw [fill=black] (\x,\y) circle (0.2);
        }
        \draw [black,line width=1pt,->,>=latex']
        (1,5) .. controls +(0,4) and +(0,6) .. (3,3);
        \draw [black,line width=1pt,->,>=latex']
        (5,7) .. controls +(0,4) and +(0,12) .. (7,1);
      \end{scope}
      \begin{scope}[yshift=12cm]
        \draw[step=1cm,black!75,ultra thin,fill=black!50]
        (-0.2,-0.2) grid (8.2,8.2);
        \foreach \x/\y in {1/5,3/3,5/7,7/1} {
          \draw [fill=black] (\x,\y) circle (0.2);
        }
        \draw [black,line width=1pt,->,>=latex']
        (1,5) .. controls +(0,4) and +(0,6) .. (3,3);
        \draw [black,line width=1pt,<-,>=latex']
        (5,7) .. controls +(0,4) and +(0,12) .. (7,1);
      \end{scope}
      \begin{scope}[xshift=10cm]
        \draw[step=1cm,black!75,ultra thin,fill=black!50]
        (-0.2,-0.2) grid (8.2,8.2);
        \foreach \x/\y in {1/5,3/3,5/7,7/1} {
          \draw [fill=black] (\x,\y) circle (0.2);
        }
        \draw [black,line width=1pt,<-,>=latex']
        (1,5) .. controls +(0,4) and +(0,6) .. (3,3);
        \draw [black,line width=1pt,->,>=latex']
        (5,7) .. controls +(0,4) and +(0,12) .. (7,1);
      \end{scope}
      \begin{scope}[xshift=10cm,yshift=12cm]
        \draw[step=1cm,black!75,ultra thin,fill=black!50]
        (-0.2,-0.2) grid (8.2,8.2);
        \foreach \x/\y in {1/5,3/3,5/7,7/1} {
          \draw [fill=black] (\x,\y) circle (0.2);
        }
        \draw [black,line width=1pt,<-,>=latex']
        (1,5) .. controls +(0,4) and +(0,6) .. (3,3);
        \draw [black,line width=1pt,<-,>=latex']
        (5,7) .. controls +(0,4) and +(0,12) .. (7,1);
      \end{scope}
      %
      \node [align=center] at (9,-4.5) {$3241$};
    \end{tikzpicture}
    &
    \begin{tikzpicture}[scale=0.165,label/.style={anchor=base}]
      \begin{scope}[]
        \draw[step=1cm,black!75,ultra thin,fill=black!50]
        (-0.2,-0.2) grid (8.2,8.2);
        \foreach \x/\y in {1/7,3/1,5/3,7/5} {
          \draw [fill=black] (\x,\y) circle (0.2);
        }
        \draw [black,line width=1pt,->,>=latex']
        (1,7) .. controls +(0,4) and +(0,12) .. (3,1);
        \draw [black,line width=1pt,->,>=latex']
        (5,3) .. controls +(0,6) and +(0,4) .. (7,5);
      \end{scope}
      \begin{scope}[yshift=12cm]
        \draw[step=1cm,black!75,ultra thin,fill=black!50]
        (-0.2,-0.2) grid (8.2,8.2);
        \foreach \x/\y in {1/7,3/1,5/3,7/5} {
          \draw [fill=black] (\x,\y) circle (0.2);
        }
        \draw [black,line width=1pt,->,>=latex']
        (1,7) .. controls +(0,4) and +(0,12) .. (3,1);
        \draw [black,line width=1pt,<-,>=latex']
        (5,3) .. controls +(0,6) and +(0,4) .. (7,5);
      \end{scope}
      \begin{scope}[xshift=10cm]
        \draw[step=1cm,black!75,ultra thin,fill=black!50]
        (-0.2,-0.2) grid (8.2,8.2);
        \foreach \x/\y in {1/7,3/1,5/3,7/5} {
          \draw [fill=black] (\x,\y) circle (0.2);
        }
        \draw [black,line width=1pt,<-,>=latex']
        (1,7) .. controls +(0,4) and +(0,12) .. (3,1);
        \draw [black,line width=1pt,->,>=latex']
        (5,3) .. controls +(0,6) and +(0,4) .. (7,5);
      \end{scope}
      \begin{scope}[xshift=10cm,yshift=12cm]
        \draw[step=1cm,black!75,ultra thin,fill=black!50]
        (-0.2,-0.2) grid (8.2,8.2);
        \foreach \x/\y in {1/7,3/1,5/3,7/5} {
          \draw [fill=black] (\x,\y) circle (0.2);
        }
        \draw [black,line width=1pt,<-,>=latex']
        (1,7) .. controls +(0,4) and +(0,12) .. (3,1);
        \draw [black,line width=1pt,<-,>=latex']
        (5,3) .. controls +(0,6) and +(0,4) .. (7,5);
      \end{scope}
      %
      \node [align=center] at (9,-4.5) {$4123$};
    \end{tikzpicture}
    &
    \begin{tikzpicture}[scale=0.165,label/.style={anchor=base}]
      \begin{scope}[]
        \draw[step=1cm,black!75,ultra thin,fill=black!50]
        (-0.2,-0.2) grid (8.2,8.2);
        \foreach \x/\y in {1/7,3/1,5/5,7/3} {
          \draw [fill=black] (\x,\y) circle (0.2);
        }
        \draw [black,line width=1pt,->,>=latex']
        (1,7) .. controls +(0,4) and +(0,12) .. (3,1);
        \draw [black,line width=1pt,->,>=latex']
        (5,5) .. controls +(0,4) and +(0,6) .. (7,3);
      \end{scope}
      \begin{scope}[yshift=12cm]
        \draw[step=1cm,black!75,ultra thin,fill=black!50]
        (-0.2,-0.2) grid (8.2,8.2);
        \foreach \x/\y in {1/7,3/1,5/5,7/3} {
          \draw [fill=black] (\x,\y) circle (0.2);
        }
        \draw [black,line width=1pt,->,>=latex']
        (1,7) .. controls +(0,4) and +(0,12) .. (3,1);
        \draw [black,line width=1pt,<-,>=latex']
        (5,5) .. controls +(0,4) and +(0,6) .. (7,3);
      \end{scope}
      \begin{scope}[xshift=10cm]
        \draw[step=1cm,black!75,ultra thin,fill=black!50]
        (-0.2,-0.2) grid (8.2,8.2);
        \foreach \x/\y in {1/7,3/1,5/5,7/3} {
          \draw [fill=black] (\x,\y) circle (0.2);
        }
        \draw [black,line width=1pt,<-,>=latex']
        (1,7) .. controls +(0,4) and +(0,12) .. (3,1);
        \draw [black,line width=1pt,->,>=latex']
        (5,5) .. controls +(0,4) and +(0,6) .. (7,3);
      \end{scope}
      \begin{scope}[xshift=10cm,yshift=12cm]
        \draw[step=1cm,black!75,ultra thin,fill=black!50]
        (-0.2,-0.2) grid (8.2,8.2);
        \foreach \x/\y in {1/7,3/1,5/5,7/3} {
          \draw [fill=black] (\x,\y) circle (0.2);
        }
        \draw [black,line width=1pt,<-,>=latex']
        (1,7) .. controls +(0,4) and +(0,12) .. (3,1);
        \draw [black,line width=1pt,<-,>=latex']
        (5,5) .. controls +(0,4) and +(0,6) .. (7,3);
      \end{scope}
      %
      \node [align=center] at (9,-4.5) {$4132$};
    \end{tikzpicture}
  \end{tabular}
  \caption{\label{fig:forbidden patterns - precedence}%
  The labeled patterns with consecutive edges avoided by any directed
  perfect matching on a permutation
  satisfying Property~$\mathbf{P_2}$.}
\end{figure}

\begin{Lemma}[Forbidden precedence patterns]
  \label{lemma:Forbidden patterns precedence}
  Let $\pi$ be a permutation and $\DMATCHING$ be a directed perfect
  matching on $\pi$ satisfying Property~$\mathbf{P_2}$.
  Then $\DMATCHING$ avoids the following labeled patterns
  \begin{equation}
    \label{equ:Forbidden crossing patterns}
    \begin{array}{cccccccccc}
      \LabeledPrecedenceRR{0.1}{1}{4}{2}{3}\;,\;\;& \LabeledPrecedenceRL{0.1}{1}{4}{2}{3}\;,\;\;&
      \LabeledPrecedenceLR{0.1}{1}{4}{2}{3}\;,\;\;& \LabeledPrecedenceLL{0.1}{1}{4}{2}{3}\;,\;\;&
      \LabeledPrecedenceRR{0.1}{1}{4}{3}{2}\;,\;\;& \LabeledPrecedenceRL{0.1}{1}{4}{3}{2}\;,\;\;&
      \LabeledPrecedenceLR{0.1}{1}{4}{3}{2}\;,\;\;& \LabeledPrecedenceLL{0.1}{1}{4}{3}{2}\;,\;\; \\
      \LabeledPrecedenceRR{0.1}{2}{3}{1}{4}\;,\;\;& \LabeledPrecedenceRL{0.1}{2}{3}{1}{4}\;,\;\;&
      \LabeledPrecedenceLR{0.1}{2}{3}{1}{4}\;,\;\;& \LabeledPrecedenceLL{0.1}{2}{3}{1}{4}\;,\;\;&
      \LabeledPrecedenceRR{0.1}{2}{3}{4}{1}\;,\;\;& \LabeledPrecedenceRL{0.1}{2}{3}{4}{1}\;,\;\;&
      \LabeledPrecedenceLR{0.1}{2}{3}{4}{1}\;,\;\;& \LabeledPrecedenceLL{0.1}{2}{3}{4}{1}\;,\;\; \\
      \LabeledPrecedenceRR{0.1}{3}{2}{1}{4}\;,\;\;& \LabeledPrecedenceRL{0.1}{3}{2}{1}{4}\;,\;\;&
      \LabeledPrecedenceLR{0.1}{3}{2}{1}{4}\;,\;\;& \LabeledPrecedenceLL{0.1}{3}{2}{1}{4}\;,\;\;&
      \LabeledPrecedenceRR{0.1}{3}{2}{4}{1}\;,\;\;& \LabeledPrecedenceRL{0.1}{3}{2}{4}{1}\;,\;\;&
      \LabeledPrecedenceLR{0.1}{3}{2}{4}{1}\;,\;\;& \LabeledPrecedenceLL{0.1}{3}{2}{4}{1}\;,\;\; \\
      \LabeledPrecedenceRR{0.1}{4}{1}{2}{3}\;,\;\;& \LabeledPrecedenceRL{0.1}{4}{1}{2}{3}\;,\;\;&
      \LabeledPrecedenceLR{0.1}{4}{1}{2}{3}\;,\;\;& \LabeledPrecedenceLL{0.1}{4}{1}{2}{3}\;,\;\;&
      \LabeledPrecedenceRR{0.1}{4}{1}{3}{2}\;,\;\;& \LabeledPrecedenceRL{0.1}{4}{1}{3}{2}\;,\;\;&
      \LabeledPrecedenceLR{0.1}{4}{1}{3}{2}\;,\;\;& \LabeledPrecedenceLL{0.1}{4}{1}{3}{2}\;;
    \end{array}
  \end{equation}
  see Figure~\ref{fig:forbidden patterns - precedence}.
\end{Lemma}

  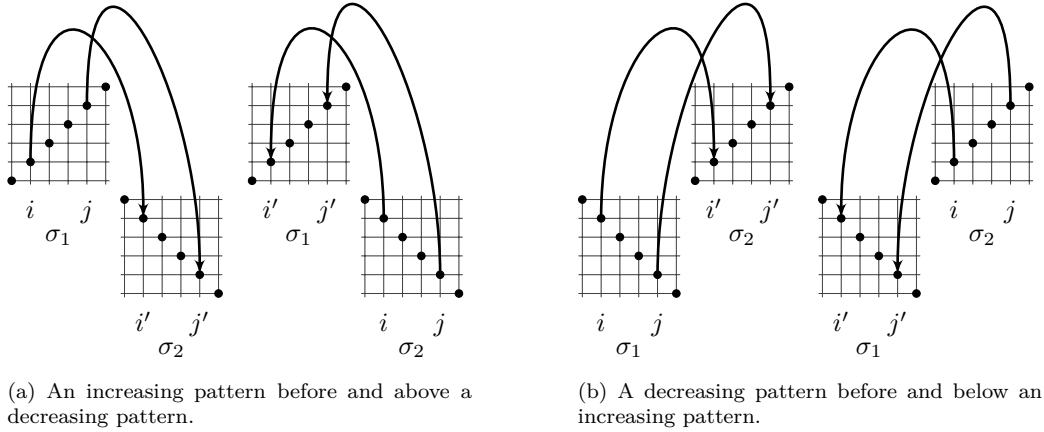
\begin{figure}[htbp!]
    \centering
    \subfigure[%
      An increasing pattern before and above a decreasing pattern.]{%
      \begin{tikzpicture}[scale=0.25,label/.style={anchor=base}]
        \draw[step=1cm,black!75,ultra thin,fill=black!50]
        (-0.2,5.8) grid (5.2,11.2);
        \foreach \x/\y in {0/6,1/7,2/8,3/9,4/10,5/11} {
          \draw [fill=black] (\x,\y) circle (0.2);
        }
        \draw[step=1cm,black!75,ultra thin,fill=black!50]
        (5.8,-0.2) grid (11.2,5.2);
        \foreach \x/\y in {6/5,7/4,8/3,9/2,10/1,11/0} {
          \draw [fill=black] (\x,\y) circle (0.2);
        }
        \node [label] (a) at (1,4) {$i$};
        \node [label] (b) at (4,4) {$j$};
        \node [label] (ap) at (7,-2) {$i'$};
        \node [label] (bp) at (10,-2) {$j'$};
        \node (nu1) at (2.5,3) {$\sigma_1$};
        \node (nu2) at (8.5,-3) {$\sigma_2$};
        \draw [black,line width=1pt,->,>=latex']
        (1,7) .. controls +(0,12) and +(0,10) .. (7,4);
        \draw [black,line width=1pt,->,>=latex']
        (4,10) .. controls +(0,12) and +(0,10) .. (10,1);
      \end{tikzpicture}
      \;
      \begin{tikzpicture}[scale=0.25,label/.style={anchor=base}]
        \draw[step=1cm,black!75,ultra thin,fill=black!50]
        (-0.2,5.8) grid (5.2,11.2);
        \foreach \x/\y in {0/6,1/7,2/8,3/9,4/10,5/11} {
          \draw [fill=black] (\x,\y) circle (0.2);
        }
        \draw[step=1cm,black!75,ultra thin,fill=black!50]
        (5.8,-0.2) grid (11.2,5.2);
        \foreach \x/\y in {6/5,7/4,8/3,9/2,10/1,11/0} {
          \draw [fill=black] (\x,\y) circle (0.2);
        }
        \node [label] (a) at (1,4) {$i'$};
        \node [label] (b) at (4,4) {$j'$};
        \node [label] (ap) at (7,-2) {$i$};
        \node [label] (bp) at (10,-2) {$j$};
        \node (nu1) at (2.5,3) {$\sigma_1$};
        \node (nu2) at (8.5,-3) {$\sigma_2$};
        \draw [black,line width=1pt,<-,>=latex']
        (1,7) .. controls +(0,12) and +(0,10) .. (7,4);
        \draw [black,line width=1pt,<-,>=latex']
        (4,10) .. controls +(0,12) and +(0,10) .. (10,1);
      \end{tikzpicture}
      \label{subfig:increasing before and above decreasing}
    }
    \qquad \qquad
    \subfigure[%
      A decreasing pattern before and below an increasing pattern.]{%
      \begin{tikzpicture}[scale=0.25,label/.style={anchor=base}]
        \draw[step=1cm,black!75,ultra thin,fill=black!50]
        (-0.2,-0.2) grid (5.2,5.2);
        \foreach \x/\y in {0/5,1/4,2/3,3/2,4/1,5/0} {
          \draw [fill=black] (\x,\y) circle (0.2);
        }
        \draw[step=1cm,black!75,ultra thin,fill=black!50]
        (5.8,5.8) grid (11.2,11.2);
        \foreach \x/\y in {6/6,7/7,8/8,9/9,10/10,11/11} {
          \draw [fill=black] (\x,\y) circle (0.2);
        }
        \node [label] (a) at (1,-2) {$i$};
        \node [label] (b) at (4,-2) {$j$};
        \node [label] (ap) at (7,4) {$i'$};
        \node [label] (bp) at (10,4) {$j'$};
        \node (nu1) at (2.5,-3) {$\sigma_1$};
        \node (nu2) at (8.5,3) {$\sigma_2$};
        \draw [black,line width=1pt,->,>=latex']
        (1,4) .. controls +(0,10) and +(0,12) .. (7,7);
        \draw [black,line width=1pt,->,>=latex']
        (4,1) .. controls +(0,10) and +(0,12) .. (10,10);
      \end{tikzpicture}
      \;
      \begin{tikzpicture}[scale=0.25,label/.style={anchor=base}]
        \draw[step=1cm,black!75,ultra thin,fill=black!50]
        (-0.2,-0.2) grid (5.2,5.2);
        \foreach \x/\y in {0/5,1/4,2/3,3/2,4/1,5/0} {
          \draw [fill=black] (\x,\y) circle (0.2);
        }
        \draw[step=1cm,black!75,ultra thin,fill=black!50]
        (5.8,5.8) grid (11.2,11.2);
        \foreach \x/\y in {6/6,7/7,8/8,9/9,10/10,11/11} {
          \draw [fill=black] (\x,\y) circle (0.2);
        }
        \node [label] (a) at (1,-2) {$i'$};
        \node [label] (b) at (4,-2) {$j'$};
        \node [label] (ap) at (7,4) {$i$};
        \node [label] (bp) at (10,4) {$j$};
        \node (nu1) at (2.5,-3) {$\sigma_1$};
        \node (nu2) at (8.5,3) {$\sigma_2$};
        \draw [black,line width=1pt,<-,>=latex']
        (1,4) .. controls +(0,10) and +(0,12) .. (7,7);
        \draw [black,line width=1pt,<-,>=latex']
        (4,1) .. controls +(0,10) and +(0,12) .. (10,10);
      \end{tikzpicture}
      \label{subfig:decreasing before and below increasing}
    }
    \caption{\label{fig:subfig:no (nu_1, nu_2)-arc}%
      Illustration of Corollary~\ref{corollary:at most one arc monotone}.}
  \end{figure}

  A useful corollary of Lemma~\ref{lemma:Forbidden patterns crossing} reads as follows.

  \begin{Corollary}
    \label{corollary:at most one arc monotone}
    Let $\pi = \pi_1 \, \sigma_1 \, \pi_2 \, \sigma_2 \, \pi_3$ be a
    permutation and
    $\DMATCHING$ be a directed perfect matching on $\pi$ satisfying Properties
    $\mathbf{P_1}$ and $\mathbf{P_2}$. The following assertions hold.
    \begin{enumerate}[label={\it (\roman*)},fullwidth]
    \item \label{item:at most one arc monotone 1}
      If $\sigma_1$ is increasing, $\sigma_2$ is decreasing, and
      $\sigma_1$ is above $\sigma_2$ (see
      Figure~\ref{subfig:increasing before and above decreasing}),
      then there is at most one arc between $\sigma_1$ and $\sigma_2$
      in $\DMATCHING$ (this arc can be a
      $(\sigma_1, \sigma_2)$-arc or a $(\sigma_2, \sigma_1)$-arc).
    \item \label{item:at most one arc monotone 2}
      If $\sigma_1$ is decreasing, $\sigma_2$ is increasing, and
      $\sigma_1$ is below $\sigma_2$ (see
      Figure~\ref{subfig:decreasing before and below increasing}),
      then there is at most one arc between $\sigma_1$ and $\sigma_2$
      in $\DMATCHING$ (this arc can be a
      $(\sigma_1, \sigma_2)$-arc or a $(\sigma_2, \sigma_1)$-arc).
    \end{enumerate}
  \end{Corollary}

  \begin{proof}[Proof of Corollary~\ref{corollary:at most one arc monotone}]
    Suppose, aiming at a contradiction, that
    \ref{item:at most one arc monotone 1} does not hold.
    Since $\DMATCHING$ has
    Property~$\mathbf{P_1}$, it avoids the unlabelled patterns of
    $\mathcal{P}_1$.
    Then it follows that $\DMATCHING$ contains
    (see Figure~\ref{subfig:increasing before and above decreasing})
    either
    two crossing $(\sigma_1, \sigma_2)$-arcs (a $\LabeledCrossingRR{.08}{3}{4}{2}{1}$ labeled pattern) or
    two crossing $(\sigma_2, \sigma_1)$-arcs (a $\LabeledCrossingLL{.08}{3}{4}{2}{1}$ labeled pattern).
    Hence, according to Lemma~\ref{lemma:Forbidden patterns crossing},
    $\DMATCHING$ cannot have Property~$\mathbf{P_2}$.
    This is the sought-after contradiction.

    The proof for \ref{item:at most one arc monotone 2} is similar
    (see Figure~\ref{subfig:decreasing before and below increasing})
    replacing the labeled patterns
    $\LabeledCrossingRR{.08}{3}{4}{2}{1}$ and $\LabeledCrossingLL{.08}{3}{4}{2}{1}$
    by
    $\LabeledCrossingRR{.08}{2}{1}{3}{4}$ and $\LabeledCrossingLL{.08}{2}{1}{3}{4}$.
  \end{proof}

  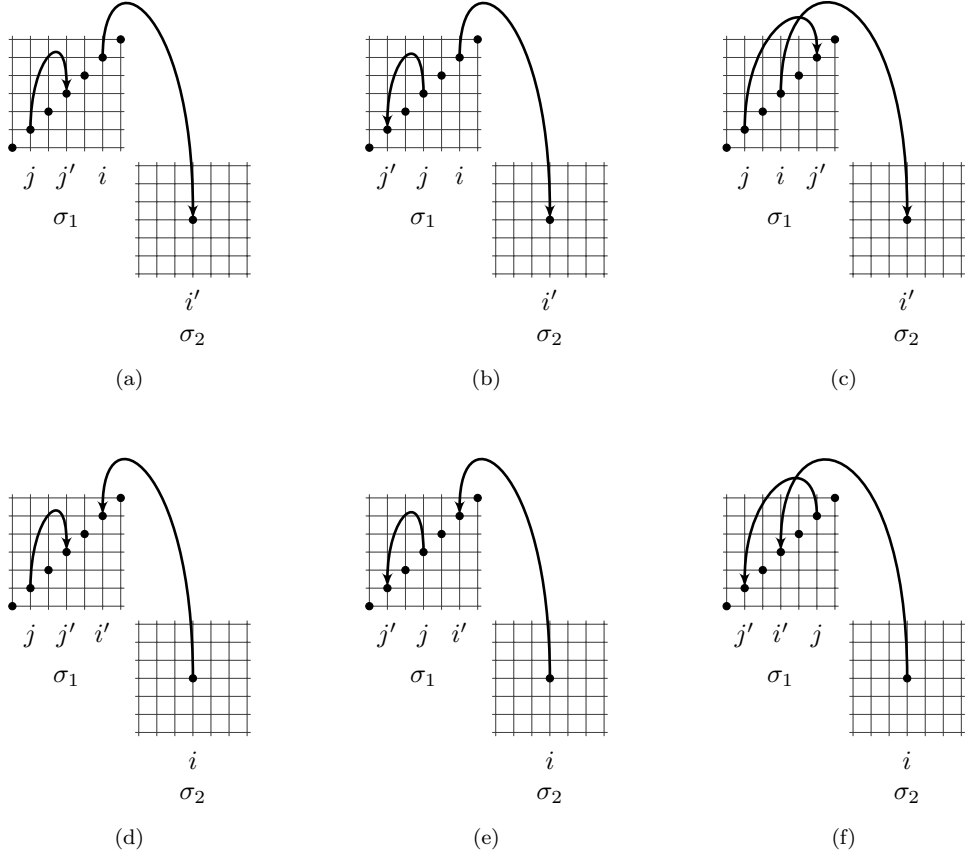
\begin{figure}[ht!]
    \centering
    \subfigure[]{%
      \begin{tikzpicture}[scale=0.24,label/.style={anchor=base}]
        \draw[step=1cm,black!75,ultra thin,fill=black!50]
        (-0.2,6.8) grid (6.2,13.2);
        \foreach \x/\y in {0/7,1/8,2/9,3/10,4/11,5/12,6/13} {
          \draw [fill=black] (\x,\y) circle (0.2);
        }
        \draw[step=1cm,black!75,ultra thin,fill=black!50]
        (6.8,-0.2) grid (13.2,6.2);
        \draw [fill=black] (10,3) circle (0.2);
        \node [label] (a) at (5,5) {$i$};
        \node [label] (ap) at (10,-2) {$i'$};
        \node [label] (b) at (1,5) {$j$};
        \node [label] (bp) at (3,5) {$j'$};
        \node (sigma1) at (3,3) {$\sigma_1$};
        \node (sigma2) at (10,-3.5) {$\sigma_2$};
        \draw [black,line width=1pt,->,>=latex']
        (1,8) .. controls +(0,4) and +(0,4) .. (3,10);
        \draw [black,line width=1pt,->,>=latex']
        (5,12) .. controls +(0,6) and +(0,12) .. (10,3);
      \end{tikzpicture}
      \label{subfig:precedence RR}
    }
    \qquad\qquad
    \subfigure[]{%
      \begin{tikzpicture}[scale=0.24,label/.style={anchor=base}]
        \draw[step=1cm,black!75,ultra thin,fill=black!50]
        (-0.2,6.8) grid (6.2,13.2);
        \foreach \x/\y in {0/7,1/8,2/9,3/10,4/11,5/12,6/13} {
          \draw [fill=black] (\x,\y) circle (0.2);
        }
        \draw[step=1cm,black!75,ultra thin,fill=black!50]
        (6.8,-0.2) grid (13.2,6.2);
        \draw [fill=black] (10,3) circle (0.2);
        \node [label] (a) at (5,5) {$i$};
        \node [label] (ap) at (10,-2) {$i'$};
        \node [label] (bp) at (1,5) {$j'$};
        \node [label] (b) at (3,5) {$j$};
        \node (sigma1) at (3,3) {$\sigma_1$};
        \node (sigma2) at (10,-3.5) {$\sigma_2$};
        \draw [black,line width=1pt,<-,>=latex']
        (1,8) .. controls +(0,4) and +(0,4) .. (3,10);
        \draw [black,line width=1pt,->,>=latex']
        (5,12) .. controls +(0,6) and +(0,12) .. (10,3);
      \end{tikzpicture}
      \label{subfig:precedence LR}
    }
    \qquad\qquad
    \subfigure[]{%
      \begin{tikzpicture}[scale=0.24,label/.style={anchor=base}]
        \draw[step=1cm,black!75,ultra thin,fill=black!50]
        (-0.2,6.8) grid (6.2,13.2);
        \foreach \x/\y in {0/7,1/8,2/9,3/10,4/11,5/12,6/13} {
          \draw [fill=black] (\x,\y) circle (0.2);
        }
        \draw[step=1cm,black!75,ultra thin,fill=black!50]
        (6.8,-0.2) grid (13.2,6.2);
        \draw [fill=black] (10,3) circle (0.2);
        \node [label] (a) at (3,5) {$i$};
        \node [label] (ap) at (10,-2) {$i'$};
        \node [label] (b) at (1,5) {$j$};
        \node [label] (bp) at (5,5) {$j'$};
        \node (sigma1) at (3,3) {$\sigma_1$};
        \node (sigma2) at (10,-3.5) {$\sigma_2$};
        \draw [black,line width=1pt,->,>=latex']
        (1,8) .. controls +(0,6) and +(0,4) .. (5,12);
        \draw [black,line width=1pt,->,>=latex']
        (3,10) .. controls +(0,8) and +(0,14) .. (10,3);
      \end{tikzpicture}
      \label{subfig:crossing RR}
    }
    \qquad\qquad
    \subfigure[]{%
      \begin{tikzpicture}[scale=0.24,label/.style={anchor=base}]
        \draw[step=1cm,black!75,ultra thin,fill=black!50]
        (-0.2,6.8) grid (6.2,13.2);
        \foreach \x/\y in {0/7,1/8,2/9,3/10,4/11,5/12,6/13} {
          \draw [fill=black] (\x,\y) circle (0.2);
        }
        \draw[step=1cm,black!75,ultra thin,fill=black!50]
        (6.8,-0.2) grid (13.2,6.2);
        \draw [fill=black] (10,3) circle (0.2);
        \node [label] (ap) at (5,5) {$i'$};
        \node [label] (a) at (10,-2) {$i$};
        \node [label] (b) at (1,5) {$j$};
        \node [label] (bp) at (3,5) {$j'$};
        \node (sigma1) at (3,3) {$\sigma_1$};
        \node (sigma2) at (10,-3.5) {$\sigma_2$};
        \draw [black,line width=1pt,->,>=latex']
        (1,8) .. controls +(0,4) and +(0,4) .. (3,10);
        \draw [black,line width=1pt,<-,>=latex']
        (5,12) .. controls +(0,6) and +(0,12) .. (10,3);
      \end{tikzpicture}
      \label{subfig:precedence RL}
    }
    \qquad\qquad
    \subfigure[]{%
      \begin{tikzpicture}[scale=0.24,label/.style={anchor=base}]
        \draw[step=1cm,black!75,ultra thin,fill=black!50]
        (-0.2,6.8) grid (6.2,13.2);
        \foreach \x/\y in {0/7,1/8,2/9,3/10,4/11,5/12,6/13} {
          \draw [fill=black] (\x,\y) circle (0.2);
        }
        \draw[step=1cm,black!75,ultra thin,fill=black!50]
        (6.8,-0.2) grid (13.2,6.2);
        \draw [fill=black] (10,3) circle (0.2);
        \node [label] (ap) at (5,5) {$i'$};
        \node [label] (a) at (10,-2) {$i$};
        \node [label] (bp) at (1,5) {$j'$};
        \node [label] (b) at (3,5) {$j$};
        \node (sigma1) at (3,3) {$\sigma_1$};
        \node (sigma2) at (10,-3.5) {$\sigma_2$};
        \draw [black,line width=1pt,<-,>=latex']
        (1,8) .. controls +(0,4) and +(0,4) .. (3,10);
        \draw [black,line width=1pt,<-,>=latex']
        (5,12) .. controls +(0,6) and +(0,12) .. (10,3);
      \end{tikzpicture}
      \label{subfig:precedence LL}
    }
    \qquad\qquad
    \subfigure[]{%
      \begin{tikzpicture}[scale=0.24,label/.style={anchor=base}]
        \draw[step=1cm,black!75,ultra thin,fill=black!50]
        (-0.2,6.8) grid (6.2,13.2);
        \foreach \x/\y in {0/7,1/8,2/9,3/10,4/11,5/12,6/13} {
          \draw [fill=black] (\x,\y) circle (0.2);
        }
        \draw[step=1cm,black!75,ultra thin,fill=black!50]
        (6.8,-0.2) grid (13.2,6.2);
        \draw [fill=black] (10,3) circle (0.2);
        \node [label] (ap) at (3,5) {$i'$};
        \node [label] (a) at (10,-2) {$i$};
        \node [label] (bp) at (1,5) {$j'$};
        \node [label] (b) at (5,5) {$j$};
        \node (sigma1) at (3,3) {$\sigma_1$};
        \node (sigma2) at (10,-3.5) {$\sigma_2$};
        \draw [black,line width=1pt,<-,>=latex']
        (1,8) .. controls +(0,6) and +(0,4) .. (5,12);
        \draw [black,line width=1pt,<-,>=latex']
        (3,10) .. controls +(0,8) and +(0,14) .. (10,3);
      \end{tikzpicture}
      \label{subfig:crossing LL}
    }
    \caption{\label{fig:sigma_1 increasing left above sigma_2}%
      Illustration of Lemma~\ref{lemma:sigma_1 increasing left above sigma_2}.}
  \end{figure}

  \begin{Lemma}
    \label{lemma:sigma_1 increasing left above sigma_2}
    Let $\pi = \pi_1 \, \sigma_1 \, \pi_2 \, \sigma_2 \, \pi_3$
    be a permutation where $\sigma_1$ is an increasing pattern
    and $\sigma_2$ is (right) below $\sigma_1$,
    and $\DMATCHING$ be a directed perfect matching on $\pi$
    that has Properties $\mathbf{P_1}$ and $\mathbf{P_2}$.
    If $\DMATCHING$ contains a $(\sigma_1, \sigma_2)$-arc
    or a $(\sigma_2, \sigma_1)$-arc, then it does not contain a
    $(\sigma_1, \sigma_1)$-arc.
  \end{Lemma}

  \begin{proof}[Proof of Lemma~\ref{lemma:sigma_1 increasing left above sigma_2}]
    Suppose, aiming at a contradiction, that
    $\DMATCHING$ contains a $(\sigma_1, \sigma_2)$-arc
    or a $(\sigma_2, \sigma_1)$-arc $(i, i')$, and a
    $(\sigma_1, \sigma_1)$-arc $(j, j')$.
    Since $\DMATCHING$ has Property~$\mathbf{P_1}$,
    it avoids the unlabeled patterns of $\mathcal{P}_1$.
    Therefore, $\DMATCHING$ contains one of the following labeled patterns:
    $\LabeledPrecedenceRR{.08}{2}{3}{4}{1}$,
    $\LabeledPrecedenceLR{.08}{2}{3}{4}{1}$,
    $\LabeledCrossingRR{.08}{2}{3}{4}{1}$,
    $\LabeledPrecedenceRL{.08}{2}{3}{4}{1}$,
    $\LabeledPrecedenceLL{.08}{2}{3}{4}{1}$ and
    $\LabeledCrossingLL{.08}{2}{3}{4}{1}$
    (see Figure~\ref{fig:sigma_1 increasing left above sigma_2}).
    Hence, according to Lemmas~\ref{lemma:Forbidden patterns crossing}
    and~\ref{lemma:Forbidden patterns precedence},
    $\DMATCHING$ cannot have Property~$\mathbf{P_2}$.
    This is the sought-after contradiction.
  \end{proof}

  \begin{Lemma} \label{lemma:reverse directed matching}
    Let $\pi$ be a permutation and $\DMATCHING$ be a directed perfect
    matching on $\pi$. If $\DMATCHING$ has properties $\mathbf{P_1}$
    and $\mathbf{P_2}$, then so does the directed perfect matching
    $\DMATCHING^\mathrm{r}$ obtained from $\DMATCHING$ by reversing
    each of its arcs.
  \end{Lemma}
  \begin{proof}[Proof of Lemma~\ref{lemma:reverse directed matching}]
    It is immediate that $\DMATCHING^\mathrm{r}$ satisfies
    Property $\mathbf{P_2}$, since, for any two arcs $(i, i')$ and
    $(j, j')$ of $\DMATCHING$, we have $\pi(i) < \pi(j)$ if and only if
    $\pi(i') < \pi(j')$. As for Property~$\mathbf{P_1}$, it is enough to
    observe that the set of unlabeled patterns $\mathcal{P}_1$ is closed
    by arc reversals.
  \end{proof}

  A direct interpretation of Lemma~\ref{lemma:reverse directed matching}
  is that, if a permutation $\pi$ is a square, one can exchange the roles
  of the two order-isomorphic patterns that cover~$\pi$. This can also
  be seen as a consequence of
  Proposition~\ref{prop:shuffle_associative_commutative} about the
  commutativity of $\SHUFFLE$. Besides, an immediate but useful
  consequence of Lemma~\ref{lemma:reverse directed matching} reads as
  follows.

  \begin{Corollary}
    \label{corollary:reverse directed matching}
    Let $\pi$ be a permutation and $i$ and $i'$ be two distinct indexes
    of $\pi$. There exists a directed perfect matching on $\pi$
    with Properties $\mathbf{P_1}$ and $\mathbf{P_2}$ that contains the
    arc $(i, i')$ if and only of there exists a directed perfect
    matching on $\pi$ with Properties $\mathbf{P_1}$ and $\mathbf{P_2}$
    that contains the arc $(i', i)$.
  \end{Corollary}

  Having disposed of these preliminary observations, we now turn to
  stating and proving the \NP-hardness of the targeted problem.

  \begin{Proposition} \label{proposition:hardness}
    Deciding whether a permutation is a square is \NP-complete.
  \end{Proposition}
  \begin{proof}[Proof of Proposition~\ref{proposition:hardness}]
    This decision problem is certainly in $\NP$. To prove that it
    is \NP-complete, we propose a reduction from the pattern involvement problem
    which is known to be \NP-complete \cite{Bose:Buss:Lubiw:1998}.

    Let $\pi \in S_n$ and $\sigma \in S_k$ be two permutations.
    Let us set
    \begin{equation}\begin{split}
        N_4 &= 2(2n + k + 2) + 3, \\
        N_3 &= 2(2N_4 + 2n + 2k + 4) + 3, \\
        N_2 &= 2(2N_3 + 2N_4 + 2n + 2k + 4) + 3, \\
        N_1 &= 2(2N_2 + 2N_3 + 2N_4 + 2n + 2k + 4) + 3.
    \end{split}\end{equation}
    Notice that $N_1, N_2, N_3$ and $N_4$ are polynomials in $n$.
    The crucial properties are that
    \begin{enumerate}[label={\it (\roman*)},fullwidth]
    \item the integers $N_1, N_2, N_3$ and $N_4$ are odd;
    \item the relation
      \begin{equation}
        N_i > \left(\sum_{i < j \leq 4} 2N_j\right) + 2n + 2k + 4
      \end{equation}
      holds for every $i \in [k]$.
    \end{enumerate}

    To construct a new permutation $\mu$ from $\pi$ and $\sigma$,
    we now turn to defining various gadgets (sequences of integers)
    that will act as building blocks.
    Recall that, for any permutation $p = p_1 \, p_2 \cdots \, p_x$ of $[x]$
    and any non-negative integer $y$,
    $p \; [y]$ stand for the sequence
    $(p_1 + y)\, (p + y) \, \cdots \, (p_x + y)$).
    Define
    \begin{equation}\begin{split}
        \sigma'  & =
        ((k+1) \; \sigma \; (k+2)) \; [2N_2 + N_4 + 2n + k + 2], \\
        \pi'     & =
        ((n+1) \; \pi \; (n+2)) \; [2N_2 + N_4 + n + k + 2], \\
        \sigma'' & = \sigma \; [2N_2 + N_4], \\
        \pi''    & = \pi \; [2N_2 + N_4 + k], \\
        \nu_1    & =
        \nearrow_{N_1} \; [2N_2 + 2N_3 + 2N_4 + 2n + 2k + 4], \\
        \nu'_1   & =
        \nearrow_{N_1} \; [N_1 + 2N_2 + 2N_3 + 2N_4 + 2n + 2k + 4], \\
        \nu_2    & = \searrow_{N_2} \; [N_2], \\
        \nu'_2   & = \searrow_{N_2}, \\
        \nu_3    & = \nearrow_{N_3} \; [2N_2 + 2N_4 + 2n + 2k + 4], \\
        \nu'_3   & =
        \nearrow_{N_3} \; [2N_2 + N_3 + 2N_4 + 2n + 2k + 4], \\
        \nu_4    & = \searrow_{N_4} \; [2N_2 + N_4 + 2n + 2k + 4], \\
        \nu'_4   & = \searrow_{N_4} \; [2N_2]\text{.}
    \end{split}\end{equation}
    We are now in position to define our target permutation $\mu$
    (see Figure~\ref{fig:reduction} for an illustration) as
    \begin{equation}
      \mu = \nu_1 \; \nu_2 \; \nu'_1 \; \nu_3 \; \sigma'
      \; \nu_4 \; \nu'_2 \; \nu'_3 \; \pi' \; \nu'_4
      \; \pi'' \; \sigma''.
    \end{equation}
    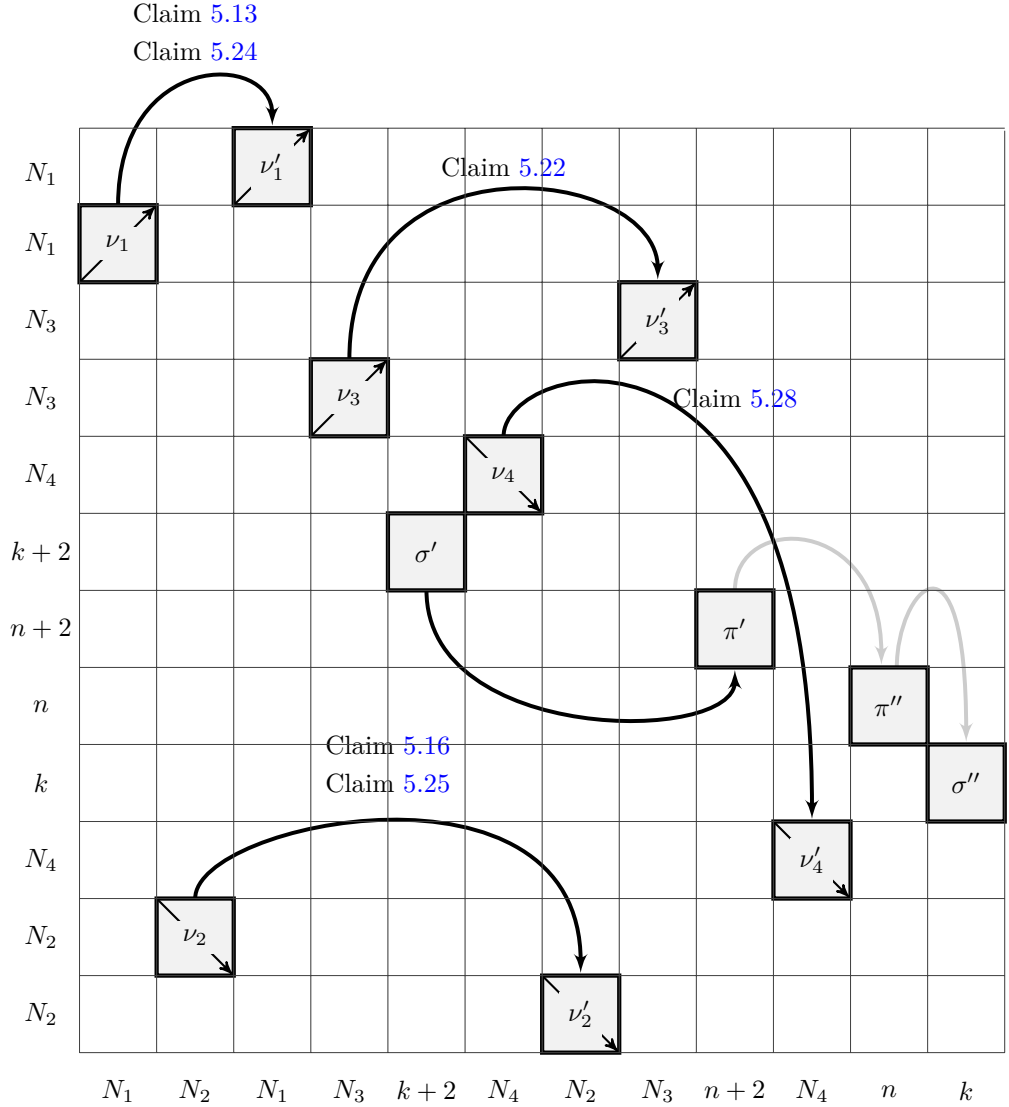
\begin{figure}[ht!]
      \centering
      \begin{tikzpicture}[scale=1.025,>=stealth',shorten >=1pt,
          main node/.style={align=center},
          cell/.style={draw,ultra thick,fill=black!5},
          structure link/.style={line width=1.5pt},
          pattern link/.style={line width=1.5pt,black!20},
          monotone/.style={->,thick}]
        \draw [pattern link,->,>=latex'] (8.5,6) .. controls +(0,1)
        and +(0,2) .. (10.4,5);
        \draw [pattern link,->,>=latex'] (10.6,5) .. controls +(0,1)
        and +(0,3) .. (11.5,4);
        \draw [structure link,->,>=latex'] (0.5,11) .. controls +(0,2)
        and +(0,1) .. (2.5,12);
        \node (claim arc 1) at (1.5,13.5)
              {Claim~\ref{claim:(nu_1, nu'_1)-arc}};
              \node (claim arc 1) at (1.5,13)
                    {Claim~\ref{claim:all (nu_1, nu'_1)-arcs}};
                    \draw [structure link,->,>=latex'] (1.5,2) .. controls +(0,1)
                    and +(0,3) .. (6.5,1);
                    \node (claim arc 2) at (4,4)
                          {Claim~\ref{claim:(nu_2, nu'_2)-arc}};
                          \node (claim arc 2) at (4,3.5)
                                {Claim~\ref{claim:all (nu_2, nu'_2)-arcs}};
                                \draw [structure link,->,>=latex'] (3.5,9) .. controls +(0,3)
                                and +(0,1.5) .. (7.5,10);
                                \node (claim arc 1) at (5.5,11.5)
                                      {Claim~\ref{claim:at least one (nu_3, nu'_3)-arc}};
                                      \draw [structure link,->,>=latex'] (5.5,8) .. controls +(0,1)
                                      and +(0,7) .. (9.5,3);
                                      \node (claim arc 1) at (8.5,8.5)
                                            {Claim~\ref{claim:at least one (nu_4, nu'_4)-arc}};
                                            \draw [structure link,->,>=latex'] (4.5,6) .. controls +(0,-2)
                                            and +(0,-1) .. (8.5,5);
                                            \draw [cell] (0,10) -- (1,10) -- (1,11) -- (0,11) -- cycle;
                                            \draw [monotone] (0,10) -- ++(1,1)
                                            node [midway,fill=white,fill=black!5] {$\nu_1$};
                                            \draw [cell] (2,11) -- (3,11) -- (3,12) -- (2,12) -- cycle;
                                            \draw [monotone] (2,11) -- ++(1,1)
                                            node [midway,fill=white,fill=black!5] {$\nu'_1$};
                                            \draw [cell] (1,1) -- (2,1) -- (2,2) -- (1,2) -- cycle;
                                            \draw [monotone] (1,2) -- ++(1,-1)
                                            node [midway,fill=white,fill=black!5] {$\nu_2$};
                                            \draw [cell] (6,0) -- (7,0) -- (7,1) -- (6,1) -- cycle;
                                            \draw [monotone] (6,1) -- ++(1,-1)
                                            node [midway,fill=white,fill=black!5] {$\nu'_2$};
                                            \draw [cell] (3,8) -- (4,8) -- (4,9) -- (3,9) -- cycle;
                                            \draw [monotone] (3,8) -- ++(1,1)
                                            node [midway,fill=white,fill=black!5] {$\nu_3$};
                                            \draw [cell] (7,9) -- (8,9) -- (8,10) -- (7,10) -- cycle;
                                            \draw [monotone] (7,9) -- ++(1,1)
                                            node [midway,fill=white,fill=black!5] {$\nu'_3$};
                                            \draw [cell] (5,7) -- (6,7) -- (6,8) -- (5,8) -- cycle;
                                            \draw [monotone] (5,8) -- ++(1,-1)
                                            node [midway,fill=white,fill=black!5] {$\nu_4$};
                                            \draw [cell] (9,2) -- (10,2) -- (10,3) -- (9,3) -- cycle;
                                            \draw [monotone] (9,3) -- ++(1,-1)
                                            node [midway,fill=white,fill=black!5] {$\nu'_4$};
                                            \draw [cell] (4,6) -- (5,6) -- (5,7) -- (4,7) -- cycle;
                                            \node [main node] (sigma1) at (4.5,6.5) {$\sigma'$};
                                            \draw [cell] (8,5) -- (9,5) -- (9,6) -- (8,6) -- cycle;
                                            \node [main node] (pi1) at (8.5,5.5) {$\pi'$};
                                            \draw [cell] (10,4) -- (11,4) -- (11,5) -- (10,5) -- cycle;
                                            \node [main node] (sigma2) at (10.5,4.5) {$\pi''$};
                                            \draw [cell] (11,3) -- (12,3) -- (12,4) -- (11,4) -- cycle;
                                            \node [main node] (pi2) at (11.5,3.5) {$\sigma''$};
                                            \draw[step=1cm,black!75,ultra thin,fill=black!50]
                                            (0,0) grid (12,12);
                                            \foreach \y/\N in
                                                     {0.5/N_2,1.5/N_2,2.5/N_4,3.5/k,4.5/n,5.5/n+2,6.5/k+2,
                                                       7.5/N_4,8.5/N_3,9.5/N_3,10.5/N_1,11.4/N_1} {
                                                       \node at (-0.5,\y) (R\y) {$\N$};
                                                     }
                                                     \foreach \x/\N in
                                                              {0.5/N_1,1.5/N_2,2.5/N_1,3.5/N_3,4.5/k+2,5.5/N_4,
                                                                6.5/N_2,7.5/N_3,8.5/n+2,9.5/N_4,10.5/n,11.5/k} {
                                                                \node at (\x,-0.5) (C\x) {$\N$};
                                                              }
      \end{tikzpicture}
      \caption{\label{fig:reduction}%
        Schematic representation of the permutation $\mu$ used in the
        proof of Proposition~\ref{proposition:hardness}. Black arcs
        denote the presence of at least one arc between two  bunches of
        positions in $\mu$. Grey arcs denote arcs that are only
        considered in the forward direction of the proof.}
    \end{figure}

    It is immediate that $\mu$ can be constructed in polynomial-time in
    $n$ and $k$. We claim that $\sigma$ occurs in $\pi$ if and only if
    there exists a directed perfect matching $\DMATCHING$ on $\mu$
    that has Properties $\mathbf{P_1}$ and $\mathbf{P_2}$ (that is,
    by Lemma~\ref{lemma:matching}, $\mu$ is a square).

    Suppose first that $\sigma$ occurs in $\pi$ and fix any occurrence.
    Construct a directed matching $\DMATCHING$ on $\mu$ as follows
    (all arcs are oriented to the right):
    \begin{enumerate}[label={\it (\arabic*)},fullwidth]
    \item $\DMATCHING$ contains $N_1$ pairwise crossing
      $(\nu_1, \nu'_1)$-arcs.
    \item $\DMATCHING$ contains $N_2$ pairwise crossing
      $(\nu_2, \nu'_2)$-arcs.
    \item $\DMATCHING$ contains $N_3$ pairwise crossing
      $(\nu_3, \nu'_3)$-arcs.
    \item $\DMATCHING$ contains $N_4$ pairwise crossing
      $(\nu_4, \nu'_4)$-arcs.
    \item $\DMATCHING$ contains $k+2$ pairwise crossing
      $(\sigma', \pi')$-arcs as depicted in
      Figure~\ref{fig:subfig:sigma' - pi' - pi'' - sigma''}.
      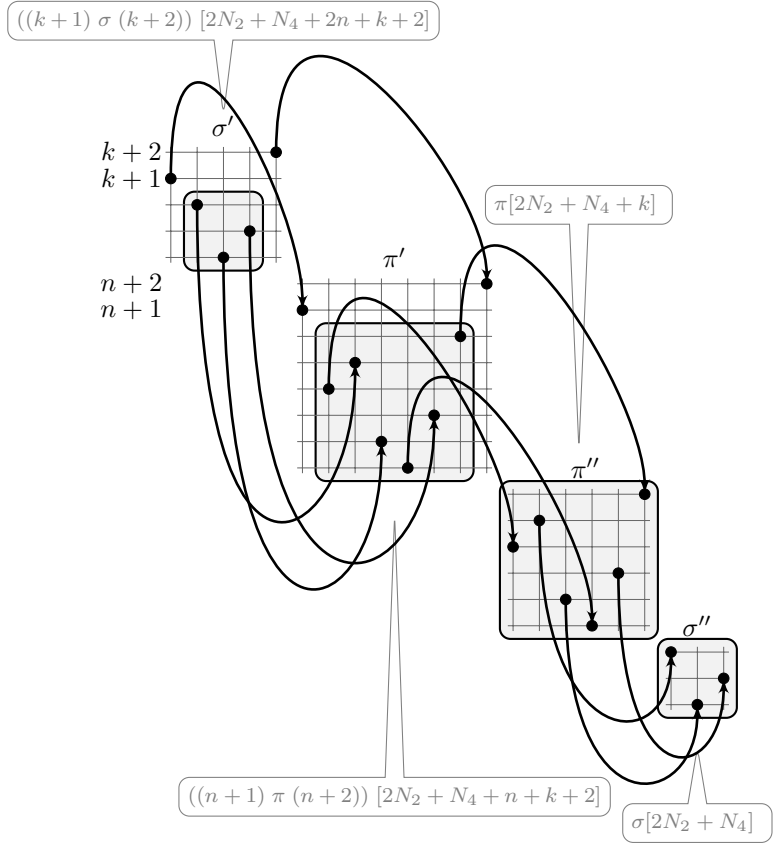
\begin{figure}[ht!]
        \centering
        \begin{tikzpicture}[yscale=.35,xscale=.35,
            label/.style={anchor=base},
            cell/.style={draw,thick,fill=black!5}]
          \draw [cell,rounded corners] (1.5,16.5) -- (4.5,16.5)
          -- (4.5,19.5) -- (1.5,19.5) -- cycle;
          \draw[step=1cm,black!60,ultra thin,fill=black!10]
          (0.8,16.8) grid (5.2,21.2);
          \foreach \x/\y in {1/20,2/19,3/17,4/18,5/21} {
            \draw [fill=black] (\x,\y) circle (0.2);
          }
          \node (pip) at (3,22) {$\sigma'$};
          \node[draw,text width=5.5cm,text justified,color=black!50,
            rounded corners,rectangle callout,font=\footnotesize,
            callout relative pointer={(0.0cm,-1cm)}] at (3,26)
               {%
                 $((k+1) \; \sigma \; (k+2)) \;
                 [2N_2 + N_4 + 2n + k + 2]$
               };%
               \node at (-0.5,20) {$k+1$};
               \node at (-0.5,21) {$k+2$};
               \draw [cell,rounded corners] (6.5,8.5) -- (12.5,8.5)
               -- (12.5,14.5) -- (6.5,14.5) -- cycle;
               \draw[step=1cm,black!60,ultra thin,fill=black!10]
               (5.8,8.8) grid (13.2,16.2);
               \foreach \x/\y in {6/15,7/12,8/13,9/10,10/9,
                 11/11,12/14,13/16} {
                 \draw [fill=black] (\x,\y) circle (0.2);
               }
               \node (pip) at (9.5,17) {$\pi'$};
               \node[draw,text width=5.5cm,text justified,color=black!50,
                 rounded corners, rectangle callout,font=\footnotesize,
                 callout relative pointer={(0.0cm,3.5cm)}] at (9.5,-3.5)
                    {%
                      $((n+1) \; \pi \; (n+2)) \; [2N_2 + N_4 + n + k + 2]$
                    };%
                    \node at (-0.5,15) {$n+1$};
                    \node at (-0.5,16) {$n+2$};
                    \draw [cell,rounded corners] (13.5,2.5) -- (19.5,2.5)
                    -- (19.5,8.5) -- (13.5,8.5) -- cycle;
                    \draw[step=1cm,black!60,ultra thin,fill=black!10]
                    (13.8,2.8) grid (19.2,8.2);
                    \foreach \x/\y in {14/6,15/7,16/4,17/3,18/5,19/8} {
                      \draw [fill=black] (\x,\y) circle (0.2);
                    }
                    \node (pipp) at (16.75,9) {$\pi''$};
                    \node[draw,text width=2.25cm,text justified,color=black!50,
                      rounded corners, rectangle callout,font=\footnotesize,
                      callout relative pointer={(0.0cm,-3cm)}] at (16.5,19)
                         {%
                           $\pi [2N_2 + N_4 + k]$
                         };%
                         \draw [cell,rounded corners] (19.5,-0.5) -- (22.5,-.5)
                         -- (22.5,2.5) -- (19.5,2.5) -- cycle;
                         \draw[step=1cm,black!60,ultra thin,fill=black!10]
                         (19.8,-0.2) grid (22.2,2.2);
                         \foreach \x/\y in {20/2,21/0,22/1} {
                           \draw [fill=black] (\x,\y) circle (0.2);
                         }
                         \node (sigmapp) at (21,3) {$\sigma''$};
                         \node[draw,text width=1.75cm,text justified,color=black!50,
                           rounded corners,rectangle callout,font=\footnotesize,
                           callout relative pointer={(0.0cm,0.75cm)}] at (21,-4.5)
                              {%
                                $\sigma [2N_2 + N_4]$
                              };%
                              \draw [black,line width=1pt,->,>=latex']
                              (1,20) .. controls +(0,9) and +(0,4) .. (6,15);
                              \draw [black,line width=1pt,->,>=latex']
                              (5,21) .. controls +(0,9) and +(0,4) .. (13,16);
                              \draw [black,line width=1pt,->,>=latex']
                              (2,19) .. controls +(0,-17) and +(0,-7) .. (8,13);
                              \draw [black,line width=1pt,->,>=latex']
                              (3,17) .. controls +(0,-17) and +(0,-7) .. (9,10);
                              \draw [black,line width=1pt,->,>=latex']
                              (4,18) .. controls +(0,-17) and +(0,-7) .. (11,11);
                              \draw [black,line width=1pt,->,>=latex']
                              (7,12) .. controls +(0,9) and +(0,4) .. (14,6);
                              \draw [black,line width=1pt,->,>=latex']
                              (10,9) .. controls +(0,9) and +(0,4) .. (17,3);
                              \draw [black,line width=1pt,->,>=latex']
                              (12,14) .. controls +(0,9) and +(0,4) .. (19,8);
                              \draw [black,line width=1pt,->,>=latex']
                              (15,7) .. controls +(0,-9) and +(0,-4) .. (20,2);
                              \draw [black,line width=1pt,->,>=latex']
                              (16,4) .. controls +(0,-9) and +(0,-4) .. (21,0);
                              \draw [black,line width=1pt,->,>=latex']
                              (18,5) .. controls +(0,-9) and +(0,-4) .. (22,1);
        \end{tikzpicture}
        \caption{\label{fig:subfig:sigma' - pi' - pi'' - sigma''}%
          Illustration of the directed perfect matching
          $\DMATCHING$ between gadgets $\sigma'$, $\pi'$, $\pi''$ and
          $\sigma''$ assuming two input permutation $\sigma = 312$ and
          $\pi = 4\mathbf{5}\mathbf{2}1\mathbf{3}6$ (where a specific
          occurrence of $\sigma$ in $\pi$ is depicted in bold).}
      \end{figure}
      More precisely,
      \begin{enumerate}[label={\it (\roman*)}]
      \item the first position of $\sigma'$ (\emph{i.e.},
        $(2N_1+N_2+N_3) + 1$) is linked to the first position of
        $\pi'$ (\emph{i.e.},
        $(2N_1 + 2N_2 + 2N_3 + N_4 + k + 2) + 1$),
      \item the last position of $\sigma'$ (\emph{i.e.},
        $(2N_1+N_2+N_3) + k+2$) is linked to the last position of
        $\pi'$ (\emph{i.e.},
        $(2N_1 + 2N_2 + 2N_3 + N_4 + k + 2) + n+2$), and all other
        positions in $\sigma'$ are linked by means of $k$ pairwise
        crossing arcs to the positions in $\pi'$ that correspond to
        the fixed occurrence of $\sigma$ in $\pi$. (Notice that we
        use here the fact that $\sigma$ occurs in $\pi$).
      \end{enumerate}
    \item $\DMATCHING$ contains $n-k$ pairwise crossing
      $(\pi', \pi'')$-arcs as depicted in
      Figure~\ref{fig:subfig:sigma' - pi' - pi'' - sigma''}.
      More precisely, all positions in $\pi'$ that do not correspond
      to the fixed occurrence of $\sigma$ in $\pi$ are linked by means
      of $n-k$ pairwise crossing arcs to the positions in $\pi''$ that
      do not correspond to the fixed occurrence of $\sigma$ in $\pi$.
    \item $\DMATCHING$ contains $k$ pairwise crossing
      $(\pi'', \sigma'')$-arcs as depicted in
      Figure~\ref{fig:subfig:sigma' - pi' - pi'' - sigma''}.
      More precisely, the positions in $\pi''$ that correspond to
      the fixed occurrence of $\sigma$ in $\pi$ are linked by means of
      $k$ pairwise crossing arcs to all positions in $\sigma''$.
      (Notice that, again, we use here the fact that $\sigma$ occurs
      in $\pi$).
    \end{enumerate}
    It can be easily checked (probably referring to
    Figure~\ref{fig:reduction}) that $\DMATCHING$ is perfect and has
    Properties $\mathbf{P_1}$ and $\mathbf{P_2}$.

    Conversely, suppose that there exists an directed perfect matching
    $\DMATCHING$ on $\mu$ that has Properties $\mathbf{P_1}$ and
    $\mathbf{P_2}$. We show that $\sigma$ occurs as a pattern in $\pi$.
    Whereas the directed perfect matching $\DMATCHING$ may not be as
    regular as  in the forward direction, the main idea is to prove that
    $\DMATCHING$ contains enough structure (more precisely, $k+2$
    $(\sigma', \pi')$-arcs) so that we can conclude that $\sigma$ occurs
    in $\pi$. We have divided the reverse direction into a set of basic
    claims that progressively defines and refines the overall structure
    of $\DMATCHING$.

    \begin{Claim}
      \label{claim:no (nu'_1, nu_1)-arc}
      We may assume that there is no $(\nu'_1, \nu_1)$-arc in $\DMATCHING$.
    \end{Claim}

    \begin{proof}[Proof of Claim~\ref{claim:no (nu'_1, nu_1)-arc}]
      We first observe that, according to Property~$\mathbf{P_1}$,
      since $\DMATCHING$ avoids the unlabeled patterns of $\mathcal{P}_1$,
      $\DMATCHING$ cannot contain both
      a $(\nu_1, \nu'_1)$-arc and a $(\nu'_1, \nu_1)$-arc.
      Now, if $\DMATCHING$ does not contain a $(\nu'_1, \nu_1)$-arc
      we are done.
      Otherwise, $\DMATCHING$ does contain some $(\nu'_1, \nu_1)$-arcs
      and
      no $(\nu_1, \nu'_1)$-arc, and the result follows from
      Lemma~\ref{lemma:reverse directed matching}.
    \end{proof}

    \begin{figure}[t!]
      \centering
      \subfigure[]{%
        \begin{tikzpicture}
          [
            scale=0.2,
            label/.style={anchor=base}
          ]
          \draw[step=1cm,black!75,ultra thin,fill=black!50] (-0.2,6.8) grid (6.2,13.2);
          \foreach \x/\y in {0/7,1/8,2/9,3/10,4/11,5/12,6/13} {
            \draw [fill=black] (\x,\y) circle (0.2);
          }
          \draw[step=1cm,black!75,ultra thin,fill=black!50] (6.8,-0.2) grid (13.2,6.2);
          \foreach \x/\y in {7/6,8/5,9/4,10/3,11/2,12/1,13/0} {
            \draw [fill=black] (\x,\y) circle (0.2);
          }
          \draw[step=1cm,black!75,ultra thin,fill=black!50] (13.8,13.8) grid (20.2,20.2);
          \foreach \x/\y in {14/14,15/15,16/16,17/17,18/18,19/19,20/20} {
            \draw [fill=black] (\x,\y) circle (0.2);
          }
          \node [label] (i) at (0,5) {$i$};
          \node [label] (ip) at (9,-2) {$i'$};
          \node [label] (j) at (11,-2) {$j$};
          \node [label] (jp) at (20,12) {$j'$};
          \node (nu1) at (3,3.5) {$\nu_1$};
          \node (nu2) at (10,-3.5) {$\nu_2$};
          \node (nu1') at (17,10.5) {$\nu'_1$};
          \draw [black,line width=1pt,->,>=latex']
          (1,8) .. controls +(0,20) and +(0,4) .. (14,14);
          \draw [black,line width=1pt,->,>=latex']
          (2,9) .. controls +(0,20) and +(0,4) .. (15,15);
          \draw [black,line width=1pt,->,>=latex']
          (3,10) .. controls +(0,20) and +(0,4) .. (16,16);
          \draw [black,line width=1pt,->,>=latex']
          (4,11) .. controls +(0,20) and +(0,4) .. (17,17);
          \draw [black,line width=1pt,->,>=latex']
          (5,12) .. controls +(0,20) and +(0,4) .. (18,18);
          \draw [black,line width=1pt,->,>=latex']
          (6,13) .. controls +(0,20) and +(0,4) .. (19,19);
          \draw [black,line width=1pt,->,>=latex']
          (0,7) .. controls +(0,10) and +(0,16) .. (9,4);
          \draw [black,line width=1pt,->,>=latex']
          (11,2) .. controls +(0,28) and +(0,10) .. (20,20);
        \end{tikzpicture}
        \label{subfig:no (nu_1, nu_2)-arc - 1}
      }
      \qquad
      \subfigure[]{%
        \begin{tikzpicture}
          [
            scale=0.2,
            label/.style={anchor=base}
          ]
          \draw[step=1cm,black!75,ultra thin,fill=black!50] (-0.2,6.8) grid (6.2,13.2);
          \foreach \x/\y in {0/7,1/8,2/9,3/10,4/11,5/12,6/13} {
            \draw [fill=black] (\x,\y) circle (0.2);
          }
          \draw[step=1cm,black!75,ultra thin,fill=black!50] (6.8,-0.2) grid (13.2,6.2);
          \foreach \x/\y in {7/6,8/5,9/4,10/3,11/2,12/1,13/0} {
            \draw [fill=black] (\x,\y) circle (0.2);
          }
          \draw[step=1cm,black!75,ultra thin,fill=black!50] (13.8,13.8) grid (20.2,20.2);
          \foreach \x/\y in {14/14,15/15,16/16,17/17,18/18,19/19,20/20} {
            \draw [fill=black] (\x,\y) circle (0.2);
          }
          \node [label] (i) at (0,5) {$i$};
          \node [label] (ip) at (11,-2) {$i'$};
          \node [label] (j) at (9,-2) {$j$};
          \node [label] (jp) at (20,12) {$j'$};
          \node (nu1) at (3,3.5) {$\nu_1$};
          \node (nu2) at (10,-3.5) {$\nu_2$};
          \node (nu1') at (17,10.5) {$\nu'_1$};
          \draw [black,line width=1pt,->,>=latex']
          (1,8) .. controls +(0,20) and +(0,4) .. (14,14);
          \draw [black,line width=1pt,->,>=latex']
          (2,9) .. controls +(0,20) and +(0,4) .. (15,15);
          \draw [black,line width=1pt,->,>=latex']
          (3,10) .. controls +(0,20) and +(0,4) .. (16,16);
          \draw [black,line width=1pt,->,>=latex']
          (4,11) .. controls +(0,20) and +(0,4) .. (17,17);
          \draw [black,line width=1pt,->,>=latex']
          (5,12) .. controls +(0,20) and +(0,4) .. (18,18);
          \draw [black,line width=1pt,->,>=latex']
          (6,13) .. controls +(0,20) and +(0,4) .. (19,19);
          \draw [black,line width=1pt,->,>=latex']
          (0,7) .. controls +(0,10) and +(0,16) .. (11,2);
          \draw [black,line width=1pt,->,>=latex']
          (9,4) .. controls +(0,28) and +(0,10) .. (20,20);
        \end{tikzpicture}
        \label{subfig:no (nu_1, nu_2)-arc - 2}
      }
      \qquad
      \subfigure[]{%
        \begin{tikzpicture}
          [
            scale=0.2,
            label/.style={anchor=base}
          ]
          \draw[step=1cm,black!75,ultra thin,fill=black!50] (-0.2,6.8) grid (6.2,13.2);
          \foreach \x/\y in {0/7,1/8,2/9,3/10,4/11,5/12,6/13} {
            \draw [fill=black] (\x,\y) circle (0.2);
          }
          \draw[step=1cm,black!75,ultra thin,fill=black!50] (6.8,-0.2) grid (13.2,6.2);
          \foreach \x/\y in {7/6,8/5,9/4,10/3,11/2,12/1,13/0} {
            \draw [fill=black] (\x,\y) circle (0.2);
          }
          \draw[step=1cm,black!75,ultra thin,fill=black!50] (13.8,13.8) grid (20.2,20.2);
          \foreach \x/\y in {14/14,15/15,16/16,17/17,18/18,19/19,20/20} {
            \draw [fill=black] (\x,\y) circle (0.2);
          }
          \node [label] (i) at (0,5) {$i$};
          \node [label] (ip) at (10,-2) {$i'$};
          \node [label] (j) at (20,12) {$j$};
          \node [label] (jp) at (24,-6) {$j'$};
          \node (nu1) at (3,3.5) {$\nu_1$};
          \node (nu2) at (10,-3.5) {$\nu_2$};
          \node (nu1') at (17,10.5) {$\nu'_1$};
          \draw [black,line width=1pt,->,>=latex']
          (1,8) .. controls +(0,20) and +(0,4) .. (14,14);
          \draw [black,line width=1pt,->,>=latex']
          (2,9) .. controls +(0,20) and +(0,4) .. (15,15);
          \draw [black,line width=1pt,->,>=latex']
          (3,10) .. controls +(0,20) and +(0,4) .. (16,16);
          \draw [black,line width=1pt,->,>=latex']
          (4,11) .. controls +(0,20) and +(0,4) .. (17,17);
          \draw [black,line width=1pt,->,>=latex']
          (5,12) .. controls +(0,20) and +(0,4) .. (18,18);
          \draw [black,line width=1pt,->,>=latex']
          (6,13) .. controls +(0,20) and +(0,4) .. (19,19);
          \draw [black,line width=1pt,->,>=latex']
          (0,7) .. controls +(0,10) and +(0,16) .. (10,3);
          \draw [black,line width=1pt,->,>=latex']
          (20,20) .. controls +(0,10) and +(0,28) .. (24,-4);
        \end{tikzpicture}
        \label{subfig:no (nu_1, nu_2)-arc - 3}
      }
      \qquad
      \subfigure[]{%
        \begin{tikzpicture}
          [
            scale=0.2,
            label/.style={anchor=base}
          ]
          \draw[step=1cm,black!75,ultra thin,fill=black!50] (-0.2,6.8) grid (6.2,13.2);
          \foreach \x/\y in {0/7,1/8,2/9,3/10,4/11,5/12,6/13} {
            \draw [fill=black] (\x,\y) circle (0.2);
          }
          \draw[step=1cm,black!75,ultra thin,fill=black!50] (6.8,-0.2) grid (13.2,6.2);
          \foreach \x/\y in {7/6,8/5,9/4,10/3,11/2,12/1,13/0} {
            \draw [fill=black] (\x,\y) circle (0.2);
          }
          \draw[step=1cm,black!75,ultra thin,fill=black!50] (13.8,13.8) grid (20.2,20.2);
          \foreach \x/\y in {14/14,15/15,16/16,17/17,18/18,19/19,20/20} {
            \draw [fill=black] (\x,\y) circle (0.2);
          }
          \node [label] (i) at (0,5) {$i$};
          \node [label] (ip) at (10,-2) {$i'$};
          \node [label] (j) at (20,12) {$j$};
          \node [label] (jp) at (24,-6) {$j'$};
          \node (nu1) at (3,3.5) {$\nu_1$};
          \node (nu2) at (10,-3.5) {$\nu_2$};
          \node (nu1') at (17,10.5) {$\nu'_1$};
          \draw [black,line width=1pt,->,>=latex']
          (1,8) .. controls +(0,20) and +(0,4) .. (14,14);
          \draw [black,line width=1pt,->,>=latex']
          (2,9) .. controls +(0,20) and +(0,4) .. (15,15);
          \draw [black,line width=1pt,->,>=latex']
          (3,10) .. controls +(0,20) and +(0,4) .. (16,16);
          \draw [black,line width=1pt,->,>=latex']
          (4,11) .. controls +(0,20) and +(0,4) .. (17,17);
          \draw [black,line width=1pt,->,>=latex']
          (5,12) .. controls +(0,20) and +(0,4) .. (18,18);
          \draw [black,line width=1pt,->,>=latex']
          (6,13) .. controls +(0,20) and +(0,4) .. (19,19);
          \draw [black,line width=1pt,->,>=latex']
          (0,7) .. controls +(0,10) and +(0,16) .. (10,3);
          \draw [black,line width=1pt,<-,>=latex']
          (20,20) .. controls +(0,10) and +(0,28) .. (24,-4);
        \end{tikzpicture}
        \label{subfig:no (nu_1, nu_2)-arc - 4}
      }
      \caption{\label{fig:subfig:no (nu_1, nu_2)-arc}%
        Illustration of Claim~\ref{claim:no (nu_1, nu_2)-arc, no (nu_2, nu_1)-arc}.
      }
    \end{figure}
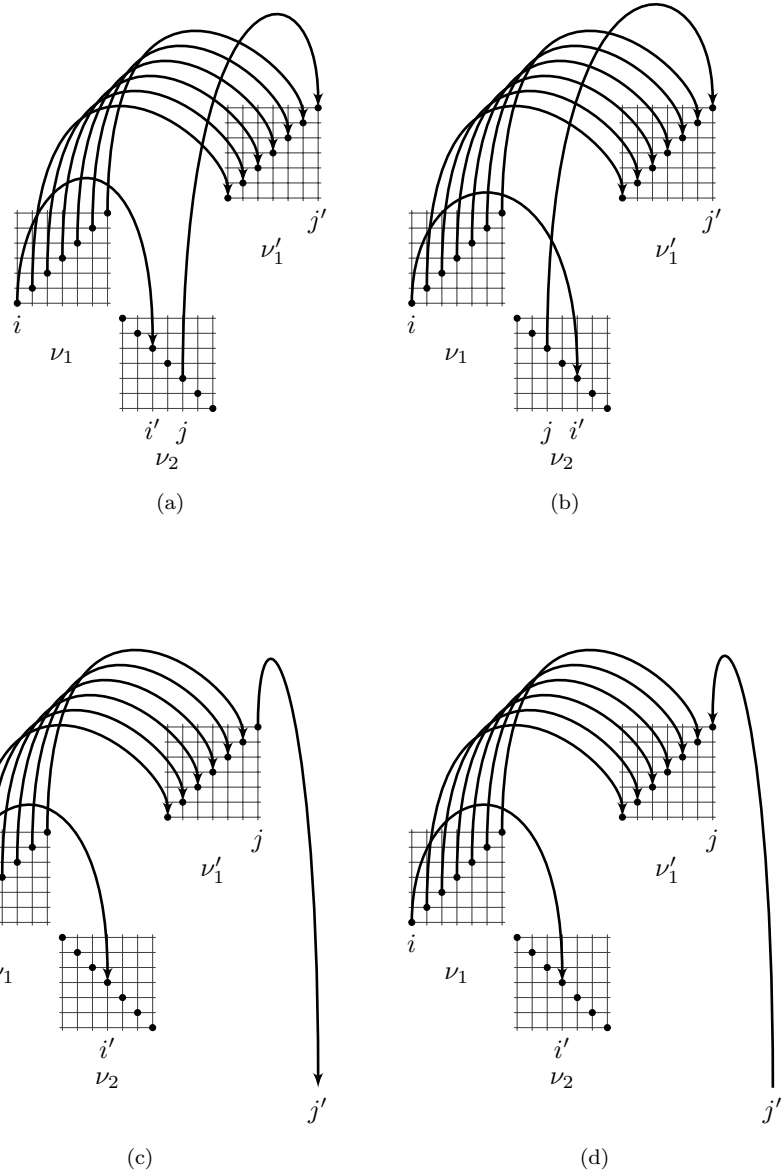

    \begin{Claim}
      \label{claim:no (nu_1, nu_2)-arc, no (nu_2, nu_1)-arc}
      There is neither a $(\nu_1, \nu_2)$-arc nor a $(\nu_2, \nu_1)$-arc
      in $\DMATCHING$.
    \end{Claim}

    \begin{proof}[Proof of Claim~\ref{claim:no (nu_1, nu_2)-arc, no (nu_2, nu_1)-arc}]
      First, according to
      Corollary~\ref{corollary:at most one arc monotone},
      there exists at most one arc between $\nu_1$ and $\nu_2$ in
      $\DMATCHING$ (this arc can be a $(\nu_1, \nu_2)$-arc or a
      $(\nu_2, \nu_1)$-arc).
      Suppose now, aiming at a contradiction, that there exists
      either one $(\nu_1, \nu_2)$-arc or one $(\nu_2, \nu_1)$-arc, say $a = (i, i')$,
      in $\DMATCHING$.
      In this case,
      according to Lemma~\ref{lemma:sigma_1 increasing left above sigma_2},
      $\DMATCHING$ does not contain any $(\nu_1, \nu_1)$-arc.
      We now claim that
      $\DMATCHING$ contains
      $N_1-1$ pairwise crossing $(\nu_1, \nu'_1)$-arcs (and $i=1$)
      if $a$ is a $(\nu_1, \nu_2)$-arc, or
      $N_1-1$ pairwise crossing $(\nu'_1, \nu_1)$-arcs (and $i'=1$)
      if $a$ is a $(\nu_2, \nu_1)$-arc
      (recall here that \CrossingLR\ and \CrossingRL\ are forbidden patterns in $\DMATCHING$).
      Indeed, observe first that
      $N_1-1
      > |\nu_3| + |\sigma'| + |\nu_4| + |\nu'_2| + |\nu'_3| + |\pi'| + |\nu'_4| + |\pi''| + |\sigma_1|$.
      Therefore, there exists at least
      one $(\nu_1, \nu'_1)$-arc if $a$ is a $(\nu_1, \nu_2)$-arc or
      at least one $(\nu'_1, \nu_1)$-arc if $a$ is a $(\nu_2, \nu_1)$-arc.
      Hence,
      if $\DMATCHING$ does not contain
      $N_1-1$ pairwise crossing $(\nu_1, \nu'_1)$-arcs or
      $N_1-1$ pairwise crossing $(\nu'_1, \nu_1)$-arcs,
      then it contains one of the following labeled patterns:
      $\LabeledPrecedenceLL{.08}{2}{3}{4}{1}$,
      $\LabeledPrecedenceLR{.08}{2}{3}{4}{1}$,
      $\LabeledPrecedenceRL{.08}{2}{3}{4}{1}$,
      $\LabeledPrecedenceRR{.08}{2}{3}{4}{1}$,
      $\LabeledCrossingLL{.08}{2}{3}{4}{1}$ and
      $\LabeledCrossingRR{.08}{2}{3}{4}{1}$.
      Applying Lemma~\ref{lemma:Forbidden patterns crossing} and
      Lemma~\ref{lemma:Forbidden patterns precedence}
      yields a contradiction.
      Then it follows that
      $\DMATCHING$ contains
      $N_1-1$ pairwise crossing $(\nu_1, \nu'_1)$-arcs (and $i=1$)
      if $a$ is a $(\nu_1, \nu_2)$-arc, or
      $N_1-1$ pairwise crossing $(\nu'_1, \nu_1)$-arcs (and $i'=1$)
      if $a$ is a $(\nu_2, \nu_1)$-arc.
      But it follows from Claim~\ref{claim:no (nu'_1, nu_1)-arc}
      that $\DMATCHING$ does not contain any $(\nu'_1, \nu_1)$-arc,
      and hence
      $\DMATCHING$ contains
      $N_1-1$ pairwise crossing $(\nu_1, \nu'_1)$-arcs
      and $a$ is a $(\nu_1, \nu_2)$-arc
      (since \CrossingLR\ is forbidden).

      We now observe that $|\nu_1| = |\nu'_1| = N_1$.
      Hence, since $\DMATCHING$ is perfect,
      there exists a position in $\nu'_1$ that is
      not involved in a $(\nu_1, \nu'_1)$-arc in $\DMATCHING$.
      We rule out this configuration by considering two cases:
      \begin{itemize}
      \item There exists a $(\nu_2, \nu'_1)$-arc $(j,j')$ in $\DMATCHING$
        (we cannot have a $(\nu'_1, \nu_2)$-arc since the unlabeled pattern
        \CrossingLR\ is forbidden),
        see Figure~\ref{subfig:no (nu_1, nu_2)-arc - 1} and
        Figure~\ref{subfig:no (nu_1, nu_2)-arc - 2}.
        Then it follows that $\DMATCHING$ contains the labeled pattern
        $\LabeledCrossingRR{.08}{2}{1}{3}{4}$
        (with arc $(j, j')$ and any $(\nu, \nu'_1)$-arc).
        Applying Lemma~\ref{lemma:Forbidden patterns crossing}
        yields the sought-after contradiction.

      \item There exists an arc $(j,j')$
        $j = 2N_1 + N_2$ and $j' > 2N_1 + N_2$,
        or
        $j' = 2N_1 + N_2$ and $j > 2N_1 + N_2$,
        see Figure~\ref{subfig:no (nu_1, nu_2)-arc - 3} and
        Figure~\ref{subfig:no (nu_1, nu_2)-arc - 4}.
        Then it follows that $\DMATCHING$ contains one of the two following
        labeled patterns:
        $\LabeledPrecedenceRR{.08}{3}{2}{4}{1}$ and
        $\LabeledPrecedenceLL{.08}{3}{2}{4}{1}$
        (with arc $(i, i')$ and arc ($j, j')$).
        Applying Lemma~\ref{lemma:Forbidden patterns precedence}
        yields the sought-after contradiction.
      \end{itemize}
    \end{proof}

    \begin{Claim}
      \label{claim:(nu_1, nu'_1)-arc}
      There is at least one $(\nu_1, \nu'_1)$-arc in $\DMATCHING$.
    \end{Claim}

    \begin{proof}[Proof of Claim~\ref{claim:(nu_1, nu'_1)-arc}]
      Suppose, aiming at a contradiction, that there is no
      $(\nu_1, \nu'_1)$-arc in $\DMATCHING$.
      Then it follows that there exists an arc $(i, i')$ in $\DMATCHING$
      that is neither a
      $(\nu_1, \nu_1)$-arc (since $N_1$ is odd)
      nor a $(\nu_1, \nu_2)$-arc (Claim~\ref{claim:no (nu_1, nu_2)-arc, no (nu_2, nu_1)-arc})
      nor a $(\nu_1, \nu'_1)$ (by our contradiction hypothesis).
      (In other words, $i \leq N_1$ and $i' > 2N_1 + N_2$,
      or $i' \leq N_1$ and $i > 2N_1 + N_2$.)
      Therefore, since $\DMATCHING$ is containment-free
      (\emph{i.e.}, it avoids the unlabeled patterns of
      $\mathcal{P}_{\mathrm{cont}}$),
      there is neither a $(\nu_2, \nu_2)$-arc nor a $(\nu'_1, \nu'_1)$-arc
      in $\DMATCHING$.
      Then it follows that $\DMATCHING$ contains
      either
      $N_1$ arcs $(j, j')$ with $N_1+N_2 < j \leq 2N_1+N_2$ and $j' > 2N_1+N_2$
      (if $i \leq N_1$ and $i' > 2N_1 + N_2$),
      or
      $N_1$ arcs $(i, i')$ with $N_1+N_2 < j' \leq 2N_1+N_2$ and $j > 2N_1+N_2$
      (if $i' \leq N_1$ and $i > 2N_1 + N_2$),
      otherwise $\DMATCHING$ would not be containment-free.
      But $|\nu'_1| = N_1
      > |\nu_3| + |\sigma'| + |\nu_4| + |\nu'_2| + |\nu'_3| + |\pi'| + |\nu'_4|
      + |\pi''| + |\sigma''|$.
      This is a contradiction.
    \end{proof}

    The above claim will be complemented in upcoming Claim~\ref{claim:all (nu_1, nu'_1)-arcs}.

    \begin{Claim}
      \label{claim:no (nu_2, nu_2)-arc}
      There is no $(\nu_2, \nu_2)$-arc in $\DMATCHING$.
    \end{Claim}

    \begin{proof}[Proof of Claim~\ref{claim:no (nu_2, nu_2)-arc}]
      Combine Claim~\ref{claim:(nu_1, nu'_1)-arc} together with
      the fact that $\DMATCHING$ is containment-free
      (\emph{i.e.}, it avoids the unlabeled patterns of
      $\mathcal{P}_{\mathrm{cont}}$).
    \end{proof}

    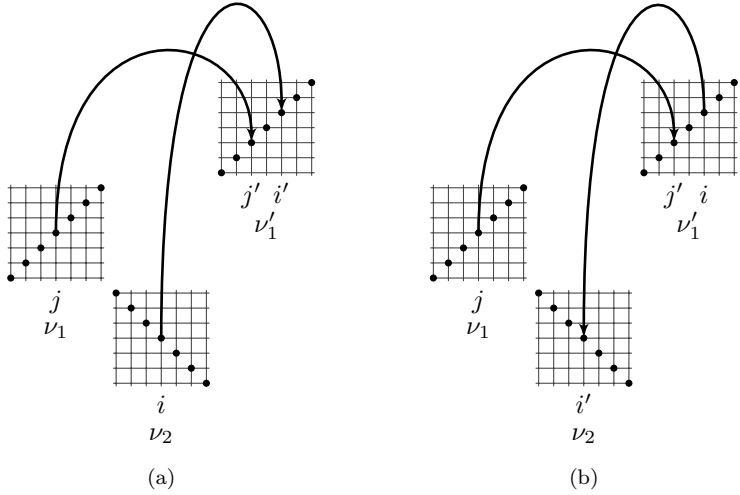
\begin{figure}[t!]
      \centering
      \subfigure[]{%
        \begin{tikzpicture}
          [
            scale=0.2,
            label/.style={anchor=base}
          ]
          \draw[step=1cm,black!75,ultra thin,fill=black!50] (-0.2,6.8) grid (6.2,13.2);
          \foreach \x/\y in {0/7,1/8,2/9,3/10,4/11,5/12,6/13} {
            \draw [fill=black] (\x,\y) circle (0.2);
          }
          \draw[step=1cm,black!75,ultra thin,fill=black!50] (6.8,-0.2) grid (13.2,6.2);
          \foreach \x/\y in {7/6,8/5,9/4,10/3,11/2,12/1,13/0} {
            \draw [fill=black] (\x,\y) circle (0.2);
          }
          \draw[step=1cm,black!75,ultra thin,fill=black!50] (13.8,13.8) grid (20.2,20.2);
          \foreach \x/\y in {14/14,15/15,16/16,17/17,18/18,19/19,20/20} {
            \draw [fill=black] (\x,\y) circle (0.2);
          }
          \node [label] (j) at (3,5) {$j$};
          \node [label] (jp) at (16,12) {$j'$};
          \node [label] (i) at (10,-2) {$i$};
          \node [label] (ip) at (18,12) {$i'$};
          \node (nu1) at (3,3.5) {$\nu_1$};
          \node (nu2) at (10,-3.5) {$\nu_2$};
          \node (nu1') at (17,10.5) {$\nu'_1$};
          \draw [black,line width=1pt,->,>=latex']
          (3,10) .. controls +(0,16) and +(0,8) .. (16,16);
          \draw [black,line width=1pt,->,>=latex']
          (10,3) .. controls +(0,28) and +(0,10) .. (18,18);
        \end{tikzpicture}
        \label{subfig:no (nu_2, nu'_)-arc - 1}
      }
      \qquad\qquad
      \subfigure[]{%
        \begin{tikzpicture}
          [
            scale=0.2,
            label/.style={anchor=base}
          ]
          \draw[step=1cm,black!75,ultra thin,fill=black!50] (-0.2,6.8) grid (6.2,13.2);
          \foreach \x/\y in {0/7,1/8,2/9,3/10,4/11,5/12,6/13} {
            \draw [fill=black] (\x,\y) circle (0.2);
          }
          \draw[step=1cm,black!75,ultra thin,fill=black!50] (6.8,-0.2) grid (13.2,6.2);
          \foreach \x/\y in {7/6,8/5,9/4,10/3,11/2,12/1,13/0} {
            \draw [fill=black] (\x,\y) circle (0.2);
          }
          \draw[step=1cm,black!75,ultra thin,fill=black!50] (13.8,13.8) grid (20.2,20.2);
          \foreach \x/\y in {14/14,15/15,16/16,17/17,18/18,19/19,20/20} {
            \draw [fill=black] (\x,\y) circle (0.2);
          }
          \node [label] (j) at (3,5) {$j$};
          \node [label] (jp) at (16,12) {$j'$};
          \node [label] (ip) at (10,-2) {$i'$};
          \node [label] (i) at (18,12) {$i$};
          \node (nu1) at (3,3.5) {$\nu_1$};
          \node (nu2) at (10,-3.5) {$\nu_2$};
          \node (nu1') at (17,10.5) {$\nu'_1$};
          \draw [black,line width=1pt,->,>=latex']
          (3,10) .. controls +(0,16) and +(0,8) .. (16,16);
          \draw [black,line width=1pt,<-,>=latex']
          (10,3) .. controls +(0,28) and +(0,10) .. (18,18);
        \end{tikzpicture}
        \label{subfig:no (nu_2, nu'_)-arc - 2}
      }
      \qquad
      \caption{\label{fig:no (nu_2, nu'_1)-arc}%
        Illustration of Claim~\ref{claim:no (nu_2, nu'_1)-arc, no (nu'_1, nu_2)-arc}.}
    \end{figure}

    \begin{Claim}
      \label{claim:no (nu_2, nu'_1)-arc, no (nu'_1, nu_2)-arc}
      There is neither a $(\nu_2, \nu'_1)$-arc
      nor a $(\nu'_1, \nu_2)$-arc in $\DMATCHING$.
    \end{Claim}

    \begin{proof}[Proof of Claim~\ref{claim:no (nu_2, nu'_1)-arc, no (nu'_1, nu_2)-arc}]
      First, according to
      Corollary~\ref{corollary:at most one arc monotone}--\ref{item:at most one arc monotone 2},
      there exists either
      at most one $(\nu_2, \nu'_1)$-arc and no $(\nu'_1, \nu_2)$-arc,
      or
      at most one $(\nu'_1, \nu_2)$-arc and no $(\nu_2, \nu'_1)$-arc
      $(i, i')$ in $\DMATCHING$ (see Figure~\ref{fig:no (nu_2, nu'_1)-arc}).
      Now from Claim~\ref{claim:(nu_1, nu'_1)-arc}, there exists at least
      one $(\nu_1, \nu'_1)$-arc, say $(j, j')$, in $\DMATCHING$.
      Hence,
      since $\DMATCHING$ is containment-free
      (\emph{i.e.}, it avoids the unlabeled patterns of
      $\mathcal{P}_{\mathrm{cont}}$),
      $\DMATCHING$ contains one of the following labeled patterns:
      $\LabeledCrossingRR{.08}{2}{1}{3}{4}$ and
      $\LabeledCrossingRL{.08}{2}{1}{3}{4}$.
      Applying Lemma~\ref{lemma:Forbidden patterns crossing}
      yields the sought-after contradiction.
    \end{proof}

    \begin{Claim}
      \label{claim:(nu_2, nu'_2)-arc}
      There is at least one $(\nu_2, \nu'_2)$-arc in $\DMATCHING$.
    \end{Claim}

    \begin{proof}[Proof of Claim~\ref{claim:(nu_2, nu'_2)-arc}]
      First,
      according to Claim~\ref{claim:(nu_1, nu'_1)-arc}, there exists at least
      one $(\nu_1, \nu'_1)$-arc in $\DMATCHING$ and hence,
      since $\DMATCHING$ avoids the unlabeled pattern \CrossingRL\, (Property~$\mathbf{P_1}$)
      there is no $(\nu'_2, \nu_2)$-arc in $\DMATCHING$.
      Now, suppose, aiming at a contradiction, that there is no
      $(\nu_2, \nu'_2)$-arc in $\DMATCHING$.
      Notice that there is neither
      a $(\nu_1, \nu_2)$-arc (Claim~\ref{claim:no (nu_1, nu_2)-arc, no (nu_2, nu_1)-arc})
      nor a $(\nu_2, \nu_1)$-arc (Claim~\ref{claim:no (nu_1, nu_2)-arc, no (nu_2, nu_1)-arc})
      nor a $(\nu_2, \nu_2)$-arc (Claim~\ref{claim:no (nu_2, nu_2)-arc})
      nor a $(\nu_2, \nu'_1)$-arc (Claim~\ref{claim:no (nu_2, nu'_1)-arc, no (nu'_1, nu_2)-arc})
      nor a $(\nu'_1, \nu_2)$-arc (Claim~\ref{claim:no (nu_2, nu'_1)-arc, no (nu'_1, nu_2)-arc})
      in $\DMATCHING$.
      But $|\nu_2| = N_2 >
      |\nu_3| + |\sigma'| + |\nu_4| + |\nu'_3| + |\pi'| + |\nu'_4|
      + |\pi''| + |\sigma''|$.
      Hence $\DMATCHING$ cannot be a directed perfect matching,
      thereby contradicting our hypothesis about~$\DMATCHING$.
    \end{proof}

    \begin{Claim}
      \label{claim:no arc below (nu_2, nu'_2)-arc}
      There is
      neither a $(\nu'_1, \nu'_1)$-arc,
      nor a $(\nu'_1, \nu_3)$-arc,
      nor a $(\nu_3, \nu'_1)$-arc,
      nor a $(\nu'_1, \sigma')$-arc,
      nor a $(\sigma', \nu'_1)$-arc,
      nor a $(\nu'_1, \nu_4)$-arc
      nor a $(\nu_4, \nu'_1)$-arc
      nor a $(\nu_3, \nu_3)$-arc,
      nor a $(\nu_3, \sigma')$-arc,
      nor a $(\sigma', \nu_3)$-arc,
      nor a $(\nu_3, \nu_4)$-arc,
      nor a $(\nu_4, \nu_3)$-arc,
      nor a $(\sigma', \sigma')$-arc,
      nor a $(\sigma', \nu_4)$-arc,
      nor a $(\nu_4, \sigma')$-arc,
      nor a $(\nu_4, \nu_4)$-arc
      in $\DMATCHING$.
    \end{Claim}

    \begin{proof}[Proof of Claim~\ref{claim:no arc below (nu_2, nu'_2)-arc}]
      Combine Claim~\ref{claim:(nu_2, nu'_2)-arc} with the fact that
      $\DMATCHING$ is containment-free
      (\emph{i.e.}, it avoids the unlabeled patterns of
      $\mathcal{P}_{\mathrm{cont}}$).
    \end{proof}

    \begin{figure}[t!]
      \centering
      \subfigure[]{%
        \begin{tikzpicture}
          [
            scale=0.2,
            label/.style={anchor=base}
          ]
          \draw[step=1cm,black!75,ultra thin,fill=black!50] (-0.2,6.8) grid (6.2,13.2);
          \foreach \x/\y in {0/7,1/8,2/9,3/10,4/11,5/12,6/13} {
            \draw [fill=black] (\x,\y) circle (0.2);
          }
          \draw[step=1cm,black!75,ultra thin,fill=black!50] (6.8,-0.2) grid (13.2,6.2);
          \foreach \x/\y in {7/6,8/5,9/4,10/3,11/2,12/1,13/0} {
            \draw [fill=black] (\x,\y) circle (0.2);
          }
          \draw [fill=black] (16,16) circle (0.2);
          \node [label] (i) at (9,-2) {$i$};
          \node [label] (ip) at (16,14) {$i'$};
          \node [anchor=west] (where) at (15,12) {$\nu'_3$, $\pi'$, $\nu'_4$, $\pi''$ or $\sigma''$};
          \node [label] (j) at (3,5) {$j$};
          \node [label] (jp) at (11,-2) {$j'$};
          \node (nu1) at (3,3.5) {$\nu_2$};
          \node (nu2) at (10,-3.5) {$\nu'_2$};
          \draw [black,line width=1pt,->,>=latex']
          (3,10) .. controls +(0,12) and +(0,16) .. (11,2);
          \draw [black,line width=1pt,->,>=latex']
          (9,4) .. controls +(0,16) and +(0,6) .. (16,16);
        \end{tikzpicture}
        \label{subfig:no (nu'_2, ?)-arc - 1}
      }
      \qquad
      \subfigure[]{%
        \begin{tikzpicture}
          [
            scale=0.2,
            label/.style={anchor=base}
          ]
          \draw[step=1cm,black!75,ultra thin,fill=black!50] (-0.2,6.8) grid (6.2,13.2);
          \foreach \x/\y in {0/7,1/8,2/9,3/10,4/11,5/12,6/13} {
            \draw [fill=black] (\x,\y) circle (0.2);
          }
          \draw[step=1cm,black!75,ultra thin,fill=black!50] (6.8,-0.2) grid (13.2,6.2);
          \foreach \x/\y in {7/6,8/5,9/4,10/3,11/2,12/1,13/0} {
            \draw [fill=black] (\x,\y) circle (0.2);
          }
          \draw [fill=black] (16,16) circle (0.2);
          \node [label,align=left] (i) at (11,-2) {$i$};
          \node [label] (ip) at (16,14) {$i'$};
          \node [anchor=west] (where) at (15,12) {$\nu'_3$, $\pi'$, $\nu'_4$, $\pi''$ or $\sigma''$};
          \node [label] (j) at (3,5) {$j$};
          \node [label] (jp) at (9,-2) {$j'$};
          \node (nu1) at (3,3.5) {$\nu_2$};
          \node (nu2) at (10,-3.5) {$\nu'_2$};
          \draw [black,line width=1pt,->,>=latex']
          (3,10) .. controls +(0,12) and +(0,16) .. (9,4);
          \draw [black,line width=1pt,->,>=latex']
          (11,2) .. controls +(0,16) and +(0,6) .. (16,16);
        \end{tikzpicture}
        \label{subfig:no (nu'_2, ?)-arc - 2}
      }
      \qquad\qquad\qquad
      \subfigure[]{%
        \begin{tikzpicture}
          [
            scale=0.2,
            label/.style={anchor=base}
          ]
          \draw[step=1cm,black!75,ultra thin,fill=black!50] (-0.2,6.8) grid (6.2,13.2);
          \foreach \x/\y in {0/7,1/8,2/9,3/10,4/11,5/12,6/13} {
            \draw [fill=black] (\x,\y) circle (0.2);
          }
          \draw[step=1cm,black!75,ultra thin,fill=black!50] (6.8,-0.2) grid (13.2,6.2);
          \foreach \x/\y in {7/6,8/5,9/4,10/3,11/2,12/1,13/0} {
            \draw [fill=black] (\x,\y) circle (0.2);
          }
          \draw [fill=black] (16,16) circle (0.2);
          \node [label] (ip) at (9,-2) {$i'$};
          \node [label] (i) at (16,14) {$i$};
          \node [anchor=west] (where) at (15,12) {$\nu'_3$, $\pi'$, $\nu'_4$, $\pi''$ or $\sigma''$};
          \node [label] (j) at (3,5) {$j$};
          \node [label] (jp) at (11,-2) {$j'$};
          \node (nu1) at (3,3.5) {$\nu_2$};
          \node (nu2) at (10,-3.5) {$\nu'_2$};
          \draw [black,line width=1pt,->,>=latex']
          (3,10) .. controls +(0,12) and +(0,16) .. (11,2);
          \draw [black,line width=1pt,<-,>=latex']
          (9,4) .. controls +(0,16) and +(0,6) .. (16,16);
        \end{tikzpicture}
        \label{subfig:no (nu'_2, ?)-arc - 3}
      }
      \qquad
      \subfigure[]{%
        \begin{tikzpicture}
          [
            scale=0.2,
            label/.style={anchor=base}
          ]
          \draw[step=1cm,black!75,ultra thin,fill=black!50] (-0.2,6.8) grid (6.2,13.2);
          \foreach \x/\y in {0/7,1/8,2/9,3/10,4/11,5/12,6/13} {
            \draw [fill=black] (\x,\y) circle (0.2);
          }
          \draw[step=1cm,black!75,ultra thin,fill=black!50] (6.8,-0.2) grid (13.2,6.2);
          \foreach \x/\y in {7/6,8/5,9/4,10/3,11/2,12/1,13/0} {
            \draw [fill=black] (\x,\y) circle (0.2);
          }
          \draw [fill=black] (16,16) circle (0.2);
          \node [label] (ip) at (11,-2) {$i'$};
          \node [label] (i) at (16,14) {$i$};
          \node [anchor=west] (where) at (15,12) {$\nu'_3$, $\pi'$, $\nu'_4$, $\pi''$ or $\sigma''$};
          \node [label] (j) at (3,5) {$j$};
          \node [label] (jp) at (9,-2) {$j'$};
          \node (nu1) at (3,3.5) {$\nu_2$};
          \node (nu2) at (10,-3.5) {$\nu'_2$};
          \draw [black,line width=1pt,->,>=latex']
          (3,10) .. controls +(0,12) and +(0,16) .. (9,4);
          \draw [black,line width=1pt,<-,>=latex']
          (11,2) .. controls +(0,16) and +(0,6) .. (16,16);
        \end{tikzpicture}
        \label{subfig:no (nu'_2, ?)-arc - 4}
      }
      \caption{\label{fig:subfig:no (nu_'2, ?)-arc}%
        Illustration of Claim~\ref{claim: no arc right of nu'_2}.
      }
    \end{figure}
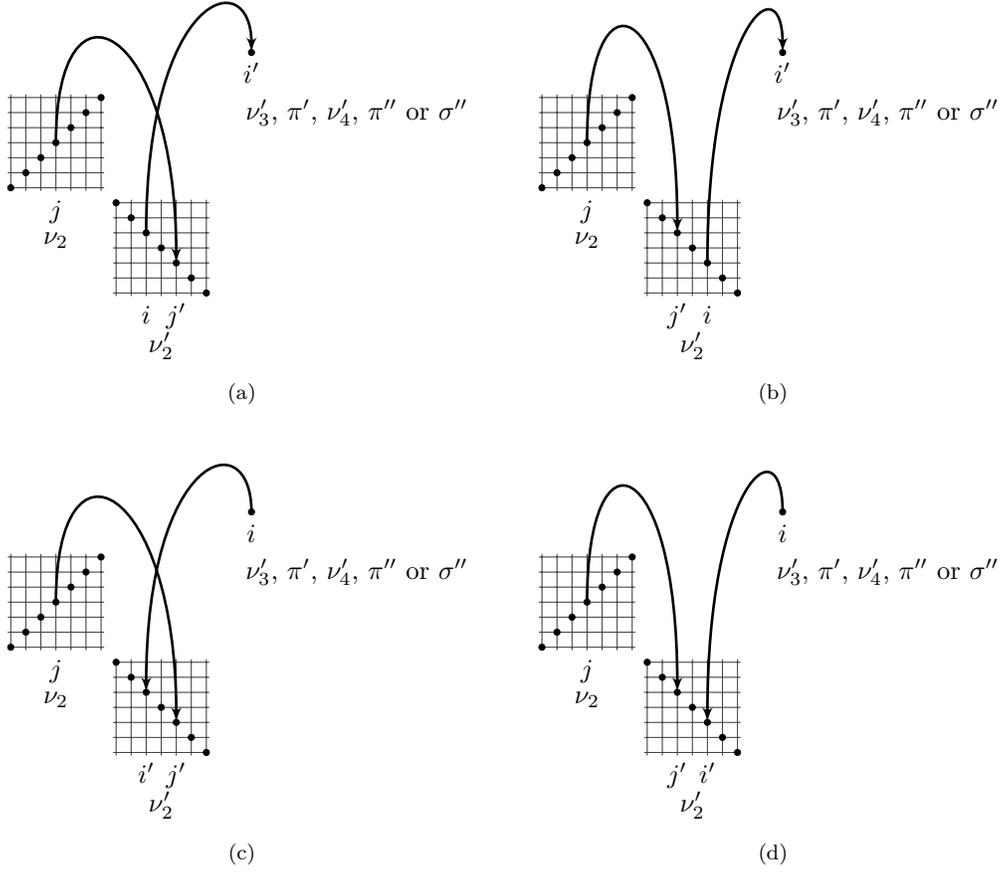

    \begin{Claim}
      \label{claim: no arc right of nu'_2}
      There is
      neither a $(\nu'_2, \nu'_3)$-arc,
      nor a $(\nu'_3, \nu'_2)$-arc,
      nor a $(\nu'_2, \pi')$-arc,
      nor a $(\pi', \nu'_2)$-arc,
      nor a $(\nu'_2, \nu'_4)$-arc,
      nor a $(\nu'_4, \nu'_2)$-arc,
      nor a $(\nu'_2, \pi'')$-arc,
      nor a $(\pi'', \nu'_2)$-arc
      nor a $(\nu'_2, \sigma'')$-arc,
      nor a $(\sigma'', \nu'_2)$-arc
      in $\DMATCHING$.
    \end{Claim}

    \begin{proof}[Proof of Claim~\ref{claim: no arc right of nu'_2}]
      Suppose aiming at a contradiction that $\DMATCHING$ contains
      a $(\nu'_2, \nu'_3)$-arc,
      a $(\nu'_3, \nu'_2)$-arc,
      a $(\nu'_2, \pi')$-arc,
      a $(\pi', \nu'_2)$-arc,
      a $(\nu'_2, \nu'_4)$-arc,
      a $(\nu'_4, \nu'_2)$-arc,
      a $(\nu'_2, \pi'')$-arc,
      a $(\pi'', \nu'_2)$-arc
      a $(\nu'_2, \sigma'')$-arc or
      a $(\sigma'', \nu'_2)$-arc, say $(i, i')$.
      We now observe that
      $\nu'_3$, $\pi'$, $\nu'_4$, $\pi''$ and $\sigma''$ are all
      right above of both $\nu_2$ and $\nu'_2$.
      Furthermore, according to Claim~\ref{claim:(nu_2, nu'_2)-arc},
      there exists a $(\nu_2, \nu'_2)$-arc, say $(j, j')$.
      Then, it follow that $\DMATCHING$ contains one of the following
      labeled patterns:
      $\LabeledCrossingRR{.08}{3}{2}{1}{4}$,
      $\LabeledPrecedenceRR{.08}{3}{2}{1}{4}$,
      $\LabeledCrossingRL{.08}{3}{2}{1}{4}$ and
      $\LabeledPrecedenceRL{.08}{3}{2}{1}{4}$
      (see Figure~\ref{fig:subfig:no (nu_'2, ?)-arc}).
      Applying Lemma~\ref{lemma:Forbidden patterns precedence}
      and Lemma~\ref{lemma:Forbidden patterns crossing}
      yields the sought-after contradiction.
    \end{proof}

    \begin{figure}[t!]
      \centering
      \begin{tikzpicture}
        [
          scale=0.2,
          label/.style={anchor=base}
        ]
        \draw[step=1cm,black!75,ultra thin,fill=black!50] (-0.2,6.8) grid (6.2,13.2);
        \foreach \x/\y in {0/13,1/12,2/11,3/10,4/9,5/8,6/7} {
          \draw [fill=black] (\x,\y) circle (0.2);
        }
        \draw[step=1cm,black!75,ultra thin,fill=black!50] (6.8,13.8) grid (13.2,20.2);
        \foreach \x/\y in {7/14,8/15,9/16,10/17,11/18,12/19,13/20} {
          \draw [fill=black] (\x,\y) circle (0.2);
        }
        \draw[step=1cm,black!75,ultra thin,fill=black!50] (13.8,-0.2) grid (20.2,6.2);
        \foreach \x/\y in {14/6,15/5,16/4,17/3,18/2,19/1,20/0} {
          \draw [fill=black] (\x,\y) circle (0.2);
        }
        \node [label] (i) at (3,5) {$i$};
        \node [label] (ip) at (16,-2) {$i'$};
        \node [label] (j) at (10,12) {$j$};
        \node [label] (jp) at (18,-2) {$j'$};
        \node (nu1) at (3,3.5) {$\nu_2$};
        \node (nu3) at (10,10.5) {$\nu_3$};
        \node (nu2') at (17,-3.5) {$\nu'_2$};
        \draw [black,line width=1pt,->,>=latex']
        (3,10) .. controls +(0,20) and +(0,30) .. (16,4);
        \draw [black,line width=1pt,->,>=latex']
        (10,17) .. controls +(0,12) and +(0,26) .. (18,2);
      \end{tikzpicture}
      \caption{\label{fig:no (nu_3, nu'_2)-arc, no (nu'_3, nu_3)-arc}%
        Illustration of Claim~\ref{claim:no (nu_3, nu'_2)-arc, no (nu'_2, nu_3)-arc}.
      }
    \end{figure}
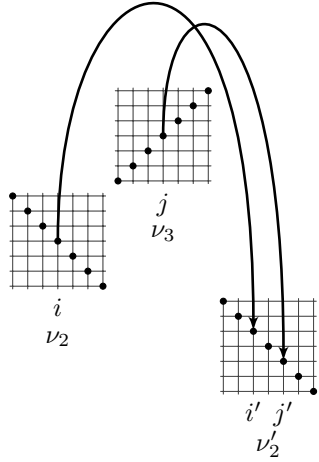

    \begin{Claim}
      \label{claim:no (nu_3, nu'_2)-arc, no (nu'_2, nu_3)-arc}
      There is neither a $(\nu_3, \nu'_2)$-arc
      nor a $(\nu'_2, \nu_3)$-arc in $\DMATCHING$.
    \end{Claim}

    \begin{proof}[Proof of Claim~\ref{claim:no (nu_3, nu'_2)-arc, no (nu'_2, nu_3)-arc}]
      Suppose, aiming at a contradiction, that there exists
      $(\nu_3, \nu'_2)$-arc or a $(\nu'_2, \nu_3)$-arc
      $(j, j')$ in $\DMATCHING$.
      According to Claim~\ref{claim:(nu_2, nu'_2)-arc}, there
      exists at least one $(\nu_2, \nu'_2)$-arc $(i, i')$ in $\DMATCHING$.
      Since $\DMATCHING$ avoids the unlabeled pattern $\CrossingRL$
      (Property~$\mathbf{P_1}$),
      there is no $(\nu'_2, \nu_3)$-arc in $\DMATCHING$
      (see Figure~\ref{fig:no (nu_3, nu'_2)-arc, no (nu'_3, nu_3)-arc}),
      and hence $(j, j')$ is a $(\nu_3, \nu'_2)$-arc.
      Then it follows that $\DMATCHING$ contains the labeled pattern
      $\LabeledCrossingRR{.08}{3}{4}{2}{1}$
      (see Figure~\ref{fig:no (nu_3, nu'_2)-arc, no (nu'_3, nu_3)-arc}).
      Applying Lemma~\ref{lemma:Forbidden patterns crossing}
      yields the sought-after contradiction.
    \end{proof}

    \begin{Claim}
      \label{claim:at most one (nu_2, nu_3)-arc or at most one (nu_3, nu_2)-arc}
      There is at most one $(\nu_2, \nu_3)$-arc or at most one $(\nu_3, \nu_2)$-arc
      in $\DMATCHING$.
    \end{Claim}

    \begin{proof}[Proof of Claim~\ref{claim:at most one (nu_2, nu_3)-arc or at most one (nu_3, nu_2)-arc}]
      Apply Corollary~\ref{corollary:at most one arc monotone}.
    \end{proof}

    We will see soon (upcoming Claim~\ref{claim:all (nu_2, nu'_2)-arcs})
    that there exists actually no $(\nu_2, \nu_3)$-arc
    in $\DMATCHING$.

    \begin{figure}[t!]
      \centering
      \begin{tikzpicture}
        [
          scale=0.2,
          label/.style={anchor=base}
        ]
        \draw[step=1cm,black!75,ultra thin,fill=black!50] (-0.2,6.8) grid (6.2,13.2);
        \foreach \x/\y in {0/7,1/8,2/9,3/10,4/11,5/12,6/13} {
          \draw [fill=black] (\x,\y) circle (0.2);
        }
        \draw[step=1cm,black!75,ultra thin,fill=black!50] (6.8,13.8) grid (13.2,20.2);
        \foreach \x/\y in {7/14,8/15,9/16,10/17,11/18,12/19,13/20} {
          \draw [fill=black] (\x,\y) circle (0.2);
        }
        \draw[step=1cm,black!75,ultra thin,fill=black!50] (13.8,-0.2) grid (20.2,6.2);
        \foreach \x/\y in {14/0,15/1,16/2,17/3,18/4,19/5,20/6} {
          \draw [fill=black] (\x,\y) circle (0.2);
        }
        \node [label] (i) at (2,5) {$i$};
        \node [label] (ip) at (10,12) {$i'$};
        \node [label] (j) at (4,5) {$j$};
        \node [label] (jp) at (17,-2) {$j'$};
        \node (nu1) at (3,3.5) {$\nu_1$};
        \node (nu3) at (10,10.5) {$\nu'_1$};
        \node (nu2') at (17,-3.5) {$\nu_3$};
        \draw [black,line width=1pt,->,>=latex']
        (2,9) .. controls +(0,16) and +(0,6) .. (10,17);
        \draw [black,line width=1pt,->,>=latex']
        (4,11) .. controls +(0,14) and +(0,30) .. (17,3);
      \end{tikzpicture}
      \caption{\label{fig:no (nu_1, nu_3)-arc, no (nu_3, nu_1)-arc}%
        Illustration of Claim~\ref{claim:no (nu_1, nu_3)-arc, no (nu_3, nu_1)-arc}
      }
    \end{figure}
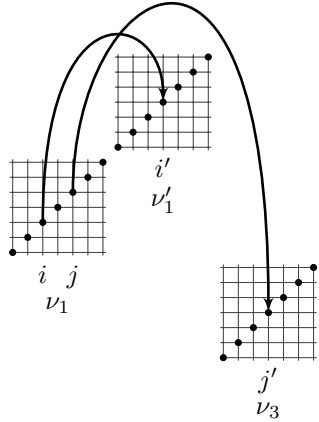

    \begin{Claim}
      \label{claim:no (nu_1, nu_3)-arc, no (nu_3, nu_1)-arc}
      There is neither a $(\nu_1, \nu_3)$-arc, nor a $(\nu_3, \nu_1)$-arc
      in $\DMATCHING$.
    \end{Claim}

    \begin{proof}[Proof of Claim~\ref{claim:no (nu_1, nu_3)-arc, no (nu_3, nu_1)-arc}]
      Suppose, aiming at a contradiction, that there exists
      a $(\nu_1, \nu_3)$-arc or a $(\nu_3, \nu_1)$-arc, say $(j, j')$, in $\DMATCHING$.
      According to Claim~\ref{claim:(nu_1, nu'_1)-arc}, there
      exists at least one $(\nu_1, \nu'_1)$-arc, say $(i, i')$, in $\DMATCHING$.
      Since $\DMATCHING$ avoids the unlabeled pattern $CrossingRL$
      (Property~$\mathbf{P_1}$),
      there is no $(\nu_3, \nu_1)$-arc in $\DMATCHING$
      (see Figure~\ref{fig:no (nu_1, nu_3)-arc, no (nu_3, nu_1)-arc}),
      and hence $(j, j')$ is a $(\nu_1, \nu_3)$-arc.
      Then it follows that $\DMATCHING$ contains the labeled pattern
      $\LabeledCrossingRR{.08}{2}{3}{4}{1}$
      (see Figure~\ref{fig:no (nu_1, nu_3)-arc, no (nu_3, nu_1)-arc}).
      Applying Lemma~\ref{lemma:Forbidden patterns crossing}
      yields the sought-after contradiction.
    \end{proof}

    \begin{Claim}
      \label{claim:at least one (nu_3, nu'_3)-arc}
      There exists a $(\nu_3, \nu'_3)$-arc in $\DMATCHING$.
    \end{Claim}

    \begin{proof}[Proof of Claim~\ref{claim:at least one (nu_3, nu'_3)-arc}]
      First, according to Claim~\ref{claim:(nu_2, nu'_2)-arc}, there
      exists at least one $(\nu_2, \nu'_2)$-arc in $\DMATCHING$.
      Since $\DMATCHING$ avoids the unlabeled pattern \CrossingRL\,
      (Property~$\mathbf{P_1}$)
      there is no $(\nu'_3, \nu_3)$-arc in $\DMATCHING$.
      Now, suppose, aiming at a contradiction, that there is no
      $(\nu_3, \nu'_3)$-arc in $\DMATCHING$.
      Combining
      Claim~\ref{claim:no arc below (nu_2, nu'_2)-arc},
      Claim~\ref{claim:no (nu_3, nu'_2)-arc, no (nu'_2, nu_3)-arc},
      Claim~\ref{claim:at most one (nu_2, nu_3)-arc or at most one (nu_3, nu_2)-arc}
      Claim~\ref{claim:no (nu_1, nu_3)-arc, no (nu_3, nu_1)-arc}
      together with our hypothesis,
      we conclude that $N_3-1$ positions in $\nu_3$ are involved
      in arcs of $\DMATCHING$ that are
      neither ($\nu_1, \nu_3)$-arcs,
      nor ($\nu_3, \nu_1)$-arcs,
      nor $(\nu_2, \nu_3)$-arcs,
      nor $(\nu_3, \nu_2)$-arcs,
      nor $(\nu'_1, \nu_3)$-arcs,
      nor $(\nu_3, \nu'_1)$-arcs,
      nor $(\nu_3, \nu_3)$-arcs,
      nor $(\nu_3, \sigma')$-arcs,
      nor $(\sigma', \nu_3)$-arcs,
      nor $(\nu_3, \nu_4)$-arcs,
      nor $(\nu_4, \nu_3)$-arcs,
      nor $(\nu_3, \nu'_3)$-arcs,
      nor $(\nu'_3, \nu_3)$-arcs.
      But $N_3 - 1 >
      |\pi'| + |\nu'_4| + |\pi''| + |\sigma''|
      = N_4 + 2n + 2k + 4$, and hence
      $\DMATCHING$ is not a perfect matching.
      This is the sought-after contradiction.
    \end{proof}

    \begin{Claim}
      \label{claim:no arc below one (nu_3, nu'_3)-arc}
      There is
      neither a $(\sigma', \sigma')$-arc,
      nor a $(\sigma', \nu_4)$-arc,
      nor a $(\nu_4, \sigma')$-arc,
      nor a $(\sigma', \nu'_2)$-arc,
      nor a $(\nu'_2, \sigma')$-arc,
      nor a $(\nu_4, \nu_4)$-arc
      nor a $(\nu_4, \nu'_2)$-arc,
      nor a $(\nu'_2, \nu_4)$-arc,
      nor a $(\nu'_2, \nu'_2)$-arc
      in $\DMATCHING$.
    \end{Claim}

    \begin{proof}[Proof of Claim~\ref{claim:no arc below one (nu_3, nu'_3)-arc}]
      Combine Claim~\ref{claim:at least one (nu_3, nu'_3)-arc} together with the fact that
      $\DMATCHING$ is containment-free
      (\emph{i.e.}, it avoids the unlabeled patterns
      of $\mathcal{P}_{\mathrm{cont}}$).
    \end{proof}

    The next two claims state that $\DMATCHING$ actually contains
    all $(\nu_1, \nu'_1)$-arcs and all $(\nu_2, \nu'_2)$-arcs.

    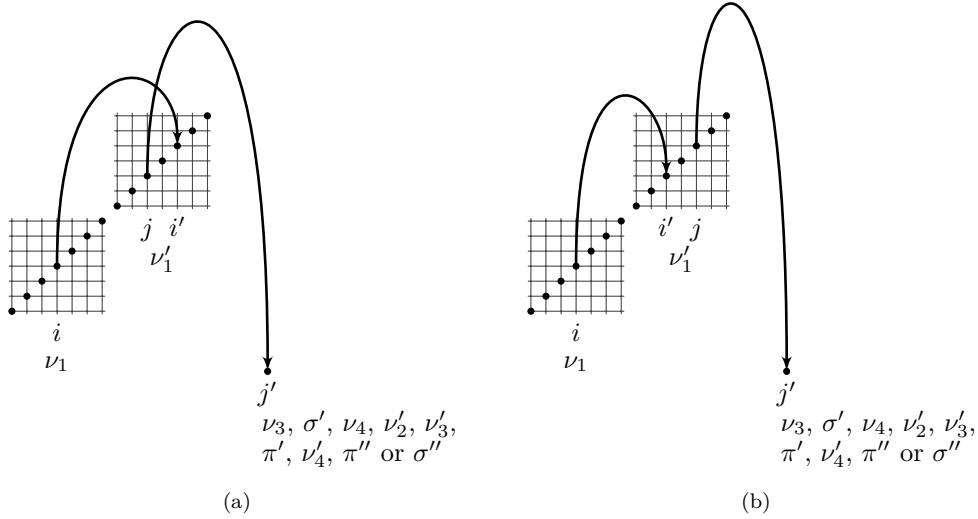
\begin{figure}[t!]
      \centering
      \subfigure[]{%
        \begin{tikzpicture}
          [
            scale=0.2,
            label/.style={anchor=base}
          ]
          \draw[step=1cm,black!75,ultra thin,fill=black!50] (-0.2,6.8) grid (6.2,13.2);
          \foreach \x/\y in {0/7,1/8,2/9,3/10,4/11,5/12,6/13} {
            \draw [fill=black] (\x,\y) circle (0.2);
          }
          \draw[step=1cm,black!75,ultra thin,fill=black!50] (6.8,13.8) grid (13.2,20.2);
          \foreach \x/\y in {7/14,8/15,9/16,10/17,11/18,12/19,13/20} {
            \draw [fill=black] (\x,\y) circle (0.2);
          }
          \draw [fill=black] (17,3) circle (0.2);
          \node [label] (i) at (3,5) {$i$};
          \node [label] (ip) at (11,12) {$i'$};
          \node [label] (j) at (9,12) {$j$};
          \node [label] (jp) at (17,1) {$j'$};
          \node (nu1) at (3,3.5) {$\nu_1$};
          \node (nu3) at (10,10.5) {$\nu'_1$};
          \node [anchor=west] (where) at (16,-.5) {$\nu_3$, $\sigma'$, $\nu_4$, $\nu'_2$, $\nu'_3$,};
          \node [anchor=west] (where) at (16,-2.5) {$\pi'$, $\nu'_4$, $\pi''$ or $\sigma''$};
          \draw [black,line width=1pt,->,>=latex']
          (3,10) .. controls +(0,16) and +(0,6) .. (11,18);
          \draw [black,line width=1pt,->,>=latex']
          (9,16) .. controls +(0,14) and +(0,30) .. (17,3);
        \end{tikzpicture}
        \label{subfig:all (nu_1, nu'_1)-arcs - 1}
      }
      \qquad
      \subfigure[]{%
        \begin{tikzpicture}
          [
            scale=0.2,
            label/.style={anchor=base}
          ]
          \draw[step=1cm,black!75,ultra thin,fill=black!50] (-0.2,6.8) grid (6.2,13.2);
          \foreach \x/\y in {0/7,1/8,2/9,3/10,4/11,5/12,6/13} {
            \draw [fill=black] (\x,\y) circle (0.2);
          }
          \draw[step=1cm,black!75,ultra thin,fill=black!50] (6.8,13.8) grid (13.2,20.2);
          \foreach \x/\y in {7/14,8/15,9/16,10/17,11/18,12/19,13/20} {
            \draw [fill=black] (\x,\y) circle (0.2);
          }
          \draw [fill=black] (17,3) circle (0.2);
          \node [label] (i) at (3,5) {$i$};
          \node [label] (ip) at (9,12) {$i'$};
          \node [label] (j) at (11,12) {$j$};
          \node [label] (jp) at (17,1) {$j'$};
          \node (nu1) at (3,3.5) {$\nu_1$};
          \node (nu3) at (10,10.5) {$\nu'_1$};
          \node [anchor=west] (where) at (16,-.5) {$\nu_3$, $\sigma'$, $\nu_4$, $\nu'_2$, $\nu'_3$,};
          \node [anchor=west] (where) at (16,-2.5) {$\pi'$, $\nu'_4$, $\pi''$ or $\sigma''$};
          \draw [black,line width=1pt,->,>=latex']
          (3,10) .. controls +(0,16) and +(0,6) .. (9,16);
          \draw [black,line width=1pt,->,>=latex']
          (11,18) .. controls +(0,14) and +(0,30) .. (17,3);
        \end{tikzpicture}
        \label{subfig:all (nu_1, nu'_1)-arcs - 2}
      }
      \caption{\label{fig:all (nu_1, nu'_1)-arcs}%
        Illustration of Claim~\ref{claim:all (nu_1, nu'_1)-arcs}.
      }
    \end{figure}

    \begin{Claim}
      \label{claim:all (nu_1, nu'_1)-arcs}
      $\DMATCHING$ contains $N_1$ pairwise crossing $(\nu_1, \nu'_1)$-arcs.
    \end{Claim}

    \begin{proof}[Proof of Claim~\ref{claim:all (nu_1, nu'_1)-arcs}]
      First, according to Claim~\ref{claim:(nu_2, nu'_2)-arc},
      $\DMATCHING$ contains at least one $(\nu_1, \nu'_1)$-arc.
      Now, suppose, aiming at a contradiction, that
      $\DMATCHING$ does not contain $N_1$ $(\nu_1, \nu'_1)$-arcs.
      Combining
      Claim~\ref{claim:no (nu_1, nu_2)-arc, no (nu_2, nu_1)-arc},
      Claim~\ref{claim:no (nu_2, nu'_1)-arc, no (nu'_1, nu_2)-arc},
      Claim~\ref{claim:no (nu_2, nu'_1)-arc, no (nu'_1, nu_2)-arc} and
      Claim~\ref{claim:no arc below (nu_2, nu'_2)-arc},
      we conclude that $\DMATCHING$ contains one of the two following
      labeled patterns:
      $\LabeledCrossingRR{.08}{2}{3}{4}{1}$ and
      $\LabeledPrecedenceRR{.08}{2}{3}{4}{1}$
      (see Figure~\ref{fig:all (nu_1, nu'_1)-arcs}).
      Applying Lemma~\ref{lemma:Forbidden patterns crossing} or
      Lemma~\ref{lemma:Forbidden patterns precedence}
      yields the sought-after contradiction.
    \end{proof}

    \begin{Claim}
      \label{claim:all (nu_2, nu'_2)-arcs}
      $\DMATCHING$ contains $N_2$ pairwise crossing $(\nu_2, \nu'_2)$-arcs.
    \end{Claim}

    \begin{proof}[Proof of Claim~\ref{claim:all (nu_2, nu'_2)-arcs}]
      The key idea is to focus on $\nu'_2$
      and combine
      Claim~\ref{claim:(nu_2, nu'_2)-arc},
      Claim~\ref{claim:no (nu_3, nu'_2)-arc, no (nu'_2, nu_3)-arc},
      Claim~\ref{claim: no arc right of nu'_2} and
      Claim~\ref{claim:no arc below one (nu_3, nu'_3)-arc}.
    \end{proof}

    \begin{Claim}
      \label{claim:no (nu_2, nu_4)-arc, no (nu_4, nu_2)-arc}
      There is neither a $(\nu_2, \nu_4)$-arc nor a
      $(\nu_4, \nu_2)$-arc in $\DMATCHING$.
    \end{Claim}

    \begin{proof}[Proof of Claim~\ref{claim:no (nu_2, nu_4)-arc, no (nu_4, nu_2)-arc}]
      According to Claim~\ref{claim:all (nu_2, nu'_2)-arcs},
      all positions in $\nu_2$ and $\nu'_2$ are involved in
      $(\nu_2, \nu'_2)$-arcs in $\DMATCHING$.
    \end{proof}

    \begin{figure}[t!]
      \centering
      \subfigure[]{%
        \begin{tikzpicture}
          [
            scale=0.2,
            label/.style={anchor=base}
          ]
          \draw[step=1cm,black!75,ultra thin,fill=black!50] (-0.2,6.8) grid (6.2,13.2);
          \foreach \x/\y in {0/7,1/8,2/9,3/10,4/11,5/12,6/13} {
            \draw [fill=black] (\x,\y) circle (0.2);
          }
          \draw[step=1cm,black!75,ultra thin,fill=black!50] (6.8,-0.2) grid (13.2,6.2);
          \foreach \x/\y in {7/6,8/5,9/4,10/3,11/2,12/1,13/0} {
            \draw [fill=black] (\x,\y) circle (0.2);
          }
          \draw[step=1cm,black!75,ultra thin,fill=black!50] (13.8,13.8) grid (20.2,20.2);
          \foreach \x/\y in {14/14,15/15,16/16,17/17,18/18,19/19,20/20} {
            \draw [fill=black] (\x,\y) circle (0.2);
          }
          \node [label] (j) at (3,5) {$j$};
          \node [label] (jp) at (16,12) {$j'$};
          \node [label] (i) at (10,-2) {$i$};
          \node [label] (ip) at (18,12) {$i'$};
          \node (nu1) at (3,3.5) {$\nu'_1$};
          \node (nu2) at (10,-3.5) {$\nu_4$};
          \node (nu1') at (17,10.5) {$\nu'_3$};
          \draw [black,line width=1pt,->,>=latex']
          (3,10) .. controls +(0,16) and +(0,8) .. (16,16);
          \draw [black,line width=1pt,->,>=latex']
          (10,3) .. controls +(0,28) and +(0,10) .. (18,18);
        \end{tikzpicture}
        \label{subfig:no (nu_4, nu'_3)-arc, no (nu'_3, nu_4)-arc - 1}
      }
      \qquad\qquad
      \subfigure[]{%
        \begin{tikzpicture}
          [
            scale=0.2,
            label/.style={anchor=base}
          ]
          \draw[step=1cm,black!75,ultra thin,fill=black!50] (-0.2,6.8) grid (6.2,13.2);
          \foreach \x/\y in {0/7,1/8,2/9,3/10,4/11,5/12,6/13} {
            \draw [fill=black] (\x,\y) circle (0.2);
          }
          \draw[step=1cm,black!75,ultra thin,fill=black!50] (6.8,-0.2) grid (13.2,6.2);
          \foreach \x/\y in {7/6,8/5,9/4,10/3,11/2,12/1,13/0} {
            \draw [fill=black] (\x,\y) circle (0.2);
          }
          \draw[step=1cm,black!75,ultra thin,fill=black!50] (13.8,13.8) grid (20.2,20.2);
          \foreach \x/\y in {14/14,15/15,16/16,17/17,18/18,19/19,20/20} {
            \draw [fill=black] (\x,\y) circle (0.2);
          }
          \node [label] (j) at (3,5) {$j$};
          \node [label] (jp) at (16,12) {$j'$};
          \node [label] (ip) at (10,-2) {$i'$};
          \node [label] (i) at (18,12) {$i$};
          \node (nu1) at (3,3.5) {$\nu'_1$};
          \node (nu2) at (10,-3.5) {$\nu_4$};
          \node (nu1') at (17,10.5) {$\nu'_3$};
          \draw [black,line width=1pt,->,>=latex']
          (3,10) .. controls +(0,16) and +(0,8) .. (16,16);
          \draw [black,line width=1pt,<-,>=latex']
          (10,3) .. controls +(0,28) and +(0,10) .. (18,18);
        \end{tikzpicture}
        \label{subfig:no (nu_4, nu'_3)-arc, no (nu'_3, nu_4)-arc - 2}
      }
      \qquad
      \caption{\label{fig:no (nu_4, nu'_3)-arc, no (nu'_3, nu_4)-arc}%
        Illustration of Claim~\ref{claim:no (nu_4, nu'_3)-arc, no (nu'_3, nu_4)-arc}.}
    \end{figure}
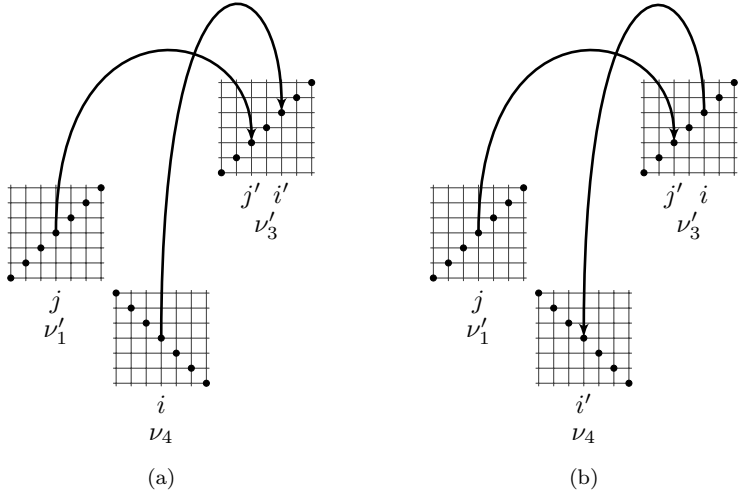

    \begin{Claim}
      \label{claim:no (nu_4, nu'_3)-arc, no (nu'_3, nu_4)-arc}
      There is neither a $(\nu_4, \nu'_3)$-arc
      nor a $(\nu'_3, \nu_4)$-arc in $\DMATCHING$.
    \end{Claim}

    \begin{proof}[Proof of Claim~\ref{claim:no (nu_4, nu'_3)-arc, no (nu'_3, nu_4)-arc}]
      First, according to
      Corollary~\ref{corollary:at most one arc monotone}--\ref{item:at most one arc monotone 2},
      there exists either
      at most one $(\nu_4, \nu'_3)$-arc and no $(\nu'_3, \nu_4)$-arc,
      or
      at most one $(\nu'_3, \nu_4)$-arc and no $(\nu_4, \nu'_3)$-arc
      $(i, i')$ in $\DMATCHING$ (see Figure~\ref{fig:no (nu_4, nu'_3)-arc, no (nu'_3, nu_4)-arc}).
      Now from Claim~\ref{claim:at least one (nu_3, nu'_3)-arc}, there exists at least
      one $(\nu_3, \nu'_3)$-arc, say $(j, j')$, in $\DMATCHING$.
      Hence,
      since $\DMATCHING$ is containment-free
      (\emph{i.e.}, it avoids the unlabeled patterns of
      $\mathcal{P}_{\mathrm{cont}}$),
      $\DMATCHING$ contains one of the two following labeled patterns:
      $\LabeledCrossingRR{.08}{2}{1}{3}{4}$
      and
      $\LabeledCrossingRL{.08}{2}{1}{3}{4}$.
      Applying Lemma~\ref{lemma:Forbidden patterns crossing}
      yields the sought-after contradiction.
    \end{proof}

    \begin{Claim}
      \label{claim:at least one (nu_4, nu'_4)-arc}
      There is at least one $(\nu_4, \nu'_4)$-arc in $\DMATCHING$.
    \end{Claim}

    \begin{proof}[Proof of Claim~\ref{claim:at least one (nu_4, nu'_4)-arc}]
      First, according to Claim~\ref{claim:at least one (nu_3, nu'_3)-arc},
      there is at least one $(\nu_3, \nu'_3$)-arc in $\DMATCHING$.
      Therefore, since $\DMATCHING$ avoids the unlabeled pattern $\CrossingRL$
      (Property~$\mathbf{P_1}$),
      there is no $(\nu'_4, \nu_4)$-arc in $\DMATCHING$.
      Now, suppose, aiming at a contradiction, that there is no
      $(\nu_4, \nu'_4)$-arc in $\DMATCHING$.
      First, according to
      Claim~\ref{claim:all (nu_1, nu'_1)-arcs} and
      Claim~\ref{claim:all (nu_2, nu'_2)-arcs},
      there is neither
      a $(\nu_1, \nu_4)$-arc
      nor a $(\nu_4, \nu_1)$-arc
      nor a $(\nu'_1, \nu_4)$-arc
      nor a $(\nu_2, \nu_4)$-arc
      nor a $(\nu_4, \nu_2)$-arc
      nor a $(\nu'_1, \nu_4)$-arc
      nor a $(\nu_4, \nu'_1)$-arc
      nor a $(\nu_3, \nu_4)$-arc
      not a $(\nu_4, \nu_3)$-arc
      nor a $(\sigma', \nu_4)$-arc
      nor a $(\nu_4, \sigma')$-arc
      nor a $(\nu_4, \nu_4)$-arc
      nor a $(\nu_4, \nu'_2)$-abstract
      nor a $(\nu'_2, \nu_4)$-arc in $\DMATCHING$.
      Furthermore, according to
      Claim~\ref{claim:no (nu_4, nu'_3)-arc, no (nu'_3, nu_4)-arc},
      there is neither
      a $(\nu_4, \nu'_3)$-arc
      nor a $\nu'_3, \nu_4)$-arc in $\DMATCHING$.
      But $N_4 > |\pi'| + |\pi''| + |\sigma''|$, and
      hence $\DMATCHING$ is not a direct perfect matching,
      which contradicts our hypothesis about $\DMATCHING$.
    \end{proof}

    \begin{Claim}
      \label{claim:no arc below (nu_4, nu'_4)-arc}
      There is neither
      a $(\pi', \pi')$-arc
      nor a $(\sigma', \pi'')$-arc
      nor a $(\pi'', \sigma')$-arc
      nor a $(\sigma', \sigma'')$-arc
      nor a $(\sigma'', \sigma')$-arc
      in $\DMATCHING$.
    \end{Claim}

    \begin{proof}[Proof of Claim~\ref{claim:no arc below (nu_4, nu'_4)-arc}]
      Combine Claim~\ref{claim:at least one (nu_4, nu'_4)-arc} together with the
      fact that $\DMATCHING$ has Property~$\mathbf{P_1}$ and hence is containment-free
      (\emph{i.e.}, it avoids the unlabeled patterns of $\mathcal{P}_{\mathrm{cont}}$).
    \end{proof}

    \begin{Claim}
      \label{claim:no (sigma', nu'_3)-arc, no (nu'_3, sigma')-arc}
      There is neither
      a $(\sigma', \nu'_3)$-arc
      nor a $(\nu'_3, \sigma')$-arc
      in $\DMATCHING$.
    \end{Claim}

    \begin{proof}[Proof of Claim~\ref{claim:no (sigma', nu'_3)-arc, no (nu'_3, sigma')-arc}]
      First, according to Claim~\ref{claim:at least one (nu_3, nu'_3)-arc},
      there is at least one $(\nu_3, \nu'_3$)-arc in $\DMATCHING$.
      Therefore, since $\DMATCHING$ avoids the unlabeled pattern $\CrossingRL$ (Property~$\mathbf{P_1}$),
      there is no $(\nu'_3, \sigma')$-arc in $\DMATCHING$.
      Now, suppose, aiming at a contradiction, that there is
      a $(\sigma', \nu'_3)$-arc in $\DMATCHING$.
      Hence,
      since $\DMATCHING$ is containment-free
      (\emph{i.e.}, it avoids the unlabeled patterns of
      $\mathcal{P}_{\mathrm{cont}}$),
      $\DMATCHING$ contains the labeled pattern
      $\LabeledCrossingRR{.08}{2}{1}{3}{4}$.
      Applying Lemma~\ref{lemma:Forbidden patterns crossing}
      yields the sought-after contradiction.
    \end{proof}

    \begin{Claim}
      \label{claim:no (sigma', nu'_4)-arc, no (nu'_4, sigma')-arc}
      There is neither
      a $(\sigma', \nu'_4)$-arc
      nor a $(\nu'_4, \sigma')$-arc
      in $\DMATCHING$.
    \end{Claim}

    \begin{proof}[Proof of Claim~\ref{claim:no (sigma', nu'_4)-arc, no (nu'_4, sigma')-arc}]
      First, according to Claim~\ref{claim:at least one (nu_4, nu'_4)-arc},
      there is at least one $(\nu_4, \nu'_4$)-arc in $\DMATCHING$.
      Therefore, since $\DMATCHING$
      is containment-free
      (\emph{i.e.}, it avoids the unlabeled patterns of
      $\mathcal{P}_{\mathrm{cont}}$),
      and avoids $\CrossingRL$ (Property~$\mathbf{P_1}$),
      there is no $(\nu'_4, \sigma')$-arc in $\DMATCHING$.
      Now, suppose, aiming at a contradiction, that there is
      a $(\sigma', \nu'_4)$-arc in $\DMATCHING$.
      Hence,
      $\DMATCHING$ contains the labeled pattern
      $\LabeledCrossingRR{.08}{3}{4}{2}{1}$
      Applying Lemma~\ref{lemma:Forbidden patterns crossing}
      yields the sought-after contradiction.
    \end{proof}

    \begin{Claim}
      \label{claim:no (pi', sigma', nu'_4)-arc}
      There is no $(\pi', \sigma')$-arc
      in $\DMATCHING$.
    \end{Claim}

    \begin{proof}[Proof of Claim~\ref{claim:no (pi', sigma', nu'_4)-arc}]
      Combine Claim~\ref{claim:at least one (nu_4, nu'_4)-arc} together with the
      fact that $\DMATCHING$ avoids the unlabeled pattern $\CrossingLR$
      (Property~$\mathbf{P_1}$).
    \end{proof}

    Combining the above claims, we conclude that there are
    $k+2$ $(\sigma', \pi')$-arcs in $\DMATCHING$.
    Recall that
    \begin{equation}
        \sigma' = ((k+1) \; \sigma \; (k+2)) \; [2N_2 + N_4 + 2n + k + 2]
    \end{equation}
    and that
    \begin{equation}
        \pi' = ((n+1) \; \pi \; (n+2)) \; [2N_2 + N_4 + n + k + 2].
    \end{equation}
    Then it follows we have
    at least $k$ (possibly $k+1$ or $k+2$) independent $(\sigma', \pi')$-arcs
    $(a, a')$ in $\DMATCHING$ with
    \begin{equation}
        2N_1 + N_2 + N_3 + 1 < a < 2N_1 + N_2 + N_3 + k + 2
    \end{equation}
    and
    \begin{equation}
        2N_1 + N_2 + N_3 + (k + 2) + 1 < a' < 2N_1 + N_2 + N_3 + (k + 2) + n + 2.
    \end{equation}
    Therefore, by our hypothesis about $\DMATCHING$,
    $\sigma$ occurs as a pattern in~$\pi$.
  \end{proof}


  \section{Conclusion and perspectives}
  \label{section:Conclusion}

  There are a number of further directions of investigation in this
  general subject. They cover several areas: algorithmic, combinatorics,
  and algebra. Let us mention several ---not necessarily new--- open
  problems that are, in our opinion, the most interesting. How many
  permutations of $S_{2n}$ are squares? How many $(213,231)$-avoiding
  permutations of $S_{2n}$ are squares? (Equivalently, by
  Proposition~\ref{prop:bijection_binary_to_permutations_squares},
  how many binary strings of length $2n$ are squares; see also Problem~4
  in \cite{Henshall:Rampersad:Shallit:2011})? How hard is the problem of
  deciding whether a $(213,231)$-avoiding permutation is a square
  (Problem~4 in \cite{Henshall:Rampersad:Shallit:2011},
  see also \cite{Buss:Soltys:2014,Rizzi:Vialette:CSR:2013})?
  Given two permutations $\pi$ and~$\sigma$, how hard is the problem of
  deciding whether $\sigma$ is a square root of~$\pi$?
  As for algebra, one can ask for a complete algebraic study of
  $\QQ[S]$ as a graded associative algebra for the shuffle  product
  $\SHUFFLE$. Describing a generating family for $\QQ[S]$, defining
  multiplicative bases of $\QQ[S]$, and determining whether $\QQ[S]$ is
  free as an associative algebra are worthwhile questions.


  \bibliographystyle{alpha}
  \bibliography{Bibliography}

\begin{thebibliography}{DHT02}

\bibitem[All00]{Allauzen:IGM:2000}
C.~Allauzen.
\newblock {Calcul efficace du shuffle de $k$ mots}.
\newblock Technical report, Institut Gaspard Monge, Universit{\'e}
  Marne-la-Vall{\'e}e, 2000.

\bibitem[BBL98]{Bose:Buss:Lubiw:1998}
P.~Bose, J.~F. Buss, and A.~Lubiw.
\newblock {Pattern Matching for Permutations}.
\newblock {\em Inform. Process. Lett.}, 65(5):277--283, 1998.

\bibitem[BS14]{Buss:Soltys:2014}
S.~Buss and M.~Soltys.
\newblock Unshuffling a square is {NP}-hard.
\newblock {\em J. Comput. Syst. Sci.}, 80(4):766--776, 2014.

\bibitem[CK97]{ChoffrutKarhumaki1997}
C.~Choffrut and J.~Karhum{\"a}ki.
\newblock {\em {Combinatorics of Words, \textrm{in} G. Rozenberg and A. Salomaa
  (eds), Handbook of Formal Languages}}.
\newblock Springer-Verlag, 1997.

\bibitem[DHT02]{DHT:IJAC:2002}
G.~Duchamp, F.~Hivert, and J.-Y. Thibon.
\newblock Noncommutative symmetric functions. {VI}. {F}ree quasi-symmetric
  functions and related algebras.
\newblock {\em Int. J. Algebr. Comput.}, 12(5):671--717, 2002.

\bibitem[EML53]{Eilenberg:MacLane:1953}
S.~Eilenberg and S.~Mac~Lane.
\newblock On the groups of {$H(\Pi,n)$}. {I}.
\newblock {\em Ann. of Math. (2)}, 58:55--106, 1953.

\bibitem[GR14]{Grinberg:Reiner:2014}
D.~Grinberg and V.~Reiner.
\newblock {Hopf Algebras in Combinatorics}.
\newblock {\em \tt \url{arXiv:1409.8356 [math.CO]}}, 2014.

\bibitem[GV16]{GV16}
S.~Giraudo and S.~Vialette.
\newblock {Unshuffling Permutations}.
\newblock {\em Latin American Theoretical Informatics Symposium},
  9644:509--521, 2016.

\bibitem[HRS12]{Henshall:Rampersad:Shallit:2011}
D.~Henshall, N.~Rampersad, and J.~Shallit.
\newblock {Shuffling and Unshuffling}.
\newblock {\em Bulletin of the EATCS}, 107:131--142, 2012.

\bibitem[JR79]{Joni:Rota:1979}
S.~A. Joni and G.-C. Rota.
\newblock Coalgebras and bialgebras in combinatorics.
\newblock {\em Stud. Appl. Math.}, 61(2):93--139, 1979.

\bibitem[Man83]{Mansfield:DAM:1983}
A.~Mansfield.
\newblock {On the computational complexity of a merge recognition problem}.
\newblock {\em Discrete Appl. Math.}, 5:119--122, 1983.

\bibitem[RV13]{Rizzi:Vialette:CSR:2013}
R.~Rizzi and S.~Vialette.
\newblock {On Recognizing Words That Are Squares for the Shuffle Product}.
\newblock In A.A Bulatov and A.M. Shur, editors, {\em 8th International
  Computer Science Symposium in Russia, CSR 2013, Ekaterinburg, Russia}, pages
  235--245, 2013.

\bibitem[Slo]{Slo}
N.~J.~A. Sloane.
\newblock {The On-Line Encyclopedia of Integer Sequences}.
\newblock \url{https://oeis.org/}.

\bibitem[Spe86]{Spehner:TCS:1986}
J.-C. Spehner.
\newblock {Le calcul rapide des melanges de deux mots}.
\newblock {\em Theor. Comput. Sci.}, 47:181--203, 1986.

\bibitem[SS85]{Simion:Schmidt:EJC:1985}
R.~Simion and F.~W. Schmidt.
\newblock {Restricted permutations}.
\newblock {\em Eur. J. Combin.}, 6(4):383--406, 1985.

\bibitem[Var14]{Vargas:2014}
Y.~Vargas.
\newblock {Hopf algebra of permutation pattern functions}.
\newblock {\em 26th International Conference on Formal Power Series and
  Algebraic Combinatorics}, pages 839--850, 2014.

\bibitem[vLN82]{Leeuwen:Nivat:IPL:1982}
J.~van Leeuwen and M.~Nivat.
\newblock {Efficient Recognition of Rational Relations}.
\newblock {\em Inform. Process. Lett.}, 14(1):34--38, 1982.

\bibitem[WH84]{Warmuth:Haussler:JCSS:1984}
M.~K. Warmuth and D.~Haussler.
\newblock {On the complexity of iterated shuffle}.
\newblock {\em J. Comput. Syst. Sci.}, 28(3):345--358, 1984.

\end{thebibliography}


\end{document}